\documentclass{style/vldb}
\usepackage{amssymb}
\usepackage{amsmath}
\usepackage{color}
\usepackage{times}
\usepackage{cite}
\usepackage{enumitem}
\usepackage{arydshln}
\usepackage{url}
\usepackage{algorithm}
\usepackage{algorithmic}
\usepackage{flushend}

\title{Rank-Based Inference Over Web Databases}
\numberofauthors{1}
\author{Md Farhadur Rahman$^\ddag$, Weimo Liu$^\dag$, Saravanan Thirumuruganathan$^\ddag$, Nan Zhang$^\dag$, Gautam Das$^\ddag$\\
\affaddr{The George Washington University$^\dag$, University of Texas at Arlington$^\ddag$}\\
}

\makeatletter
\def\@copyrightspace{\relax} 
\def\@maketitle{\newpage
 \null
 \setbox\@acmtitlebox\vbox{%
\baselineskip 20pt
\vskip 2em                   
   \begin{center}
    {\ttlfnt \@title\par}       
    \vskip 1.5em                
{\subttlfnt \the\subtitletext\par}\vskip 1.25em
    {\baselineskip 16pt\aufnt   
     \lineskip .5em             
     \begin{tabular}[t]{c}\@author
     \end{tabular}\par}
    \vskip 1.5em               
   \end{center}}
 \dimen0=\ht\@acmtitlebox
 \unvbox\@acmtitlebox
 \ifdim\dimen0<0.0pt\relax\vskip-\dimen0\fi}
\makeatother

\begin{document}
\maketitle

\begin{abstract}
In recent years, there has been much research in the adoption of Ranked Retrieval model (in addition to the Boolean retrieval model) in structured databases, especially those in a client-server environment (e.g., web databases). With this model, a search query returns top-$k$ tuples according to not just exact matches of selection conditions, but a suitable ranking function. While much research has gone into the design of ranking functions and the efficient processing of top-$k$ queries, this paper studies a novel problem on the {\em privacy implications} of database ranking.

The motivation is a novel yet serious privacy leakage we found on real-world web databases which is caused by the ranking function design. Many such databases feature private attributes - e.g., a social network allows users to specify certain attributes as only visible to him/herself, but not to others. While these websites generally respect the privacy settings by not directly displaying private attribute values in search query answers, many of them nevertheless take into account such private attributes in the ranking function design. The conventional belief might be that tuple ranks alone are not enough to reveal the private attribute values. Our investigation, however, shows that this is not the case in reality.

To address the problem, we introduce a taxonomy of the problem space with two dimensions, (1) the type of query interface and (2) the capability of adversaries. For each subspace, we develop a novel technique which either guarantees the successful inference of private attributes, or does so for a significant portion of real-world tuples. We demonstrate the effectiveness and efficiency of our techniques through theoretical analysis, extensive experiments over real-world datasets, as well as successful online attacks over websites with tens to hundreds of millions of users - e.g., Amazon Goodreads and Renren.com.
\end{abstract}

\section{Introduction} \label{sec:intro}

\subsection{Motivation}
While traditional structured databases generally support the Boolean Retrieval model (i.e., return all tuples that exactly match the search query selection condition), in recent years there has been much research into exploring the applicability of an alternate Ranked Retrieval model (e.g., a $k$NN interface that returns top-$k$  tuples according to a suitable ranking function).
The ranked retrieval model has become an important component of many databases, especially in a client-server environment (e.g., web databases, where a client specifies and sends queries via a web interface to a backend database).
Prior research has primarily focused on the effective design of ranking functions and the efficient processing of top-$k$ queries for a given ranking function 
(e.g., \cite{chaudhuri2003automated,ilyas2008survey,chaudhuri2004probabilistic}).

However, in this paper we investigate a novel problem on the {\em privacy implications} of database ranking, which has not been studied before. We show how privacy leakage (through the top-$k$ interface) can be caused by a seemingly innocent design of the ranking function in such ranked retrieval models.

To understand how the privacy leakage occurs, note that many databases in a client-server environment feature both public and {\em private} attributes. For example, social networking websites often allow users to specify privacy settings that hide certain attributes from the public's view, e.g., profile demographics such as race, gender, income; location; past posts, etc.  These websites honor the privacy settings by omitting the private attributes from being displayed in the returned query answers. Thus, the results include a ranked list of $k$ tuples, but with only the public attributes displayed, and the private attributes hidden.

The problem here, however, is that many websites indeed include these private attributes as {\em input} to the ranking function. The purpose of doing so is, understandably, to make ranking more effective - 
e.g., the friend-search feature in a social network would preferably return users that have similar demographics or behavior patterns (e.g., posting with similar frequencies) as the user who executes the search,  as common-sense indicates
that they are more likely to be interested in each other.
From the privacy perspective, this design might look harmless as well - after all, while a ranking function might take as input a large number of attributes, its output is merely the (relative) rank of a tuple among returned results - not even the actual ranking score! Naturally, the traditional belief here is that it is impossible to infer private attribute values from just the ranking of a returned tuple.

In our investigation of real-world client-server databases (including popular web databases), we found this traditional belief to be {\em wrong}. Specifically, in this paper, we develop a novel technique which, by asking a carefully constructed sequence of top-$k$ queries and observing the corresponding change of tuple ranks in the query answers, may successfully infer the value of private attributes.

Before introducing our technical results, we would like to first illustrate the real-world impact of this privacy leakage by briefly demonstrating a very simple attack one can deploy using this technique on Renren.com, the equivalent of Facebook in China which has hundreds of millions of users. We chose this website as an example not only because of its large user base, but because it supports extensive privacy settings - allowing a user to specify as private any subset of profile attributes such as hometown, work affiliation, university attended, etc. It also respects these privacy settings in the display of search results - e.g., if a user specifies hometown as private and ``only visible to friends'', then the user's hometown information will be hidden from all search and/or recommendation results unless the query is issued by a friend of the user.

Nevertheless, we also found that when ranking users in search or recommendation results, the ranking function used by Renren.com takes into account {\em all} attributes of a user's profile, regardless of whether a user has specified it to be private and/or who is issuing the query. For example, Figure~\ref{fig:renrenBeforeAndAfter}a shows the screenshot\footnote{Note that, since Renren.com does not have an English version, this screenshot is taken in Chrome with the automated webpage translation feature of Google Translations enabled.} of the ranked list of tuples (i.e., users) returned for a friend-search query issued by a user \textsc{Lionel} with hometown = Beijing, China and no other profile attribute specified. The query is formed using the only public attribute of our victim user \textsc{Target} (with a red target icon in the screenshot), name = Jia Ming. 
Since \textsc{Target} sets his hometown to be a private attribute ``only visible to friends'' and \textsc{Lionel} is not a friend of \textsc{Target}, the hometown of \textsc{Target} is hidden from display in the query answer.
Figure~\ref{fig:renrenBeforeAndAfter}b shows the answer to the exact same query after \textsc{Lionel} changes his hometown to Shanghai, China. The rank of \textsc{Target} now moves up from No.~3 to No.~1 in the new answer - and indeed ranks even higher than a few other users with the same name from Shanghai and studying in Fudan university (in Shanghai). The change of rank indicates a strong likelihood of \textsc{Target} having hometown = Shanghai (even though it does not form a proof). In this paper, we shall show how one can indeed prove that \textsc{Target} must come from Shanghai using just a few other query answers.


\begin{figure}[ht]
\centering
\includegraphics[scale=0.4]{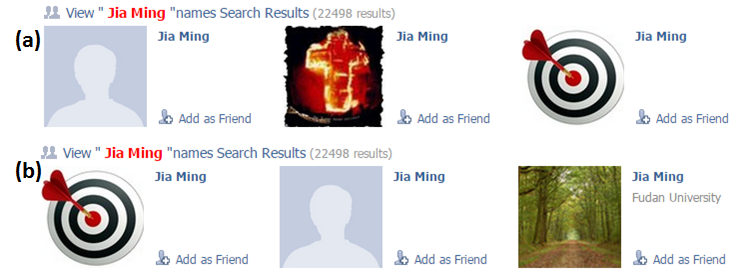}
\caption{Demonstration of an attack over Renren}
\label{fig:renrenBeforeAndAfter}
\end{figure}


\subsection{Novel Problem: Rank-Based Inference}
The above motivating example led us to identify an important and novel problem of {\em ranked-based inference of private attributes}.  From a conceptual standpoint, this problem is interesting as, to the best of our knowledge, privacy compromise from {\em tuple ranks} has not been studied before. From a practical standpoint, this problem is important as many client-server databases, especially web databases that attract large amounts of user contributions, commonly offer top-$k$ query interfaces yet contain sensitive data (e.g., profiles, demographics) that users would like to keep private. 

We formalize the problem as follows.
Consider a database $D$ with $n$ tuples and $m + m^\prime$ attributes, $m$ of which $A_1, \ldots, A_m$ are public while the other $m^\prime$, $B_1, \ldots, B_m$, are private.
The database allows top-$k$ queries where $k$ is a small number ($k \ll n$). To specify a query $q$, one assigns a predicate on each of the $m + m^\prime$ attributes. The predicate can be point\footnote{For the purpose of this paper, we consider all attributes to be discrete, which can be categorical or ordinal, and assume the proper discretization of numeric attributes.} (i.e., $A_i = v$) or range (e.g., $B_i \in \{v_1, v_2\}$ or $*$, i.e., the entire domain).

Given a top-$k$ query $q$, the database computes a predetermined ranking function $s(t|q)$ for each tuple $t$ in the database, and returns the $k$ tuples with the smallest $s(t|q)$. Of course, only the $m$ public attributes are displayed on the return interface - not the private attributes or the ranking score. In most websites, the ranking function is a closely guarded secret - so we assume the adversary has no knowledge of the ranking function other than two very basic properties, {\em monotonicity} and {\em additivity}, which we shall define in \S\ref{sec:pre} and demonstrate that they hold for almost all reasonable ranking functions used in the real-world.

The objective of an adversary is to compromise the privacy of a pre-determined victim tuple $v$. Of course, the adversary can readily acquire the public attributes of $v$. Nonetheless, it does not know the ranking function being used (and of course no knowledge whatsoever of the tuples' ranking scores). Thus, the technical challenge for the adversary is to unveil the private attribute values, e.g., $v[B_1]$, by issuing a small number of queries through the web interface and observing only the public attribute values of the returned tuples and the order in which they are returned. Note that an important goal for the adversary is to keep the number of queries as small as possible, because almost all websites enforce a limit on the number of queries one can issue through the web interface for a given time period (e.g., from one IP address or one user account each day), in order to prevent overburdening its backend database or to thwart third-party crawling of its contents.


To the best of our knowledge, the above problem of inferencing sensitive data from the {\em ranking} of a tuple is very novel. While top-$k$ querying has been extensively studied by the database community~\cite{ilyas2008survey,fagin2003optimal,chaudhuri2004probabilistic}, much of the efforts were focused on (1) developing techniques to answer such queries efficiently~\cite{ilyas2008survey,hristidis2001prefer,li2009unified,kiessling2002foundations}, and (2) designing proper distance/ranking functions for various applications~\cite{bruno2002top,motro1988vague,rui1997content}.  There have been prior work on data privacy in the general area of {\em query inferencing}\cite{farkas2002inference,domingo2008survey,chin1986security}, but most focus was on learning individual values from aggregates such as SUM, MIN, MAX, etc. 
We discuss related work in more detail in \S\ref{sec:relWork}.

\subsection{Overview of Technical Results}
As one of our important contributions, we introduce a comprehensive taxonomy of the problem space according to two dimensions: (1) the type of query interfaces widely used in practice and (2) the capability of adversaries. Then, for each subspace of the problem, we develop a novel technique which either guarantees the successful inference of private attributes, or (when such an inference is provably infeasible in the worst-case scenario) accomplishes the attack for a significant portion of real-world tuples.

Consider the first dimension. We distinguish between interfaces which only support ``point queries'' (i.e., a single value must be specified for each attribute in the query), and those that also support ``IN queries'' (i.e., where a subset/range of values can be specified for an attribute). For the second dimension, we distinguish between two types of adversaries: (1) those who are ``query-only'' (Q-only adversaries) - i.e., they are {\em passive} adversaries who only issue queries and observe query answers, but never tamper with (e.g., insert fake tuples into) the database; and (2) adversaries who ``query-and-insert'' (Q\&I-adversaries), i.e., they only issue queries but also insert fake tuples into the database (e.g., by registering for fictitious user accounts on a social media website). As we shall further elaborate in \S\ref{sec:spa}, while many web databases have no restriction on the registration of new accounts, others makes it difficult for users to create fictitious accounts - e.g., Catch22Dating~\cite{catch22Rules}, a vulnerable website we shall study in the experiments, manually authenticates the real-world identity of each new account; while the aforementioned Renren.com also manually checks and verifies all user name changes. In these cases, most adversaries would be Q-only, while those who have adequate resources to acquire multiple real-world identities can become Q\&I.

We have carefully investigated the four problem subspaces arising out of this taxonomy, and developed four novel attacks: Q\&I-Point, Q-Point, Q\&I-IN,  and Q-IN. The fundamental ideas behind these attacks include two critical reductions: One reduces the problem of compromising a private attribute to finding so-called {\em differential queries} (defined in \S\ref{sec:eqp}) which exclude all but one values in the domain. The second further reduces the problem to just finding a query which returns the victim tuple - nevertheless, this reduction holds only for Q\&I-adversaries.

\begin{table}[h]
\vspace{-3mm}
\caption{Feasibility, Worst- and Practical Query Cost}
\begin{tabular}{lcccc}
  \hline
   & Q\&I-Point & Q-Point & Q\&I-IN & Q-IN\\
  \hline
  Feasibility & Yes & Maybe & Yes & Maybe\\
  Worst-case & $\prod^{m^\prime}_{i=1}|V^\mathrm{B}_i|$ & N/A & $\prod^{m^\prime}_{i=1}|V^\mathrm{B}_i|$ & N/A\\
  In Practice & High & Highest & Lowest & Low\\
  \hline
\end{tabular}
\small{Note: $|V^\mathrm{B}_i|$ is the domain size for private attribute $B_i$.}
\end{table}

The differences on the applicability of these reductions lead to fundamentally different feasibilities of the attack, as illustrated in Table~1. Specifically, we find that while Q\&I adversaries are always able to accomplish the attack, there are cases where Q-only ones will fail. In terms of query cost, while the worst-case cost for even Q\&I adversaries can be exponential, the query cost in practice is very reasonable - and can be significantly reduced when IN queries are available, even though IN has no impact on the (theoretical) worst-case query cost.

In summary, we make the following contributions in this paper:

\begin{itemize}[noitemsep,topsep=1pt,parsep=1pt,partopsep=0pt]
\item We have identified a novel and important problem of rank-based inferencing over web databases.

\item We introduce a comprehensive taxonomy of the problem space, and identify four important subspaces based on varying database  interface limitations and adversarial capabilities.

\item For each problem subspace, we developed nontrivial adversaries, and carried out a rigorous theoretical analysis of their performance. Our results show that in almost all cases, the adversaries can launch efficient and
successful attacks.

\item We performed extensive experiments over real-world datasets, with results corroborating well with our theoretical findings. We also conducted successful online experiments over real-world websites including the aforementioned social network Renren.com as well as other types of web databases such as Amazon Goodreads and Catch22Dating.
\end{itemize}

\section{Preliminaries} \label{sec:pre}


\subsection{Model of Web Databases}

As discussed in the introduction, many web databases store both public and private attributes of a user.  Consider an $n$-tuple (i.e., $n$-user) database $D$ with a total of $m + m^\prime$ attributes, including $m$ public ones $A_1, \ldots, A_m$ and $m^\prime$ private ones $B_1, \ldots, B_{m^\prime}$.  Let $V^\mathrm{A}_i$ and $V^\mathrm{B}_j$ be the attribute domain (i.e., set of all attribute values) for $A_i$ and and $B_j$, respectively. For the purpose of this paper, we consider $V^\mathrm{A}_i$ and $V^\mathrm{B}_j$ to be discrete and publicly known, and leave studies of numeric/infinite/unknown domains to future research.

We use $t[A_i]$ (resp.~$t[B_j]$) to denote the value of a tuple $t \in D$ on attributes $A_i$ (resp.~$B_j$). For the purpose of this paper, we assume there is no duplicate tuple in the database (before an adversary makes any modification to the database) - i.e., every {\em bona fide} tuple has a unique value combination for the $m + m^\prime$ attributes. While we assume that $D$ does not change during the course of an attack, we include discussions in \S\ref{sec:rfo} to address the scenario where this assumption is violated.

Recall from the introduction that the database allows top-$k$ queries where $k$ is a small number such as 10 or 50. Given a {\em supported query} $q$ defined below, the database computes the {\em ranking function} $s(t | q)$ for each tuple $t \in D$, and selects/returns $k$ tuples in the {\em ascending} order of $s(t | q)$ (i.e., only the $k$ tuples with minimum $s(t|q)$ will be returned). Of course, only the public attribute values, i.e., $t[A_1], \ldots, t[A_m]$, will be returned for each of the $k$ tuples. Of course, since we allow duplicates on public attribute values - i.e., multiple tuples might share the same vaule combination on $A_1, \ldots, A_m$ - there must be a way to distinguish different returned tuples with the same public-attribute value-combination. For this purpose, we assume each tuple to be returned alongside a unique identifier (e.g., user ID) - and the adversary knows the unique identifier of the victim tuple as prior knowledge. It is important to note that the ranking score $s(t|q)$ is {\em not} returned - in addition, the design of $s(\cdot|\cdot)$ itself is a secret kept by the database owner.

\vspace{1mm}
\noindent{\bf Supported Queries:} For the purpose of this paper, we consider ranking functions/queries that take into account both public and private attribute information. In other words, the web database supports queries which specify values/conditions on some or all of the $m + m^\prime$ (public and private) attributes. Consider friend recommendation on a social media website as an example - when the website uses private information of a user (say \texttt{education}) while generating the recommendations, it is indeed answering a query that contains a predicate on private attribute \texttt{education} - with the ranking function likely taking into account whether a tuple's value on \texttt{education} is equal to that specified in the query.

In this paper, we consider two types of predicates that can be specified on an attribute: {\em point} and {\em IN}. Let the predicate specified in a query $q$ for attribute $A_i$ (resp.~$B_i$) be $q[A_i]$ (resp.~$q[B_i]$). A point predicate assigns a single value in the domain, i.e., $q[A_i] \in V^\mathrm{A}_i$, while an IN predicate assigns a subset of values, i.e., $q[A_i] \subseteq V^\mathrm{A}_i$. Consider a dating website as an example. While gender is often specified as a point predicate (i.e., male or female), interests and age can be considered IN ones (i.e., find users who most closely match the interest set \{reading, travel, cycling, cooking\} or age range $[25, 30]$). A special example of IN predicate is $q[A_i] = V^\mathrm{A}_i$ - i.e., $q[A_i] = *$ - indicating ``do-not-care'' on an attribute.

\vspace{1mm}
\noindent{\bf Practical Constraints:} Most, if not all, web databases enforce practical constraints on how one might interact with the web interface. The two most important constraints here are {\em query-rate limitation} and {\em tuple insertion constraint}.

Most web databases enforce certain query-rate limits, i.e., limits on the number of queries one can issue (e.g., from an IP address or a user account) per time period (e.g., each day), in order to prevent overburdening of the backend database and/or third-party crawling of its contents. Hence an adversary must aim to minimize the query cost of a rank-based inference attack, as otherwise it would have to acquire more resources (e.g., more IP addresses, registering more accounts) in order to issue all queries required by the attack.

Tuple insertion constraint, on the other hand, refers to ones ability to {\em insert} tuples into the database. Some web databases, including many online social networks, do not enforce this constraint - i.e., one can freely insert new tuples (i.e., user accounts) to the database by registering for new accounts (e.g., using a new email address). Nonetheless, there are also others that require users' real identities and use offline authentication to check them. For example, catch22dating, a popular online dating website used in our real-world experiments, requires each user to have an authenticated identity as student of selected universities. For these databases, inserting new/fake tuples becomes extremely difficult, if not impossible. We say that the web database enforces a tuple insertion constraint which prevents an adversary from inserting arbitrary tuples.


\subsection{Properties of Ranking Function}

There has been significant research in database ranking (e.g., \cite{ilyas2008survey,geng2008query,li2009unified}) which studies the design of ranking function $s(t|q)$, including in cases where the query has IN predicates (e.g., \cite{ho1997range,ilyas2008survey}). While this paper aims to study {\em generic} rank-based inferences that work for a broad class of ranking functions, it is important to note that {\em no} attack will work without assuming certain properties of the ranking function. To understand why, consider a simple example where $s(t|q)$ is generated uniformly at random from $[1, n]$. Since the rank of a tuple has nothing to do with the tuple's (private) attribute values, no adversary can compromise any private information from the returned ranks. Thus, it is the objective of this subsection to define a minimum set of conditions that are satisfied by most if not all ranking functions used in practice. Specifically, we consider {\em monotonicity} and {\em additivity}, respectively as follows.

\vspace{1mm}
\noindent{\bf Monotonicity Condition:} Intuitively, the monotonicity condition simply states that, for a given query, the relative rank between two tuples which differ only on one attribute should be determined by that attribute alone. Formally, for a point-query interface, if two tuples $t$ and $t^\prime$ differ only on $A_i$ and $t[A_i] = q[A_i]$, then $t$ should have a smaller distance to $q$ than $t^\prime$. More generally, we have the following definition. Note that in this definition, we consider $q[A_i]$ (resp.~$q[B_j]$) to be a set (in the case of point-query, containing just a single value) without introducing ambiguity.

\vspace{1mm}\hrule\vspace{.5mm}
\noindent{\em Monotonicity:} $\forall q$, $t \in D$, and $i \in [1, m]$ (resp.~$j \in [1, m^\prime]$), if $t$ and $t^\prime$ share the same value on all attributes except $A_i$ (resp.~$B_j$) and $t[A_i] \in q[A_i]$ while $t^\prime[A_i] \not \in q[A_i]$ (resp.~$t[B_j] \in q[B_j]$, $t^\prime[B_j] \not \in q[B_j]$), there must be $s(t|q) < s(t^\prime|q)$.
\vspace{1mm}\hrule\vspace{.5mm}

\vspace{2mm}
\noindent{\bf Additivity Condition:} Intuitively, the additivity condition states that, for two tuples $t$ and $t^\prime$, if $t$ is already ranked higher than $t^\prime$ in query $q$, then further changing the predicate of $q$ on $A_i$ (resp.~$B_j$) to exactly match $t$ - i.e., making $q[A_i] = t[A_i]$ (resp.~$q[B_j] = t[B_j]$) - should not change the relative rank between the two tuples. More formally, we have the following definition:

\vspace{1mm}\hrule\vspace{.5mm}
\noindent{\em Additivity:} $\forall q$ and $t, t^\prime \in D$, if $s(t|q) < s(t^\prime|q)$, then there must be $s(t|q^\prime) < s(t^\prime|q^\prime)$, where $q^\prime$ is the same as $q$ on all but one attribute $A_i$ (resp.~$B_j$), on which $q[A_i] = t[A_i]$ (resp.~$q[B_j] = t[B_j]$).
\vspace{1mm}\hrule\vspace{.5mm}

\vspace{1mm}
One can see that both monotonicity and additivity are common-sense conditions that should be reasonably expected of a ranking function. Our studies of real-world web databases (in \S\ref{sec:exp}) verified this observation, as all websites considered satisfy both conditions.

\section{Problem Space} \label{sec:spa}

In this section, we define the rank-based inference problem studied in the paper. Specifically, we start with defining the objectives of an adversary. Then, we partition the entire problem space into four quadrants along two dimensions: the type of queries supported, and the type of operations an adversary can perform. 

\subsection{Adversary Model} \label{sec:adm}

The objective of an adversary is two-fold: {\em compromising privacy} and {\em minimizing query cost}. Privacy-wise, an adversary aims to compromise private attributes of a victim tuple $v$. Without loss of generality, we assume that the adversary aims to compromise the value of $v[B_1]$ based on prior knowledge of all public attributes of $v$, i.e., $v[A_1], \ldots, v[A_m]$. In \S\ref{sec:pq}, we shall address cases where an adversary aims to compromise all private attributes of $v$.

To ensure the versatility of our algorithms, we make a conservative assumption that an adversary has {\em no} prior knowledge of the ranking function other than the fact that it satisfies the monotonicity and additivity conditions defined above. Clearly, all algorithms in the paper still work if an adversary does know the ranking function. While it is possible that prior knowledge of certain ranking functions can enable more efficient attacks than those in the paper, we leave such ranking-function-specific studies to future work.

Given the query-rate limitation discussed in \S\ref{sec:pre}, an important goal of the adversary is to minimize the query cost for the attack, as otherwise the website-enforced limit on the number of queries from each user (e.g., IP-address) may stop the attack from being completed. To this end, it is important to note that our key efficiency measure here is the number of requests issued to the web database (hereafter referring to as the {\em query cost}, including both search queries and requests to insert tuples, if the database does not enforce the aforementioned tuple insertion constraint) - while other measures such as local (CPU or I/O) processing overhead are secondary.

\subsection{Two Dimensions}

The first dimension we use for partitioning the problem space is the type of queries supported. There are two different cases: (1) Point-Query Interface which requires a point predicate defined in \S\ref{sec:pre} to be specified for {\em every} attribute. An example here is the friend recommendation offered by many social media websites - each user has to complete his/her own profile to enable the feature, essentially requiring the user to specify a point predicate on every public and private attribute. (2) IN-Query Interface which supports {\em IN queries} over all attributes. Clearly, here a user can choose ``do not care'' for an attribute by assigning its entire value domain to the IN condition. Since point queries are special cases of IN, all queries supported by the point-query interface are also supported here.

The second dimension for partitioning the problem space is the adversary power. Specifically, we consider the following two cases:
\begin{itemize}[noitemsep,topsep=1pt,parsep=1pt,partopsep=0pt]
\item {\em Query-and-Insert (Q\&I) Adversary} can not only issue queries but also {\em insert} tuples to the database. It can also update or delete any tuple it inserted. These adversaries exist for websites which do not enforce the tuple insertion constraint.
\item {\em Query-only (Q-only) Adversary} can query the web database but cannot change it. This is the case when the website enforces the tuple insertion constraint (see \S\ref{sec:pre}).
\end{itemize}
One can see from the definitions that Q\&I adversaries are stronger - i.e., any attack launched by a Q-only adversary can also be launched by a Q\&I-one, while the opposite is not true. We shall show later in the paper that the ability to {\em insert} leads to significant differences on the outcome of a rank-based inference attack. Specifically, while a Q\&I adversary can {\em always} accomplish the attack even in the worst-case scenario, the same is not true for Q-only adversaries. 

\subsection{Problem Definition}

Given the two dimensions, we partition the problem space into four quadrants: (1) point query interface with Q\&I adversaries, (2) point query with Q-only, (3) IN with Q\&I, and (4) IN with Q-only.

\vspace{1mm}
\noindent{\em Problem Definition (Rank-Based Inference): Given a database $D$ and a victim tuple $v \in D$, find the shortest sequence of queries $q_1, \ldots, q_c$ supported by the interface and a corresponding sequence of tuple sets $T_1, \ldots, T_c$, such that
\begin{align}
\delta(q_1(D \cup T_1), q_2(D \cup T_2), \ldots, q_c(D \cup T_c)) = v[B_1].
\end{align}
where $q_i(D \cup T_i)$ is the answer to $q_i$ over the $D \cup T_i$ and $\delta(\cdot)$ is a (deterministic) function for rank-based inferencing. For a Q-only adversary, there must be $T_1 = \cdots = T_c = \emptyset$.}

Naturally, the problem could be extended to infer multiple, if not all, private attributes of victim tuple $v$. In fact, as we shall describe in \S\ref{sec:pq}, our algorithm iteratively learns private attribute values one at a time till $v[B_1]$ is inferred. Extending it to infer all attributes is trivial. To better illustrate our ideas and to significantly simplify the notations, in the theoretical discussions in this paper, we focus on the case where $k = 1$ (note that $k=1$ is actually a conservative worst-case assumption for the attack design), and discuss the straightforward extension to larger $k$ in the experiments section.

\vspace{1mm}
\noindent{\bf Running example of ranking function:} All algorithms developed in this paper work for any ranking function satisfying monotonicity and additivity - so does all complexity and lower bound analysis. Nonetheless, when studying the practical performance of attacks and illustrating how different ranking-function designs affect attack effectiveness, it is necessary to consider certain concrete ranking function designs - for this purpose only, we consider the following linear ranking function as a running example:

\vspace{-2mm}
\begin{small}
\begin{align}
s(t|q) = \sum^m_{i=1} w^\mathrm{A}_i \cdot \rho(q[A_i], t[A_i]) + \sum^{m^\prime}_{i=1} w^\mathrm{B}_i \cdot \rho(q[B_i], t[B_i]), \label{equ:scr}
\end{align}
\end{small}

\noindent where $w^\mathrm{A}_i, w^\mathrm{B}_i \in (0, 1]$ are the {\em ranking weight} for attribute $A_i$ and $B_i$, respectively. The distance measure for each attribute, i.e., $\rho(q[A_i], t[A_i])$, is a variation of the discrete metric: (1) $\rho(q[A_i]$, $t[A_i]) = 0$ if $t[A_i] \in q[A_i]$ (note that for point queries, this means $t[A_i]$ being equal to the single value in $q[A_i]$), and (2) $\rho(q[A_i]$, $t[A_i]) = 1$ if $t[A_i] \not \in q[A_i]$.

Once again, we would like to note that the adversary has {\em no} knowledge of the ranking function design whatsoever (other than its monotonicity and additivity). This linear ranking function based running example merely provides a concrete basis for the analysis of attack performance in practice.

\section{Point Query Interface} \label{sec:pq}

We start by considering a point query interface. Specifically, we shall start with reducing rank-based inference to the problem of finding pairs of {\em differential queries} based on the victim tuple $v$. Then, we discuss the design of Q\&I-Point and Q-Point, our rank-based inference algorithms for Q\&I and Q-only adversaries over a point query interface, respectively.

\subsection{Goal: Finding Differential Queries} \label{sec:eqp}

We start by showing that, for the worst-case scenario of $k = 1$, the problem of compromising the private attribute $B_1$ of victim tuple $v$ can be reduced to finding for each possible value of $B_1$ except $v[B_1]$, i.e., $\forall \theta \in (V^\mathrm{B}_1 \backslash v[B_1])$, a pair of {\em differential queries} $q_\theta$ and $q^\prime_\theta$ which satisfy three properties: (1) they share the same predicate on all attributes but $B_1$, (2) $q^\prime_\theta[B_1] = \theta$ while $q_\theta[B_1] \neq \theta$, and (3) $q_\theta$ returns the victim tuple $v$ while $q^\prime_\theta$ does not - i.e.,
\begin{multline}
\label{eq:q1q2Constraints}
\qquad \qquad \forall t \in D \text{ where } t \neq v, s(t|q_\theta) > s(v|q_\theta), \\
\exists t \in D \text{ with  } t \neq v \text{ such that } s(t|q^\prime_\theta) < s(v|q^\prime_\theta).
\end{multline}

The proof of this reduction is straightforward: Due to the additivity condition, we can infer from (\ref{eq:q1q2Constraints}) that the value of victim tuple $v$ on $B_1$ must {\em not} be the same as $q^\prime_\theta[B_1]$, i.e., $v[B_1] \neq \theta$. Since we found differential queries $q_\theta, q^\prime_\theta$ for all $\theta \in V^\mathrm{B}_1\backslash v[B_1]$, the only remaining possibility is the correct value of $v[B_1]$.

While this reduction is the basis of our following discussions, it is important to note that reduction in the opposite direction does not hold - i.e., an adversary does {\em not} have to find all $|V^\mathrm{B}_1| - 1$ pairs of differential queries in order to compromise $v[B_1]$. To understand why, consider an example where a Q\&I-adversary inserts into the database a dummy tuple $t$ with value 0 on all public and private attributes. Then, upon issuing a query with $q[A_i] = 0$ and $q[B_i] = 0$ on all attributes, the adversary receives $v$ rather than $t$ as the No.~1 result. One can see that the adversary can safely infer $v[B_1] = 0$ without issuing any additional query or identifying the differential queries for any value of $B_1$.

\subsection{Q\&I adversary}

We develop Algorithm Q\&I-Point in this subsection. Specifically, we start with a somewhat surprising finding - for a Q\&I adversary, as long as it has the ability to find a query that returns the victim tuple $v$ for a given database, then it can {\em always} successfully compromise $v[B_1]$ (by finding differential queries for all other values in $V^\mathrm{B}_1$). Then, we present Algorithm Q\&I-Point and analyze its worst- and average-case query costs.

\subsubsection{Reduction to finding a query that returns $v$} \label{sec:rfo}

\vspace{1mm}\noindent{\bf Algorithm Q\&I-Point:} To construct the reduction to finding one query which returns the victim tuple, we start by assuming an oracle FIND-Q which, upon given input of a database D and the victim tuple $v$, returns a query $q$ which returns $v$. We first develop Algorithm Q\&I-Point which calls upon this oracle FIND-Q to compromise $v[B_1]$, and then introduce the design of FIND-Q afterwards. The pseudocode of Q\&I-Point is shown in Algorithm~\ref{alg:qAndIPoint}.


For the ease of understanding, we introduce some simple notations: We represent the domain of every attribute as 0, 1, \ldots, $|V| - 1$, where $|V|$ is the domain size of the attribute. Let there be $v[A_1] = \cdots = v[A_m] = 0$. Without loss of generality, we assume the output of FIND-Q to always have $q[A_1] = \cdots = q[A_m] = 0$. The reason here is simple - if $q$ differs from $v$ on any public attribute, we can always change the attribute to 0 - the new query will still return $v$ due to the additivity condition of the ranking function.

We start by inserting into the database a tuple $t$ that has all attributes equal to 0. Then, we call FIND-Q over the new database to discover $q$ which returns $v$. Note that if FIND-Q fails to do so - i.e., no query over the database returns $v$ - then we already succeed because, due to the no-duplicate assumption, the only scenario for this to happen is when $v = t$. Given the result $q$ of FIND-Q, we note that $q$ must differ from $t$ on at least one private attribute - again, if $q = t$ and yet returns $v$, there must be $v = t$. Without loss of generality, suppose that $q$ differs from $t$ by having value 1 on private attributes $B^\prime_1, \ldots, B^\prime_h$ - i.e., $\forall i \in [1, h]$, $q[B^\prime_i] = 1$.

We now construct $h + 1$ queries $q_0, \ldots, q_h$ as follows: all these $h + 1$ queries share the exact same value (i.e., 0) as $t$ on all attributes but $B^\prime_1, \ldots, B^\prime_h$. For those $h$ attributes, we assign to query $q_i$ ($i \in [0, h]$) $q_i[B^\prime_j] = q[B^\prime_j] = 1$ if $j \leq h - i$ and $q_i[B^\prime_j] = t[B^\prime_j] = 0$ otherwise (i.e., if $j > h - i$). The following table shows an example. Note at the two extremes $q_0 = q$ and $q_h = t$.

\begin{table}[h]
\begin{center}
\begin{tabular}{cccccc}
& $A_1, \ldots, A_m$ & $B^\prime_1$ & $B^\prime_2$ & $B^\prime_3$ & $B_\mathrm{others}$\\
\hline
$q$ & 0 & 1 & 1 & 1 & 0\\
\hdashline
$q_0$ & 0 & 1 & 1 & 1 & 0\\
$q_1$ & 0 & 1 & 1 & 0 & 0\\
$q_2$ & 0 & 1 & 0 & 0 & 0\\
$q_3$ & 0 & 0 & 0 & 0 & 0\\
\hdashline
$t$ & 0 & 0 & 0 & 0 & 0\\
\end{tabular}
\end{center}
\end{table}

There are two important observations from the above query sequence: First, unless $v = t$, queries at the two ends must return different results - specifically, $q_0$ returns $v$ while $q_h$ returns $t$. The only exception here is when $q_h$ also returns $v$ - but this must mean $v = t$ because $q_h$ exactly matches $t$ - leading to an immediate compromise of $v[B_1]$. Second, every pair of adjacent queries in the sequence differ by exactly one attribute - i.e., query $q_i$ and $q_{i+1}$ differ on $B^\prime_{h - i}$. Combining two observations, we know two things: (1) there must exist a pair of adjacent queries $q_i$ and $q_{i+1}$ such that $q_i$ returns $v$ while $q_{i+1}$ does not - because otherwise all $h+1$ queries would return $v$, contradicting Observation 1. (2) this pair of adjacent queries differ on exactly one attribute $B^\prime_{h - i}$. In other words, they serve as a pair of {\em differential queries} for value 0 in the domain of $B^\prime_{h - i}$, and prove $v[B^\prime_{h - i}] \neq 0$. Note that the process of finding this pair of differential queries takes at most $h \leq m^\prime$ queries.

Of course, this may not yet achieve the adversarial goal of compromising $v[B_1]$.  Nonetheless, note that once we know $v[B^\prime_{h - i}] \neq 0$, we can insert into the database a new $t$ which replaces its value on $B^\prime_{h - i}$ with another value (other than 0) in its domain. We can then repeat the exact same process and get one of only two possible outcomes: either (1) we find another pair of differential queries and exclude from consideration a(other) value for one of the private attributes; or (2) an anomaly occurs - either FIND-Q cannot find $q$ or $q^h$ returns $v$ instead of $t$ - meaning $t = v$ and we have compromised $v[B_1]$ already.

One can see that, the worst-case scenario here is for us to repeat the process for $\sum^{m^\prime}_{i=1} (|V^\mathrm{B}_i| - 1)$ times - more repetitions is impossible because we would have already excluded all wrong values for $B_1, \ldots, B_{m^\prime}$. Throughout all repetitions, the number of queries issued by Algorithm Q\&I-Point (excluding those required by FIND-Q) is $O(m^{\prime} \cdot \sum^{m^\prime}_{i=1} |V^\mathrm{B}_i|)$.

\vspace{1mm}\noindent{\bf Practical Implications:} We now discuss the practical implications of Algorithm Q\&I-Point. First, while we shall address the design and theoretical bounds of FIND-Q in detail next, we would like to first point out here that, in practice, FIND-Q is usually a straightforward and efficient process, especially when there are many public attributes. The reason is simple: those public attributes alone are often sufficient to uniquely identify the victim tuple. Since FIND-Q knows $v[A_1], \ldots, v[A_m]$, it largely just needs to avoid hitting the few ``fake'' tuples Algorithm Q\&I-Point inserts (by avoiding their private attribute values) in order to find a query that returns $v$.

The cost of FIND-Q aside, there are three interesting observations we can make regarding Algorithm Q\&I-Point. First, its query cost depends on the SUM (not product) of domain size of private attributes. This works to the attacker's advantage in practice as real-world websites often feature only a few private attributes with small domains\footnote{e.g., in the case of dating website discussed in Section~\ref{subsec:expOnline}, all private attributes are Boolean - e.g.,  ``whether a user is willing to consider matches of a different race.''}. Nonetheless, this also means that large-domain attributes such as ZIP code can be very costly to attack. Intuitively, this is caused by nature of the point query interface - as each query here ``covers'' only one of the many domain values.

The second observation we would like to make is the {\em anytime} nature of the algorithm. While our problem definition focuses on compromising $v[B_1]$, one can see from the design of Q\&I-Point that it indeed learns all private attributes of $v$. Specifically, every iteration (costing at most $m^\prime$ queries) excludes one value from consideration for one of the private attributes. Thus, even if we interrupt the algorithm at anytime (say running out of query allowance by the database), we would still have learned substantial information about many private attributes. This anytime feature makes the algorithm particularly difficult to thwart in practice.

Third, note from the design of Q\&I-Point that all queries it issues (including those by FIND-Q) must have public attribute values equal to those of $v$. This makes the algorithm fairly resilient against changes to the database during the course of an attack - because the only changes that would affect the execution of Q\&I-Point are those that feature tuples with the exact same public-attribute value-combination as $v$ - an extremely unlikely event in practical databases.

Last, while a comprehensive discussion of defense strategies is out of the scope of this paper, one can see from the design of Q\&I-Point a basic idea of thwarting rank-based inference: Note that during the attack, an adversary makes $t$ closer and closer to the victim tuple $v$, in order to ``isolate'' the rank difference to a private attribute and infer its value. Thus, intuitively, the defense should try to detect and/or prevent the insertion of tuples that are too ``close'' to an existing tuple, as such insertions likely signal the Q\&I-Point attack rather than a bona fide new tuple (which is highly unlikely to be almost identical to an existing tuple).

\begin{algorithm}[!htb]
\caption{{\bf Q\&I-Point}}
\begin{algorithmic}[1]
\label{alg:qAndIPoint}
\STATE {\bf Input:} $q, v$ \qquad {\bf Output:} $v[B_1]$
\STATE $H_v = \varnothing$; $t[A_i]=v[A_i] \, \forall i \in [1,m]$; $t[B_j] = 0 \, \forall j \in [1,m^\prime]$
\STATE {\bf while} $v[B_1]$ is not yet inferred {\bf do}
    \STATE \hspace{\algorithmicindent} Insert $t$ into $D$
    \STATE \hspace{\algorithmicindent} {\bf if} $q$ does not return $v$ {\bf then} $q \leftarrow$ FIND-Q$(v, H_v)$
    \STATE \hspace{\algorithmicindent} {\bf if} no such $q$, {\bf then} {\bf return} $t[B_1]$ \qquad // case: $t=v$
    \STATE \hspace{\algorithmicindent} Let $B^{\prime}_{1} \ldots B^{\prime}_{h}$ be the attributes differing between $t$ and $q$
    \STATE \hspace{\algorithmicindent} $i=0$; \qquad $q_i = q$
    \STATE \hspace{\algorithmicindent} {\bf for} {$i=1$ to $h$} {\bf do}
        \STATE \hspace{\algorithmicindent} \hspace{\algorithmicindent} $q_i = q_{i-1}$; \quad $q_i[B^{\prime}_{h-i+1}] = t[B^{\prime}_{h-i+1}]$
        \STATE \hspace{\algorithmicindent} \hspace{\algorithmicindent} {\bf if} { $q_i$ and $q_{i-1}$ return different tuples} {\bf then}
            \STATE \hspace{\algorithmicindent} \hspace{\algorithmicindent} \hspace{\algorithmicindent} 
                     $q = q_{i-1}$; \quad Set $t[B^{\prime}_i]$ to an unexplored value
            \STATE \hspace{\algorithmicindent} \hspace{\algorithmicindent} \hspace{\algorithmicindent} 
                    {\bf break} for loop
\end{algorithmic}
\end{algorithm}

\subsubsection{Query Cost Analysis} \label{sec:qip}

\vspace{1mm}\noindent{\bf Algorithm FIND-Q:} We now describe the algorithm for finding a query $q$ that returns $v$ for a given database $D$. Algorithm~\ref{alg:findQ} shows the pseudocode for FIND-Q. The design is mostly straightforward - we randomly generate and issue a query $q$ with $q[A_i] = v[A_i]$ for all $i \in [1, m]$ and each $q[B_j]$ ($j \in [1, m^\prime]$) drawn i.i.d.~uniformly at random from $V^\mathrm{B}_j$ - and repeat this process until finding $q$ that returns $v$. The only note of caution here is that the random generation is done {\em without replacement}, and {\em with memory} across different executions of FIND-Q. To understand why, note from the design of Algorithm Q\&I-Point that we only insert tuples into the database, and do not tamper with or delete the existing tuple values. Thus, any query which does not return $v$ before cannot return $v$ in the future - justifying the design.

One can see from the design of FIND-Q that it always succeeds. As such, our focus here is to consider its query cost. First, all calls of FIND-Q, altogether, consume a worst-case query cost of $O(\prod^{m^\prime}_{i=1}|V^\mathrm{B}_i|)$. While this seems like an outrageously high cost, we make two interesting notes here: First, the worst-case scenario indeed requires these many queries - as proved by the following lower bound result which shows that the cost cannot be improved beyond a constant factor. Second, the real-world query cost for FIND-Q is likely much smaller, as demonstrated by an average-case example study at the end of this subsection.

\begin{algorithm}[!htb]
\caption{{\bf FIND-Q}}
\begin{algorithmic}[1]
\label{alg:findQ}
\STATE {\bf Input:} Victim $v$, Query history $H_v$ {\bf Output:} $q$ that returns $v$
\STATE Let $\mathcal{Q}$ be set of all possible point queries $q$ with $q[A_i] = v[A_i] \quad \forall i \in [1,m]$
\STATE Find a query $q \in \mathcal{Q} \setminus H_v$ that returns $v$; \quad Update $H_v$ 
\STATE {\bf return} $q$ if it exists else return failure 
\end{algorithmic}
\end{algorithm}

\vspace{1mm}\noindent{\bf Lower Bound on Worst-Case Query Cost:} The following theorem shows that, in the worst-case scenario, no algorithm can accomplish the attack without issuing $\Omega(\prod^{m^\prime}_{i=1}|V^\mathrm{B}_i|)$ queries.
\newtheorem{theorem}{Theorem}
\begin{theorem} \label{thm:wc1}
Given any ranking function and victim tuple $v$, there exists a database $D$ such that no Q\&I-adversary can compromise $v[B_1]$ without issuing $\Omega(\prod^{m^\prime}_{i=1}|V^\mathrm{B}_i|)$ queries.
\end{theorem}


\begin{proof}
First, we prove that, in order for an adversary to compromise $v[B_1]$, it must be able to find at least one query which returns $v$. Let the queries the adversary issues before inferring the value of $v[B_1]$ be $\epsilon_1, \ldots, \epsilon_x$. Note that in order for the inference to hold, the following condition must be true: if we change the value of $v[B_1]$ to $\theta \in V^\mathrm{B}_1\backslash v[B_1]$, then at least one of the queries $\epsilon_1, \ldots, \epsilon_x$ must have a different answer. To understand why, note that if all query answers the adversary received remain the same after the change, then there is no way for the adversary to always infer $v[B_1]$ correctly from the query answers, because any deterministic algorithm that takes the answers to $\epsilon_1, \ldots, \epsilon_x$ as input will output the exact same value when $v[B_1] = \theta$. Also, note that since the only change to the database is the value of $v$, the difference on query answer must be whether the query returns $v$ or not. Thus, the adversary can always find at least one query which returns $v$.


We now construct a database which requires an expected number of $\Omega(\prod^{m^\prime}_{i=1}|V^\mathrm{B}_i|)$ queries for finding one query which returns $v$. Let there be $\sum^{m^\prime}_{i=1}(|V^\mathrm{B}_i| - 1)$ tuples in the database:
\begin{align}
&v^1_1, \ldots, v^{|V^\mathrm{B}_1|-1}_1\\
&v^1_2, \ldots, v^{|V^\mathrm{B}_2|-1}_2\\
&\ldots \nonumber\\
&v^1_{m^\prime}, \ldots, v^{|V^\mathrm{B}_{m^\prime}|-1}_{m^\prime}
\end{align}
Specifically, for any $j \in [1, m^\prime]$, every $v^i_j$ ($i \in [1, |V^\mathrm{B}_j| - 1]$) shares the same value as $v$ on all attributes but $B_j$. In addition, each $v^i_j$ takes a unique domain value in $V^\mathrm{B}_j$ that is different from $v[B_j]$. Due to the existence of these tuples, any query $q$ which differs from $v$ on at least one attribute will not return $v$ - because, according to the monotonicity condition, there must exist at least one tuple in $v^i_j$ with a smaller distance from $q$. Since the adversary was given prior knowledge of $v[A_1], \ldots, v[A_m]$ but no information about the private attribute values, the optimal adversarial strategy is to issue queries $q$ with $q[A_i] = v[A_i]$ for all $i \in [1, m]$ and $q[B_i]$ chosen uniformly at random from $V^\mathrm{B}_i$. One can see that, in this worst-case scenario, the expected query cost required for finding a query that returns $v$ is $\Omega(\prod^{m^\prime}_{i=1}|V^\mathrm{B}_i|)$.
\end{proof}

\vspace{1mm}\noindent{\bf Running Example Query Cost:} We now consider how Q\&I-Point performs over the running example of a linear-combination ranking function in Equation~\ref{equ:scr} and a database where each tuple is generated i.i.d.~randomly according to the uniform distribution, while the victim $v$ is chosen uniformly from the database.
\begin{theorem} \label{thm:feb}
In the running example, the expected number of queries Q\&I-Point issues to compromise $v[B_1]$ is at most $1/p + \sum^{m^\prime}_{i=1} (|V^\mathrm{B}_i| - 1)$, where
\begin{small}
\begin{align}
p = \prod_{t \in D, t \neq v} \left(\frac{1}{2} + \frac{1}{2} \cdot \mathrm{erf}\left(\frac{d^A(v, t)}{\sqrt{\sum^{m^\prime}_{i=1} 2w^{\prime2}_i \cdot \frac{|V^\mathrm{B}_i| - 1}{|V^\mathrm{B}_i|^2}}}\right)\right) \label{equ:qcq}
\end{align}
\end{small}
\hspace{-1mm}where $\mathrm{erf}(\cdot)$ is the standard error function \cite{cox2006principles}, 
and $d^A(v, t)$ is the distance between $v$ and $t$ on public attributes - i.e., $d^A(v, t) = w_1 \cdot \rho(v[A_1], t[A_1]) + \cdots + w_m \cdot \rho(v[A_m], t[A_m])$, where $\rho$ is the distance function defined in the running example.
\end{theorem}

\begin{proof}
Note that $q$ generated in the above-described random process has $d(q, v)$ following Multinomial distribution with mean $\mu$ and variance $\sigma^2$ as follows.
\begin{align}
\mu = \sum^{m^\prime}_{i=1} w^{\prime}_i \cdot \frac{1}{|V_i^B|}; 
\hspace{2mm} 
\sigma^2 = \sum^{m^\prime}_{i=1} w^{\prime2}_i \cdot \frac{|V_i^B-1|}{|V_i^B|^2}
\end{align}
In addition, $\forall t \in D$, given $d^A(q, t)$, the overall distance $d(q, t)$ follows the Multinomial distribution with mean $\mu_0 = d^A(q, t) + \mu$, and the same variance $\sigma^2$ as above. Note that since $q$ shares the same attribute values as $v$ on all public attributes, we have $d^A(q, t) = d^A(v, t)$. As such, the probability for a query $q$ with $q[A_i] = v[A_i]$ for all $i \in [1, m]$ and $q[B_i]$ chosen uniformly at random from $V_i^B$ to return $v$ is
\begin{align}
p = \prod_{t \in D, t \neq v} \left(\frac{1}{2} + \frac{1}{2} \cdot \mathrm{erf}\left(\frac{d^A(v, t)}{\sqrt{\sum^{m^\prime}_{i=1} 2w^{\prime2}_i \cdot \frac{|V^\mathrm{B}_i| - 1}{|V^\mathrm{B}_i|^2}}}\right)\right) \label{equ:ppq}
\end{align}
In other words, the expected number of queries the adversary needs to issue before finding a query that returns $v$ is
\begin{align}
\frac{1}{p} = \frac{1}{\prod_{t \in D, t \neq v} \left(\frac{1}{2} + \frac{1}{2} \cdot \mathrm{erf}\left(\frac{d^A(v, t)}{\sqrt{\sum^{m^\prime}_{i=1} 2w^{\prime2}_i \cdot \frac{|V^\mathrm{B}_i| - 1}{|V^\mathrm{B}_i|^2}}}\right)\right)}
\label{equ:qpq}
\end{align}

Of course, after finding a query that returns $v$, the adversary in the worst case has to issue an additional $\sum^{m^\prime}_{i=1} (|V^\mathrm{B}_i| - 1)$ queries (see proof of Theorem~\ref{thm:wc1}) for inferring the value of $v[B_1]$.
\end{proof}

Note from (\ref{equ:qcq}) why the average-case query cost of FIND-Q (and thereby Q\&I-Point) is likely much smaller than its worst-case bound: $p$ is the probability for a query $q$ randomly tested in FIND-Q to return $v$. One can see that, when there is a large number of public attributes or a small number of private ones - i.e., a larger $d^A(v, t)$ or a smaller $w^\prime_i$, the probability for a tuple $t$ ($t \neq v$) to ``overcome'' its difference with $q$ on public attributes (with which $v$ has zero difference) by private attribute values is fairly small - leading to a larger $p$ and, ultimately, a smaller query cost.

The query cost required for Q\&I-Point to compromise $v[B_1]$ actually {\em decreases} with a {\em smaller} weight on the private attributes. This observation seems counter-intuitive because when $w^\prime_1 = 0$, no privacy disclosure occurs as the rank becomes independent of $v[B_1]$ - but the worst disclosure occurs when $w^\prime_1$ takes the smallest positive value! To understand why, note that the smaller private ranking weights $w^\prime_i$ are, the easier it is for an adversary to pinpoint a query that returns $v$, as the adversary already has prior knowledge of all public attribute values of $v$. Given that, for a Q\&I-adversary, finding a query returning $v$ is (almost) equivalent with compromising $v[B_1]$, we have this seemingly counter-intuitive observation.

\subsection{Q-only adversary}
\noindent{\bf Design of Q-Point:} For adversaries subject to the tuple-insertion constraint, the feasibility of compromising $v[B_1]$ is not of certainty as in the Q\&I adversary case, as shown in the following theorem.
\begin{theorem} \label{thm:fqp}
Given any victim tuple $v$, there exists a ranking function $s(\cdot | \cdot)$ and a database $D$ such that no Q-only adversary can perform a rank-based inference of $v[B_1]$ over $D$.
\end{theorem}

\begin{proof}
Consider the linear-combination ranking function defined in (\ref{equ:scr}). We now show that there exists certain value combinations of $w_i$ and $w^\prime_i$ that make it impossible for a Q-only adversary to compromise $v[B_1]$ as long as there does {\em not} exist another tuple $v^\prime$ in $D$ which shares the same value with $v$ on all attributes but $B_1$. Specifically, if $\forall S_1, S_2$ $\subseteq$ $\{w_1,$ $\ldots$, $w_m$, $w^\prime_2$, $\ldots$, $w^\prime_{m^\prime}\}$ where $S_1 \neq S_2$, there is $|\sum S_1 - \sum S_2| > w^\prime_1$, then one can see that no query answer will be changed for all values of $v[B_1]$, because, for any query $q$, the change of $d(v, q)$ caused by changing the value of $v[B_1]$ is smaller than even the smallest possible rank difference between any two tuples. In other words, it is infeasible for a Q-only adversary to compromise $v[B_1]$.
\end{proof}

As a simple example, consider the linear ranking function in the running example and a database with only two attributes, one public $A_1$ and one private $B_1$. If the weighting on $A_1$ is larger than $B_1$, and each tuple in the database takes a different value on $A_1$, then there is no way for a Q-only adversary to infer $v[B_1]$ because the results of every possible query is already determined without knowing the value of $B_1$ for any tuple. Specifically, a query will always return the tuple that shares its value on $A_1$, regardless of what values the tuples have on $B_1$. As such, the inference of $v[B_1]$ from tuple ranks becomes infeasible.

Despite of the worst-case infeasibility, however, in practice it is quite likely for a Q-only adversary to find enough queries to unveil $v[B_1]$, as we shall show in the experimental results. To address these cases, we now develop Algorithm Q-Point for a Q-only adversary to launch a rank-based inference attack over a point query interface. Once again, our goal here is to find a pair of {\em differential queries} $q_\theta$ and $q^\prime_\theta$ for each value $\theta \in V^\mathrm{B}_1 \backslash v[B_1]$. Like in the Q\&I-case, without loss of generality, we denote the domain values in $V^\mathrm{B}_1$ as 0, 1, \ldots, $|V^\mathrm{B}_1| - 1$.

We start by calling Algorithm FIND-Q to find a query $q$ which returns $v$. Then, we construct and issue $|V^\mathrm{B}_1|$ queries $f_0(q)$, $f_1(q)$, \ldots, $f_{|V^\mathrm{B}_1| - 1}(q)$. While all these queries share the exact same predicates as $q$ on $A_1$, $\ldots$, $A_m$, $B_2$, $\ldots$, $B_{m^\prime}$, there is $f_i(q)[B_1] = i$ for all $i \in [0, |V^\mathrm{B}_1| - 1]$. Due to the additivity property, at least one of these $|V^\mathrm{B}_1|$ queries must return $v$.  If only one does, then our attack on $v[B_1]$ already succeeds - the one which returns $v$ must have the same value as $v$ on $B_1$. If more than one return $v$, we can do two things: First, we can exclude from consideration those values corresponding to the queries that do not return $v$ - for those, we have already found their differential queries to support the exclusion. Second, we can proceed to revise $q$ (and correspondingly $f_i(q)$) as follows to continue the exclusion process.

\begin{figure}[h]
\centering
\includegraphics[scale=0.4]{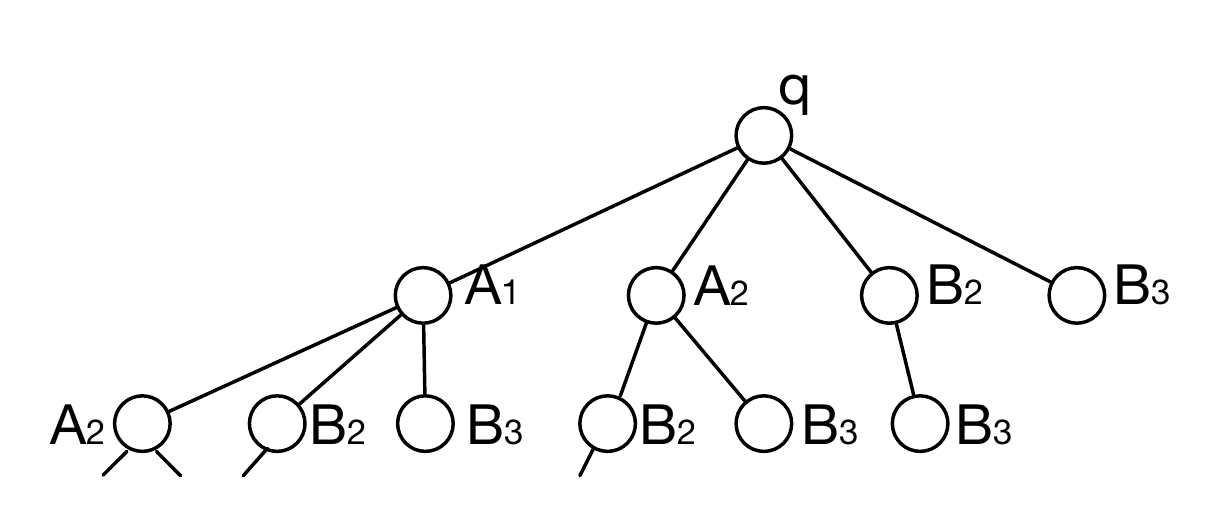}
\caption{Enumeration Tree in Algorithm Q-Point}
\label{fig:treeqp}
\end{figure}

Specifically, our query-revision process can be considered performing a breadth-first search over the tree structure depicted in Figure~\ref{fig:treeqp}, which demonstrates a special case where $m = 2$ and $m^\prime = 3$. In the tree, each node consists of a class of revisions to $q$. Specifically, a node contains all queries that differ from $q$ exactly on the attributes that appear on the path from the node to the root in the tree. For example, the bottom-left corner node in Figure~\ref{fig:treeqp} contains all queries that differ from $q$ on $A_1$ and $A_2$.

During the search process, for each node encountered, we enumerate all queries $q^\prime$ in the node and repeat the value-exclusion process described above by issuing $f_i(q^\prime)$ for all $i \in [0, |V^\mathrm{B}_1|-1]$. Note that the enumeration can be made more efficient with a pruning-based optimization: If for a query $q^\prime$, none of the $|V^\mathrm{B}_1|$ queries $f_i(q^\prime)$ returns $v$, then we can safely exclude from future consideration all queries in the subtree of the current node which only differs from $q^\prime$ on public attribute values. Algorithm~\ref{alg:qPoint} summarizes the pseudocode for Algorithm Q-Point.

\vspace{1mm}
\noindent{\bf Performance Analysis:} 
One can see from the design of Q-Point that, in the worst-case scenario, it issues enough queries to determine for every query specifiable through the point-query interface whether it returns $v$. Thus, Q-Point always accomplishes the attack as long as such an attack is at all feasible over the point-query interface. Nonetheless, the query complexity of Q-Point is $O(\prod^{m}_{i=1}|V^\mathrm{A}_i| \cdot \prod^{m^\prime}_{i=1}|V^\mathrm{B}_i|)$ - much higher than Q\&I-Point, given that real-world databases often feature more public attributes with large domains.

We consider again the linear ranking function in the running example and a database where each tuple is generated i.i.d.~uniformly at random. We have the following result for Q-Point:
\begin{theorem} \label{thm:rlb}
In the above scenario, given $q$ produced by FIND-Q, the probability (taken over the randomness of database $D$) for Q-Point to infer $v[B_1]$ after issuing only $|V^\mathrm{B}_1|$ queries is at least
\begin{small}
\begin{align}
\left(1 - \prod_{t \in D, t \neq v} \frac{1 + \mathrm{erf}\left(\frac{(d^A(v, t) - w^\prime_1)}{\sqrt{2\sum^{m^\prime}_{i=2} (w^{\prime2}_i \cdot (|V^\mathrm{B}_i| - 1) / |V^\mathrm{B}_i|^2)}}\right)}{1 + \mathrm{erf}\left(\frac{d^A(v, t)}{\sqrt{2\sum^{m^\prime}_{i=2} (w^{\prime2}_i \cdot (|V^\mathrm{B}_i| - 1) / |V^\mathrm{B}_i|^2)}}\right)}\right)^{|V^\mathrm{B}_1| - 1}
\end{align}
\end{small}
\end{theorem}
\begin{proof}
Following the results from the proof of Theorem~\ref{thm:feb}, the expected ratio of point queries $q$ which has $q[A_i] = v[A_i]$ for all $i \in [1, m]$ and returns $v$ is
\begin{align}
p = \prod_{t \in D, t \neq v} \left(\frac{1}{2} + \frac{1}{2} \cdot \mathrm{erf}\left(\frac{d^A(v, t)}{\sqrt{\sum^{m^\prime}_{i=1} 2w^{\prime2}_i \cdot \frac{|V^\mathrm{B}_i| - 1}{|V^\mathrm{B}_i|^2}}}\right)\right)
\end{align}
because $d^A(v, t)$ is the extra distance a tuple $t \neq v$ has compared with $v$ - and such a distance has to be ``covered'' by the private attributes in order for $t$ to be returned. A key observation here is that, changing $q[B_1]$ between adjacent values in $V_1^B$ changes the distance by at most $w^\prime_1$. Thus, in order for a tuple $t$ to be returned after the change, the private attributes of $t$ still have to ``cover'' a distance of at least $d^A(v, t) - w^\prime_1$. In other words, the expected ratio of point queries $q$ which (1) has $q[A_i] = v[A_i]$ for all $i \in [1, m]$, (2) returns $v$, and (3) still returns $v$ after the flip of $q[B_1]$ is at most
\begin{align}
p^\prime \leq \prod_{t \in D, t \neq v} \left(\frac{1}{2} + \frac{1}{2} \cdot \mathrm{erf}\left(\frac{(d^A(v, t) - w^\prime_1)}{\sqrt{2\sum^{m^\prime}_{i=2} (w^{\prime2}_i \cdot (|V^\mathrm{B}_i| - 1) / |V^\mathrm{B}_i|^2)}}\right) \right)
\end{align}
Thus, among all point queries $q$ which have $q[A_i] = v[A_i]$ for all $i \in [1, m]$ and return $v$, the ratio that, upon changing the value of $q[B_1]$ to all values in $V_1^B \setminus v[B_1]$, return another tuple in the database is at least
\begin{align}
& \left(1 - p^\prime/p \right)^{|V_1^B|-1} =    \nonumber \\
& \left(1 - \prod_{t \in D, t \neq v} \frac{1 + \mathrm{erf}\left(\frac{(d^A(v, t) - w^\prime_1)}{\sqrt{2\sum^{m^\prime}_{i=2} (w^{\prime2}_i \cdot (|V^\mathrm{B}_i| - 1) / |V^\mathrm{B}_i|^2)}}\right)}{1 + \mathrm{erf}\left(\frac{d^A(v, t)}{\sqrt{2\sum^{m^\prime}_{i=2} (w^{\prime2}_i \cdot (|V^\mathrm{B}_i| - 1) / |V^\mathrm{B}_i|^2)}}\right)}\right)^{|V^\mathrm{B}_1| - 1}
\label{equ:rlb}
\end{align}

One can see that if a point query $q$ has $q[A_i] \neq v[A_i]$ for certain $i \in [1, m]$, then this ratio must be even higher because of the now shorter distance $d^A(q, t)$ a tuple $t \neq v$ needs to ``cover'' with the private attributes. Thus, (\ref{equ:rlb}) is indeed a lower bound on the expected ratio for all point queries which return $v$.
\end{proof}

Similar to the Q\&I case, the attack is more (likely to be) efficient with a smaller $|V^\mathrm{B}_1|$. Nonetheless, an interesting observation here is that, contrary to the Q\&I case, now the larger $w^\prime_1$ is, the more efficient the attack is likely to be. On the other hand, the efficiency also increases with a larger database size $|D|$ and a smaller weight on other private attributes $w^\prime_i$ (as $\mathrm{erf}(x)$ has a larger derivative when $x$ is close to 0). 

\begin{algorithm}[!htb]
\caption{{\bf Q-Point}}
\begin{algorithmic}[1]
\label{alg:qPoint}
\STATE {\bf Input:} $v$ \qquad {\bf Output:} $v[B_1]$
\STATE {\bf while} some query returns $v$ {\bf do}
    \STATE \hspace{\algorithmicindent} $q \leftarrow$ FIND-Q($v, H_v$); Construct enumeration tree $T_q$ for $q$
    \STATE \hspace{\algorithmicindent} {\bf for } $i = 1$ to  $m+m^{\prime}$ {\bf do}
        \STATE \hspace{\algorithmicindent} \hspace{\algorithmicindent} 
                    {\bf for} each query node $q^{\prime}$ in level $i$ of $T_q$ {\bf do}
            \STATE \hspace{\algorithmicindent} \hspace{\algorithmicindent} \hspace{\algorithmicindent}
                    Construct queries $f_0(q^{\prime}) \ldots f_{|V_1^B|-1}(q^{\prime})$
            \STATE \hspace{\algorithmicindent} \hspace{\algorithmicindent} \hspace{\algorithmicindent} 
                    {\bf if} none return $v$ {\bf then} prune subtree($q^{\prime}$)
            \STATE \hspace{\algorithmicindent} \hspace{\algorithmicindent} \hspace{\algorithmicindent} 
                    {\bf if} only $f_j(q^{\prime})$ returns $v$ {\bf then} {\bf return} $j$ as $v[B_1]$ 
            \STATE \hspace{\algorithmicindent} \hspace{\algorithmicindent} \hspace{\algorithmicindent} 
                    Exclude query nodes $f_k(q^{\prime})$ that does not return $v$
\RETURN failure
\end{algorithmic}
\end{algorithm}


\section{IN Query Interface} \label{sec:rq}

\subsection{Q\&I Adversary}

For Q\&I adversaries, the feasibility of rank-based inference attack is established for point-query interface in \S\ref{sec:pq}. Since point-query interface is a special case of IN, the attack feasibility here is already established. Thus, our focus here is to study how the additional power of IN queries further empowers Q\&I adversaries.

Recall from \S\ref{sec:pq} that, for Q\&I adversaries, rank-based inference can be fairly efficiently reduced to the task of FIND-Q - i.e., identifying a (now IN) query returning victim $v$. We shall start by showing that, despite of the larger space of queries, the reduction still holds - leading to the design of Algorithm Q\&I-IN. Then, we show that, while FIND-Q for IN has the same worst-case query cost as in Q\&I-Point, the query cost in practice is likely much smaller.

\subsubsection{Reduction to finding an IN query that returns $v$}

We start by showing that, so long as a Q\&I adversary can call upon FIND-Q to identify an {\em IN query} $q$ that returns $v$, it can always infer $v[B_1]$ within $O(m^{\prime} \cdot \sum^{m^\prime}_{i=1}|V^\mathrm{B}_i|)$ queries. The reduction in \S\ref{sec:pq} cannot be directly used here as it relies on the ability to find a {\em point} query returning $v$ - which we do not want FIND-Q to do over an IN-query interface due to high query cost associated with it.

To enable the reduction to finding an IN query, the only difference from point-query case (\S\ref{sec:rfo}) is that now the input $q$ might have ranges like $\{0, 1, 2\}$ specified as predicates on $B_i$, instead of a single value as in the point-query case (which we denoted as 0). Fortunately, this change does not alter the key design of reduction construction.  What we do now is to define $B^\prime_1, \ldots, B^\prime_h$ as those attributes on which the inserted tuple $t$ has a value that differs from the set specified in $q$. This could be that $t[B^\prime_i]$ falls outside of the range specified in $q[B^\prime_i]$ (e.g., when $t[B^\prime_i] = 3$ while $q[B^\prime_i] = \{0, 1, 2\}$; or that $t[B^\prime_i]$ is in the range but not the only element of $q[B^\prime_i]$ (e.g., when $t[B^\prime_i] = 0$ and $q[B^\prime_i] = \{0, 1, 2\}$. Here is an example of the sequence of queries we construct:
\begin{table}[h]
\begin{center}
\begin{tabular}{cccccc}
& $A_1, \ldots, A_m$ & $B^\prime_1$ & $B^\prime_2$ & $B^\prime_3$ & $B_\mathrm{others}$\\
\hline
$q$ & \{0\} & \{0,1\} & \{1,2\} & \{1\} & \{0\}\\
\hdashline
$q_0$ & \{0\} & \{0,1\} & \{1,2\} & \{1\} & \{0\}\\
$q_1$ & \{0\} & \{0,1\} & \{1,2\} & \{0\} & \{0\}\\
$q_2$ & \{0\} & \{0,1\} & \{0\} & \{0\} & \{0\}\\
$q_3$ & \{0\} & \{0\} & \{0\} & \{0\} & \{0\}\\
\hdashline
$t$ & 0 & 0 & 0 & 0 & 0\\
\end{tabular}
\end{center}
\vspace{-2mm}
\end{table}

Once again, there must exist a pair of adjacent queries $q_i$ and $q_{i+1}$ such that $q_i$ returns $v$ while $q_{i+1}$ does not. The remaining inference process follows \S\ref{sec:rfo}. For example, if $q_1$ returns $v$ but $q_2$ does not, then we can safely infer that $v[B^\prime_2] \neq 0$ due to the additivity condition. Similarly, if $q_2$ returns $v$ while $q_3$ does not, we can infer that $v[B^\prime_1] \neq 0$. Thus, just like in the Q\&I-Point case, excluding the query cost of FIND-Q, a Q\&I-adversary requires at most $O(m^\prime \cdot \sum^{m^\prime}_{i=1}|V^\mathrm{B}_i|)$ queries to compromise $v[B_1]$.

\subsubsection{Efficiency Enhancement in Q\&I-IN}
Given that the reduction still holds, we are now ready to study how IN queries empower an adversary to quickly accomplish FIND-Q and find a query that returns $v$. In the following, we describe a concrete example which demonstrates the significant saving brought by IN queries, followed by the design of Algorithm Q\&I-IN.

\vspace{1mm}
\noindent{\bf Example of significant query savings:} To understand why IN queries significantly reduce the query cost, consider a simple example where: (1) the number of public attributes $m$ is sufficiently large, so each tuple in the database has a unique value combination for the $m$ public attributes; and (2) the number of private attributes $m^\prime$ is even larger, so the probability for a randomly generated point query to return $v$ is extremely small.

The first observation from this example is that FIND-Q over a point-query interface actually requires an extremely large number of queries. Specifically, note from Theorems \ref{thm:wc1} and \ref{thm:feb} that, for a given $m$, the query cost can be made arbitrarily large with an increasing $m^\prime$. On the other hand, if IN queries are available, the attack query cost - more specifically, the number of queries required to find one query returning $v$ - is exactly 1 because an IN query $q$ with $A_i = v[A_i]$ for $i \in [1, m]$ and $B_j = V^\mathrm{B}_j$ (essentially ``*'', i.e., do-not-care) for $j \in [1, m^\prime]$ always returns $v$.

One can see from the example that the usage of IN queries significantly reduces the attack query cost because of a simple reason: the ability for an adversary to {\em eliminate} all private attributes from a query specification makes it much easier for FIND-Q to unveil the victim tuple from the database, so that the adversary can compromise the private attributes one at a time using the above-described reduction. In other words, with IN queries, an adversary no longer has to get lucky and guess multiple private attributes correctly at the same time (e.g., in order to have $v$ returned by a point query).

\vspace{1mm}
\noindent{\bf Design of Q\&I-IN:} 
Algorithm~\ref{alg:qAndIRange} 
depicts the pseudocode for Algorithm Q\&I-IN, which enables a Q\&I-adversary to launch our rank-based inference attack on $v[B_1]$ over an IN query interface. With the algorithm, we start with a query $q$ which has $q[A_i] = v[A_i]$ for all $i \in [1, m]$ and $q[B_j] = V^\mathrm{B}_j$ for all $j \in [1, m^\prime]$. Then, if $q$ does not return $v$, we gradually replace predicates on $B_i$ with point predicates (i.e., $B_i = v$ where $v \in V^\mathrm{B}_i$). Specifically, we perform what is essentially a {\em breadth-first search} process which enumerates all value combinations for $B_1$, $\{B_1, B_2\}$, $\{B_1, B_2, B_3\}$, $\ldots$, $\{B_1, \ldots, B_{m^\prime}\}$ in order. For example, when $V^\mathrm{B}_1 = V^\mathrm{B}_2 = \{0, 1\}$, the queries we issue are $B_1 = 0$, $B_1 = 1$, $B_1 = 0$ AND $B_2 = 0$, $B_1 = 0$ AND $B_2 = 1$, $B_1 = 1$ AND $B_2 = 0$, $B_1 = 1$ AND $B_2 = 1$, $\ldots$, where each query also includes $q[A_i] = v[A_i]$ for all $i \in [1, m]$ and $q[B_j] = V^\mathrm{B}_j$ for all unspecified $B_j$. When we find a query that returns $v$, we launch the above-described reduction process to complete the attack of $v[B_1]$.

One can see from the algorithm design that, just like in the point-query case, we guarantee a successful attack. But the worst-case query cost for Q\&I-IN is also just like Q\&I-Point - i.e., $O(\prod^{m^\prime}_{i=1}|V^\mathrm{B}_i|)$. As we shall demonstrate in the following worst-case analysis, this query cost still cannot be improved beyond a constant factor.

\begin{algorithm}[!htb]
\caption{{\bf Q\&I-IN}}
\begin{algorithmic}[1]
\label{alg:qAndIRange}
\STATE {\bf Input:} $v$ \qquad {\bf Output:} $v[B_1]$
\STATE Initialize starting query $q$:  $q[A_i] = v[A_i] \, \forall i \in [1,m]$ and $q[B_j] = V_j^{B} \, \forall j \in [1,m']$
\STATE Iteratively convert $q$ to a point query till it returns $v$
\STATE $v[B_1] \leftarrow$ Q\&I-Point($q,v$)
\end{algorithmic}
\end{algorithm}

\subsubsection{Query cost analysis}

The main result here is that, while the availability of an IN query interface does {\em not} help a Q\&I adversary at all in the worst-case scenario, it does have the potential to significantly reduce the query cost in practice - especially when the number of public attributes is large. 
To understand why the worst-case scenario remains unchanged, consider the construction in the proof of Theorem~\ref{thm:wc1} which inserts to the database $\sum^{m^\prime}_{i=1}(|V^\mathrm{B}_i| - 1)$ tuples described in \S\ref{sec:pq}. Given the worst-case assumption that, when there is a draw (i.e., $s(t_1|q) = s(t_2|q)$), any inserted tuple will be returned before the victim $v$, one can see that the adversary gets no help from IN queries - because as long as a query $q$ contains an IN predicate, say on $B_i$, it is impossible for $q$ to return the victim tuple $v$ as there must exist an inserted tuple which matches $q$ on $B_i$, has the exact same value combination as $v$ on all other attributes, and therefore will be returned ahead of $v$ in the answer to $q$. Thus, the worst-case query cost Q\&I-IN remains $\Omega(\prod^{m^\prime}_{i=1}|V^\mathrm{B}_i|)$ - same as Q\&I-Point.

\vspace{1mm}
\begin{theorem} \label{thm:qra}
In the running example, the expected number of queries FIND-Q requires for finding a query that returns $v$ is $1$ if $\min_{t \in D, t \neq v} d^A(v, t) > 0$, and at most $\sum^{m^\prime - 1}_{h=1} (c_{h+1} \cdot (1 - (1-p(h))^{c_h})$ otherwise, where $c_h = \sum^h_{i=1}\prod^i_{j=1}|V^\mathrm{B}_i|$ and
\begin{small}
\begin{align}
p(h) = \prod_{t \in D, t \neq v} \left(\frac{1}{2} + \frac{1}{2} \cdot \mathrm{erf}\left(\frac{d^A(v, t)}{\sqrt{2\sum^h_{i=1} w^{\prime2}_i \cdot \frac{|V^\mathrm{B}_i| - 1}{|V^\mathrm{B}_i|^2}}}\right)\right)
\end{align}
\end{small}
\end{theorem}
\begin{proof}
First, when no other tuple $t \in D$ shares the same public-attribute value-combination as $v$ (i.e., $\min_{t \in D, t \neq v} d^A(v, t) > 0$), then as we discussed in the design of Q\&I-IN, only one query (with point-predicates and IN-predicates on all public and private attributes, respectively) is required. For other cases, in analogy to the proof of Theorem~\ref{thm:feb}, one can see that the probability for a randomly generated query $q$ with $q[A_i] = v[A_i]$ for all $i \in [1, m]$ and $h$ point-query predicates specified on private attributes to return $v$ is
\begin{align}
p(h) = \prod_{t \in D, t \neq v} \left(\frac{1}{2} + \frac{1}{2} \cdot \mathrm{erf}\left(\frac{d^A(v, t)}{\sqrt{2\sum^h_{i=1} w^{\prime2}_i \cdot \frac{|V^\mathrm{B}_i| - 1}{|V^\mathrm{B}_i|^2}}}\right)\right)
\end{align}
Since the total number of such queries is $c^h$, and the overall query cost after enumerating queries with $h$ or fewer predicates is (i.e., on $B_1$, $\{B_1, B_2\}$, $\ldots$, $\{B_1, \ldots, B_h\}$, as specified in Q\&I-IN) is $c^{h+1}$, the expected number of queries required by Q\&I-IN is at most $\sum^{m^\prime - 1}_{h=1} (c_{h+1} \cdot (1 - (1-p(h))^{c_h})$.
\end{proof}

One can see from the theorem the substantial promise for IN queries to significantly reduce the query cost - not only the query cost can be cut to 1 when no other tuple shares the same public-attribute value-combination as $v$, but the value of $p(h)$ - i.e., the probability for a query with $h$ point-predicates on private attributes to return $v$ - actually decreases with $h$. As such, the query cost is likely much smaller than Q\&I-Point, especially when the number of public attributes $m$ is large (which leads to a large $d^A(v, t)$).

\subsection{Q-only adversary}

Just like the availability of IN queries does not help reduce the worst-case query cost for Q\&I-adversaries, it cannot change the (in)feasibility result for Q-only adversaries either. To understand why, consider a database with the aforementioned linear ranking function and all tuples sharing the same value on $B_1$. Clearly, the returned tuples will be of the same order regardless of what range the query specifies on $B_1$. Thus, there is no way for a Q-only adversary to infer which value in $v[B_1]$ all tuples take - proving that Q-only adversaries cannot guarantee the success of rank-based inference even for IN query interfaces. Nonetheless, as we shall show in this subsection and in the experimental results, the availability of IN queries does help with reducing the query cost in practice, especially when the number of public attributes is large.

Algorithm~\ref{alg:qRange} 
depicts the pseudocode for Algorithm Q-IN, which enables a Q-only adversary to launch our rank-based inference attack on $v[B_1]$ over an IN query interface. We start with calling Algorithm FIND-Q to find one query $q$ which returns $v$. Note that, according to the design of FIND-Q, $q$ always has $q[A_i] = v[A_i]$ for all $i \in [1, m]$.
After obtaining $q$, Algorithm Q-IN issues $|V^\mathrm{B}_1|$ queries $f_0(q)$, \ldots, $f_{|V^\mathrm{B}_1| - 1}(q)$ defined in the same way as in \S\ref{sec:pq} - i.e., while all these queries are exactly the same as $q$ on $A_1$, $\ldots$, $A_m$, $B_2$, $\ldots$, $B_{m^\prime}$, there is $f_i(q)[B_1] = i$ for all $i \in [0, |V^\mathrm{B}_1| - 1]$. Similar to the discussion in Algorithm Q-Point, one can see that at least one of these queries must return $v$, and the attack is already successful if only one of them does. If more than one returns $v$, we can exclude from consideration those values corresponding to queries that do not return $v$, and then gradually revise $q$ according to the following procedure.

Specifically, we start with revising $q$ to $q_1, \ldots, q_m$ by changing the predicate of $q_i$ on $A_i$ to ($A_i$ IN $V^\mathrm{A}_i$). For each $q_i$ which returns $v$, we repeat the above process and issue $f_j(q_i)$ for each $j \in [0, |V^\mathrm{B}_1| - 1]$ that is not yet excluded as a possible value of $v[B_1]$. Once again, this either directly reveals $v[B_1]$ or further excludes additional values from consideration. If we still cannot pin down $v[B_1]$ after enumerating $q_1, \ldots, q_m$, we consider the process by setting an additional public attribute to its entire domain. For example, if $q_1$ returns $v$, we construct $q_{1,x_1}, \ldots, q_{1,x_h}$, such that (1) $q_{x_1}, \ldots, q_{x_h}$ also return $v$, and (2) $q_{1,i}$ is the same as $q_1$ on all attributes but $A_i$, for which there is $q_{1,i}[A_i] = V^\mathrm{A}_i$. We repeat this value-exclusion process until finding the exact value of $v[B_1]$, or when we have exhausted all combinations of public attributes - at which time we move back to Algorithm FIND-Q, find another query $q$ which returns $v$, and attempt the revision process again.


\begin{algorithm}[!htb]
\caption{{\bf Q-IN}}
\begin{algorithmic}[1]
\label{alg:qRange}
\STATE {\bf Input:} $v$ \qquad {\bf Output:} $v[B_1]$
\STATE {\bf while} some query returns $v$ {\bf do}
    \STATE \hspace{\algorithmicindent} $q \leftarrow$ FIND-Q($v, H_v$)
    \STATE \hspace{\algorithmicindent} {\bf for} $i=0$ to $m$ {\bf do}
        \STATE \hspace{\algorithmicindent} \hspace{\algorithmicindent}  {\bf for} each ${m \choose i}$ combination of $C$ of $\{A_1, \ldots, A_m\}$ {\bf do}
            \STATE \hspace{\algorithmicindent} \hspace{\algorithmicindent}  \hspace{\algorithmicindent}  
                    $q' \leftarrow q$; \qquad  $q[A_{i'}] = V_{i'}^A \quad \forall A_{i'} \in C$
            \STATE \hspace{\algorithmicindent} \hspace{\algorithmicindent}  \hspace{\algorithmicindent}  Construct queries $f_0(q') \ldots f_{|V_1^B|-1}(q')$
            \STATE \hspace{\algorithmicindent} \hspace{\algorithmicindent}  \hspace{\algorithmicindent}  {\bf if} only $f_j(q')$ returned $v$ then return $j$ as $V[B_1]$
            \STATE \hspace{\algorithmicindent} \hspace{\algorithmicindent}  \hspace{\algorithmicindent}  Exclude query nodes that did not return $v$
\end{algorithmic}
\end{algorithm}


One can see from the design of Algorithm Q-IN that its worst-case query cost is the same as Q-Point, i.e., $O(\prod^{m}_{i=1}|V^\mathrm{A}_i| \cdot \prod^{m^\prime}_{i=1}|V^\mathrm{B}_i|)$. For the running example and a database where each tuple is generated i.i.d.~uniformly at random, we have the following results:
\newtheorem{corollary}{Corollary}
\begin{corollary} \label{thm:rlbc}
In the above scenario, given $q$ from FIND-Q which (1) has point-predicates on $S \subseteq \{A_1, \ldots, A_m\}$, (2) has point-predicates on $B_1$ and $S^\prime \subseteq \{B_2, \ldots, B_{m^\prime}\}$, and (3) returns $v$, the probability (taken over the randomness of database $D$) for Q-IN to infer $v[B_1]$ after issuing only $|V^\mathrm{B}_1|$ queries is at least
\begin{small}
\begin{align}
\left(1 - \prod_{t \in D, t \neq v} \frac{1 + \mathrm{erf}\left(\frac{(d^S(v, t) - w^\prime_1)}{\sqrt{2\sum_{i: B_i \in S^\prime} (w^{\prime2}_i \cdot (|V^\mathrm{B}_i| - 1) / |V^\mathrm{B}_i|^2)}}\right)}{1 + \mathrm{erf}\left(\frac{d^S(v, t)}{\sqrt{2\sum_{i: B_i \in S^\prime} (w^{\prime2}_i \cdot (|V^\mathrm{B}_i| - 1) / |V^\mathrm{B}_i|^2)}}\right)}\right)^{|V^\mathrm{B}_1| - 1}
\end{align}
\end{small}
\end{corollary}

The corollary follows directly from Theorem~\ref{thm:rlb}. We can observe from the theorem the substantial promise for IN queries to significantly reduce the query cost - specifically, note that the smaller $S$ or $S^\prime$ is, the higher this expected ratio will be. As such, the overall query cost is likely much smaller than Q-Point.

\section{Discussions} \label{sec:ext}
\subsection{Numeric Attributes}
\vspace{1mm}
\noindent{\bf Attack Precision for Numeric Private Attribute:} In the original problem definition discussed in \S\ref{sec:pre}, we consider an attack to succeed if and only if the adversary unveils the exact value of a (Boolean or categorical) private attribute. For a numeric (private) attribute, however, it becomes more complex to measure the success of an attack. Specifically, as we shall demonstrate as follows, while there are cases where an adversary can able infer a numeric attribute value to an arbitrary precision, there are also cases where the precision is limited to a (small) fixed range. Nonetheless, either case still represents serious compromise of user privacy.

Interestingly, whether an adversary can infer a numeric attribute $B_i$ to arbitrary precision depends on the {\em ranking function}, specifically the definition of $s(q[B_i]|t[B_i])$, used by the query interface. If the query interface allows a range to be specified for each attribute, and the ranking function simply assigns $s(q[B_i]|t[B_i]) = 0$ if $t[B_i] \in q[B_i]$ and $1$ otherwise, then any adversary which can successfully launch the attack (i.e., finding $q_1$ and $q_2$ which only differ on $B_1$ yet return $t$ at different ranks) can always infer $t[B_i]$ to any precision level (by continuously shrinking $q[B_1]$) as long as the interface allows an arbitrarily small range to be specified in the query. On the other hand, if the interface is point-query only and the ranking function is $s(q[B_i]|t[B_i]) = |q[B_i] - t[B_i]|$ (or with range-query allowed and $s(q[B_i]|t[B_i])$ being the difference between $t[B_i]$ and the center point of $q[B_i]$) with precision set to two digits after decimal point, then clearly no adversary can infer $v[B_1]$ beyond a precision level of 0.01.

Given the wide variety of ranking functions a query interface might feature, and the fact that even a fairly wide interval on $B_1$ (as long as it is significantly narrower than the entire domain) is usually a significant threat to privacy in practice, discussing the achievable precision for each type of interfaces is beyond the scope of this paper. Instead, we make an assumption that numeric attributes can be properly discretized (and treated as a categorical one) according to two principles: (1) the discretized range is narrow enough so each tuple has a unique value combination of all attributes, and (2) the range should be as wide as possible, so as to minimize $|V^\mathrm{A}_i|$ and $|V^\mathrm{B}_i|$, thereby minimizing the query cost of the attack.

\subsection{Defense Against Rank-based Inference}

Since our main objective here is to unveil a novel rank-based inference attack on web databases, a comprehensive discussion of defense methodologies is beyond the scope of this paper. Nonetheless, we would like to briefly describe a few simple defense strategies, and discuss how the analysis of various algorithms in the paper might shed lights on the design of defense.

An obvious defense methodology is to enforce more stringent {\em practical constraints} discussed in the paper - e.g., requiring a user to answer a CAPTCHA challenge \cite{von2003captcha} before issuing each query, performing rigid authentication for each tuple insertion/update operation, etc. Another possible strategy here is to {\em delay} any new tuples from appearing in query answers. As one can see from the design of Q\&I-Point and Q\&I-IN, this delay may significantly prolong the amount of time a Q\&I-adversary needs to launch the attack. However, it is important to note that all defense strategies in this category are essentially making a tradeoff between privacy protection and the {\em convenience} of bona fide users, and therefore must be designed and implemented carefully (e.g., after user studies).

Another category of defense is to adjust the assignment of public/private attributes and/or the design of ranking function. Recall from the discussions of Q\&I-Point and Q-Point that the more attributes the database owner assigns to be private, and the higher weights the ranking function assigns on private attributes, the more difficult it is for an adversary to launch the attack, as the prior knowledge held by an adversary on the victim tuple (i.e., $v[A_1]$, $\ldots, v[A_m]$) now plays a lesser role on determining the rank, making it harder for the adversary to efficiently locate the victim tuple. 

Nonetheless, this strategy does not work as effectively on an IN-query interface. To understand why, note from the design of Q\&I-IN that, as long as the public attributes are sufficient for uniquely identifying the victim tuple, a Q\&I-adversary can succeed with $O(\sum_{i=1}^{m\prime} |V_i^B|)$ queries no matter how much weight the ranking function places on the private attributes. In this case, the defender can choose to publicize fewer attributes (if doing so prevents an adversary from learning these attribute values for the victim tuple), or disabling IN-query predicates on certain attributes. As we discussed in \S\ref{sec:rq}, the reduction of IN-query predicates may significantly delay the attack in the average-case scenario.


\begin{figure*}[ht]
\begin{minipage}[t]{0.25\linewidth}
\centering
\includegraphics[width = 55mm, height = 32mm]{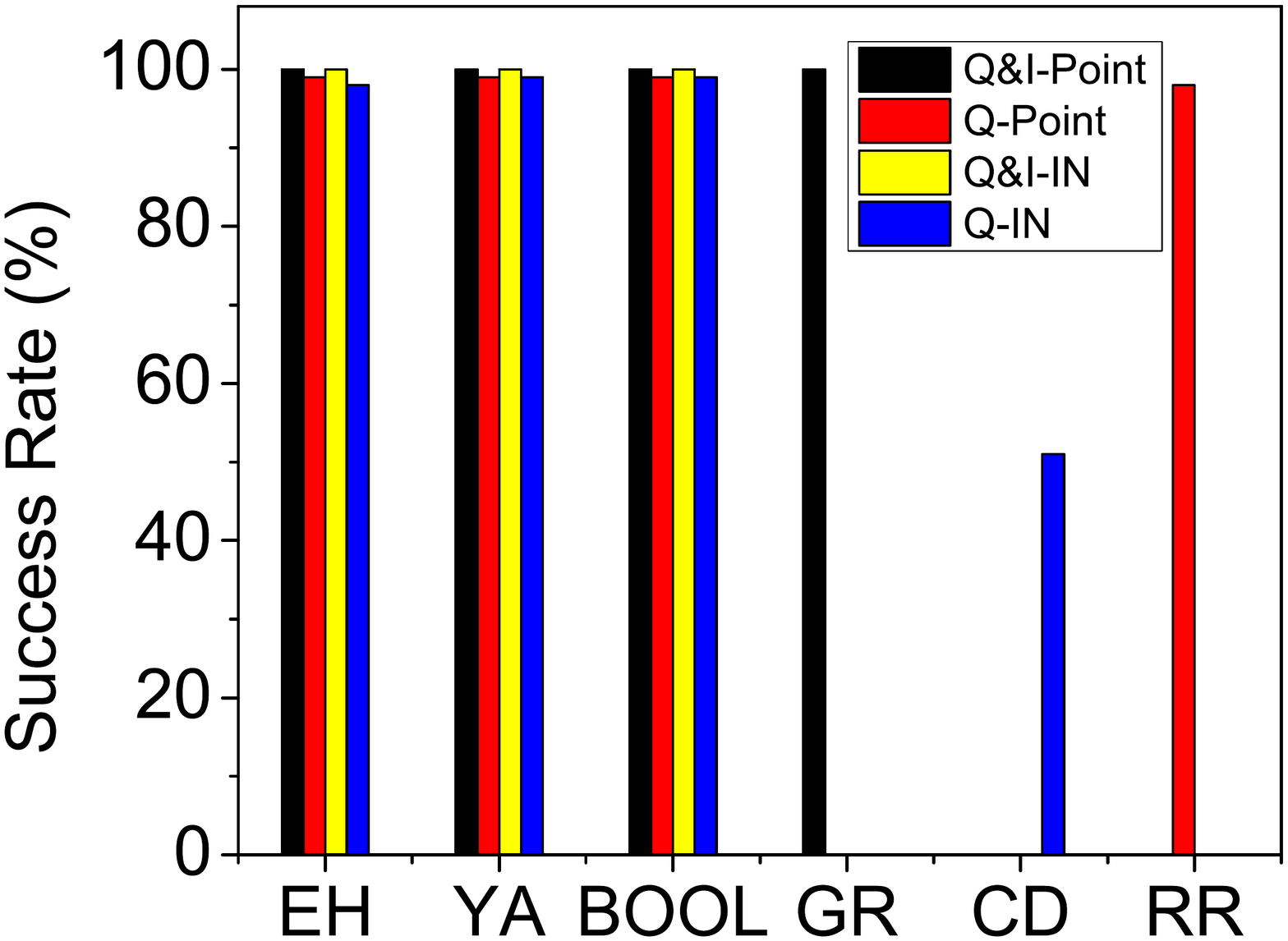}
\vspace{-7mm}\caption{Attack Success Rate}
\label{fig:attackSuccessRate}
\end{minipage}
\hspace{1mm}
\begin{minipage}[t]{0.23\linewidth}
\centering
\includegraphics[width = 45mm, height = 32mm]{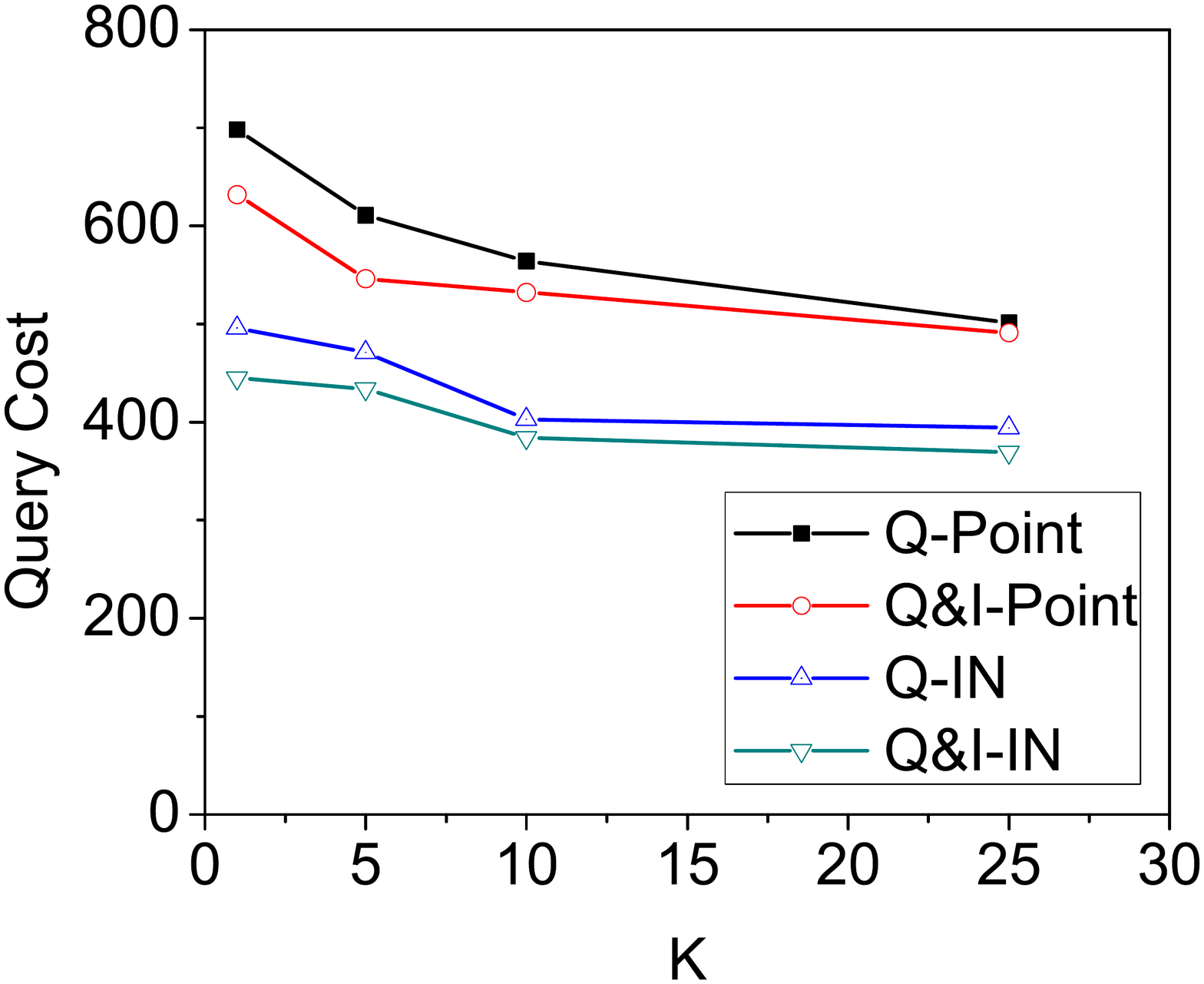}
\vspace{-7mm}\caption{Varying $k$}
\label{fig:eh_kVsQueryCost}
\end{minipage}
\hspace{1mm}
\begin{minipage}[t]{0.23\linewidth}
\centering
\includegraphics[width = 45mm, height = 32mm]{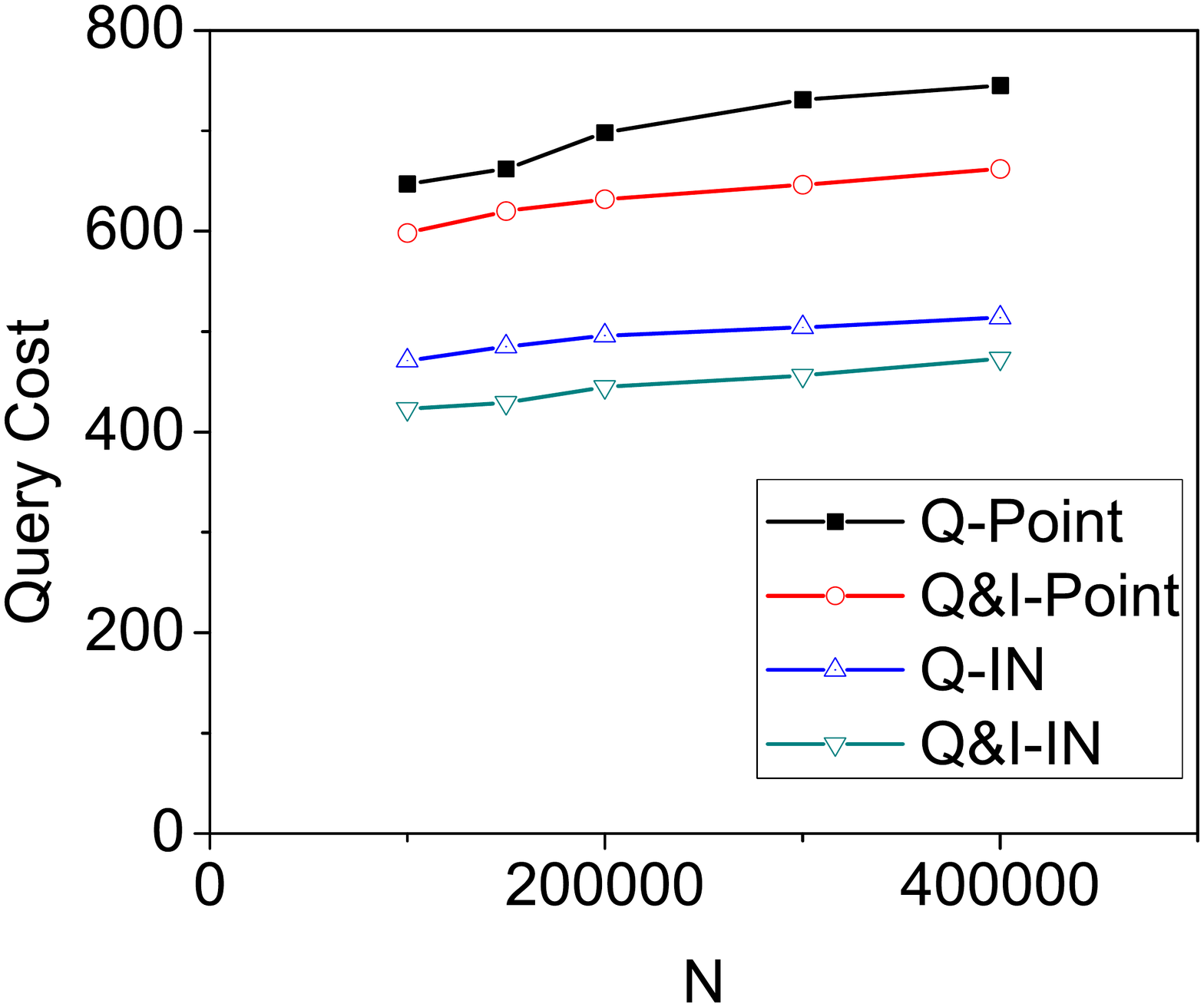}
\vspace{-7mm}\caption{Varying $n$}
\label{fig:eh_NVsQueryCost}
\end{minipage}
\hspace{1mm}
\begin{minipage}[t]{0.23\linewidth}
\centering
\includegraphics[width = 45mm, height = 32mm]{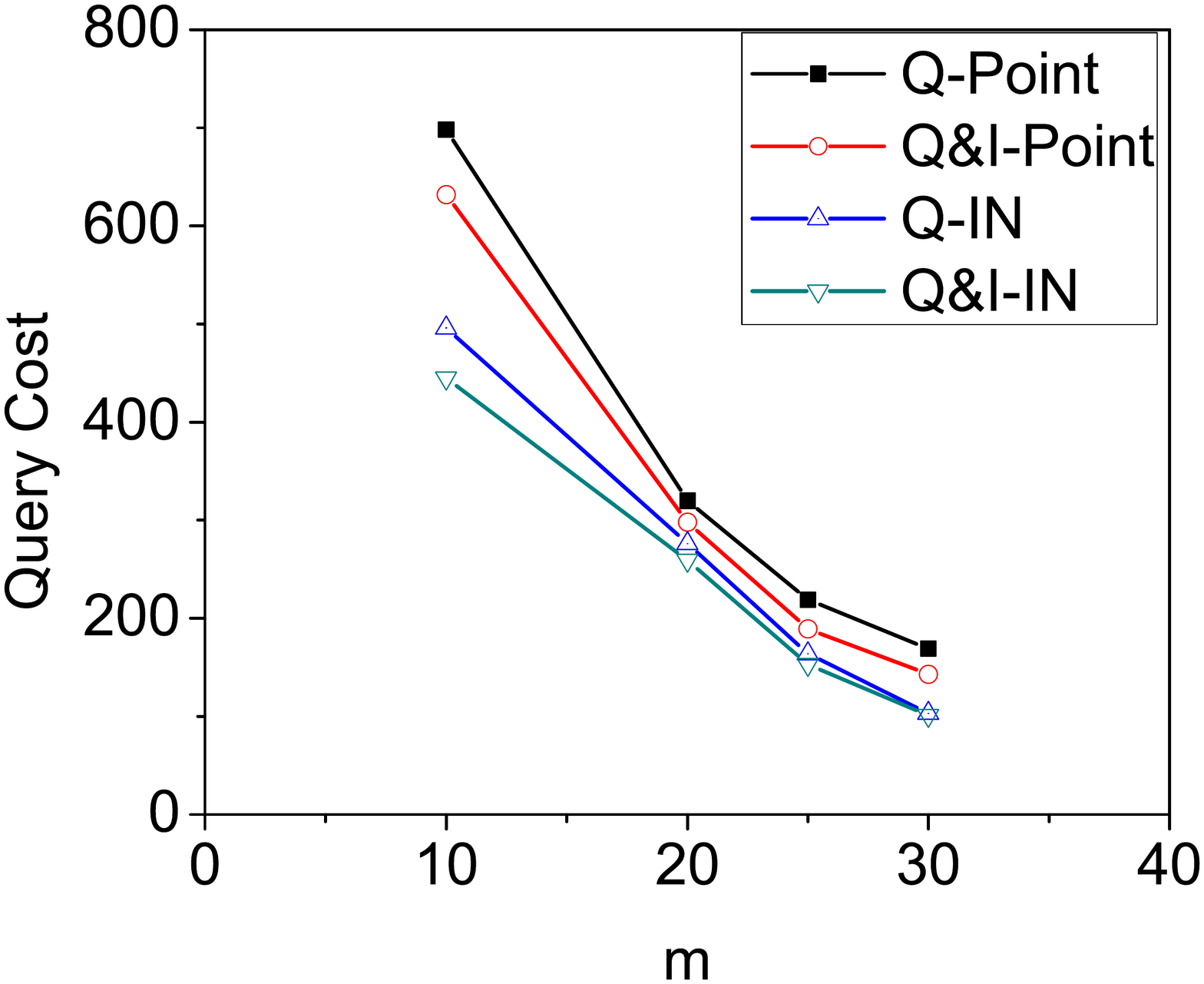}
\vspace{-7mm}\caption{Varying $m$}
\label{fig:eh_mVsQueryCost}
\end{minipage}
\hspace{-2mm}
\end{figure*}

\begin{figure*}[t]
\begin{minipage}[t]{0.23\linewidth}
\centering
\includegraphics[width = 50mm, height = 32mm]{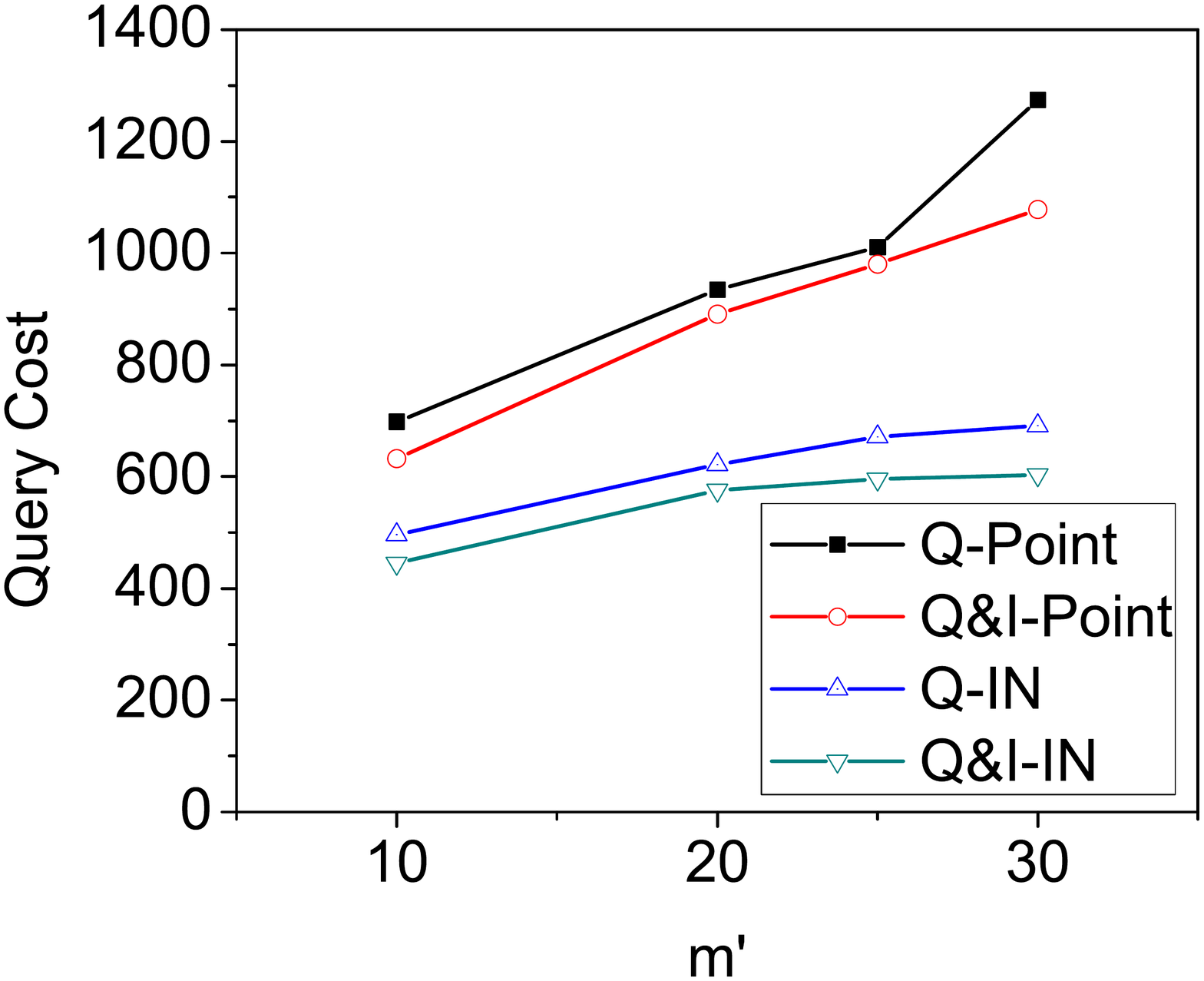}
\vspace{-7mm}\caption{Varying $m^\prime$}
\label{fig:eh_mPrimeVsQueryCost}
\end{minipage}
\hspace{1mm}
\begin{minipage}[t]{0.23\linewidth}
\centering
\includegraphics[width = 50mm, height = 32mm]{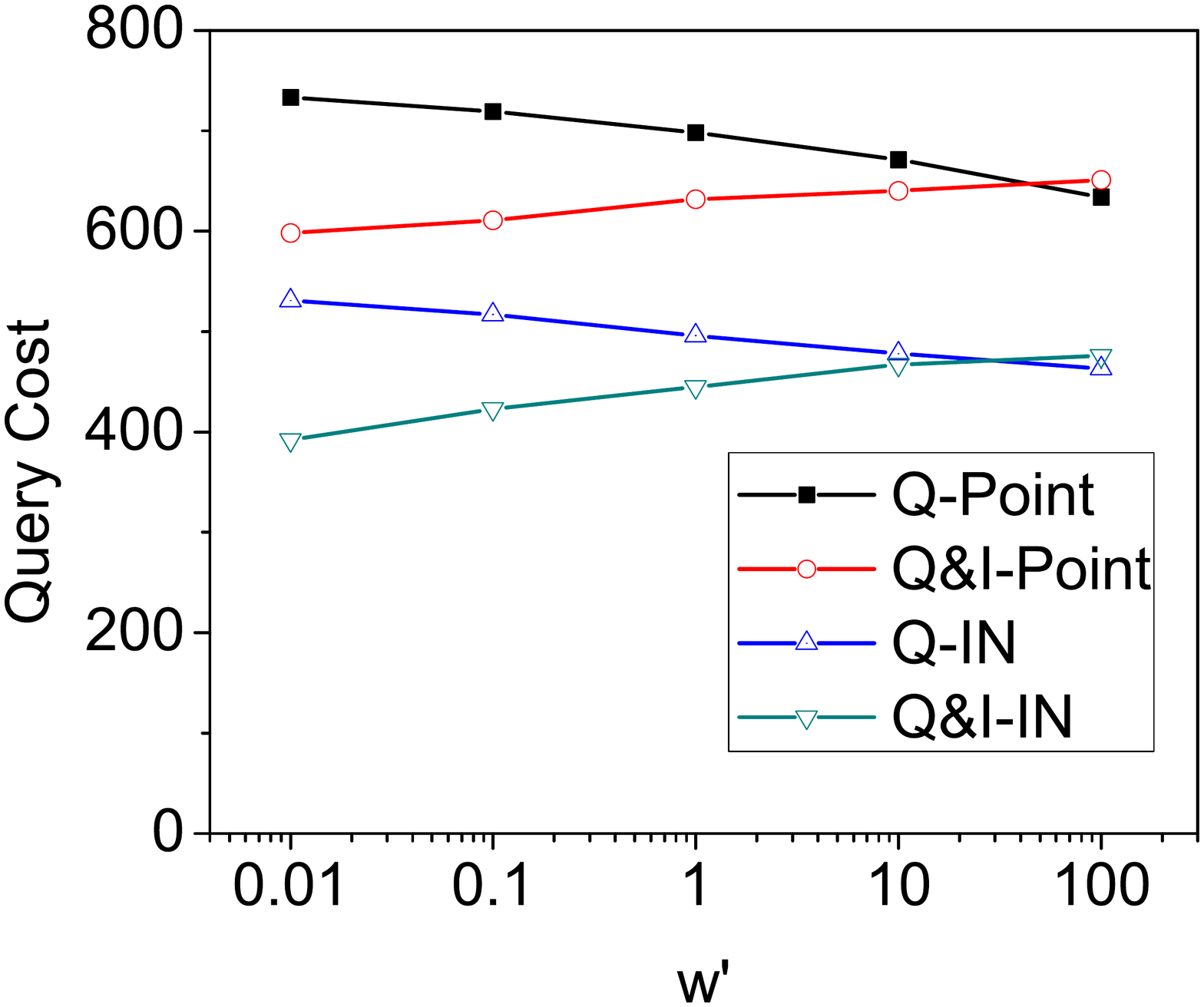}
\vspace{-7mm}\caption{Varying $w^{\prime}_1$}
\label{fig:eh_wPrimeVsQueryCost}
\end{minipage}
\hspace{1mm}
\begin{minipage}[t]{0.23\linewidth}
\centering
\includegraphics[width = 50mm, height = 32mm]{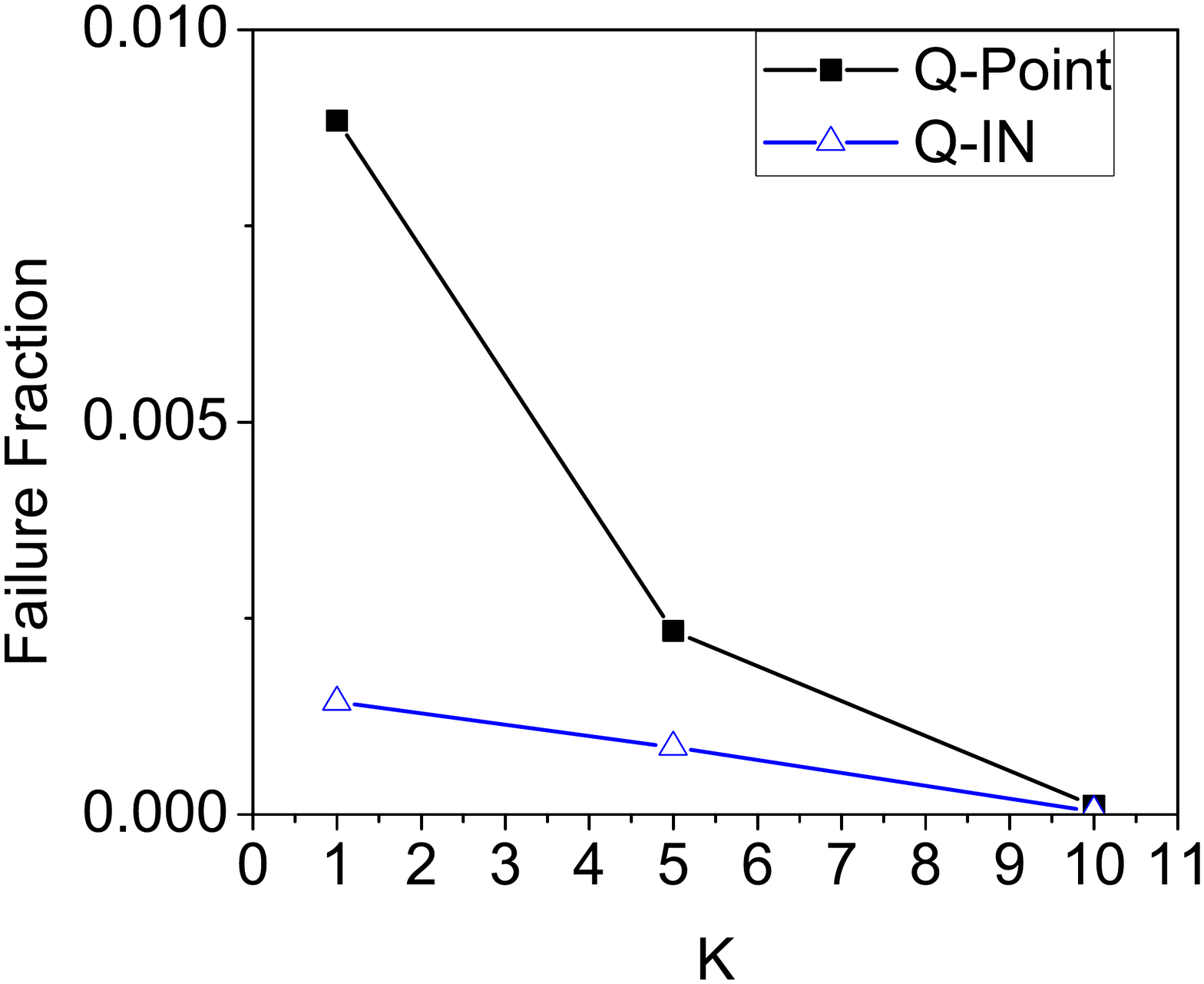}
\vspace{-7mm}\caption{Fraction of Uncompromised Accounts }
\label{fig:eh_fracSuccQOnly}
\end{minipage}
\hspace{1mm}
\begin{minipage}[t]{0.23\linewidth}
\centering
\includegraphics[width = 50mm, height = 32mm]{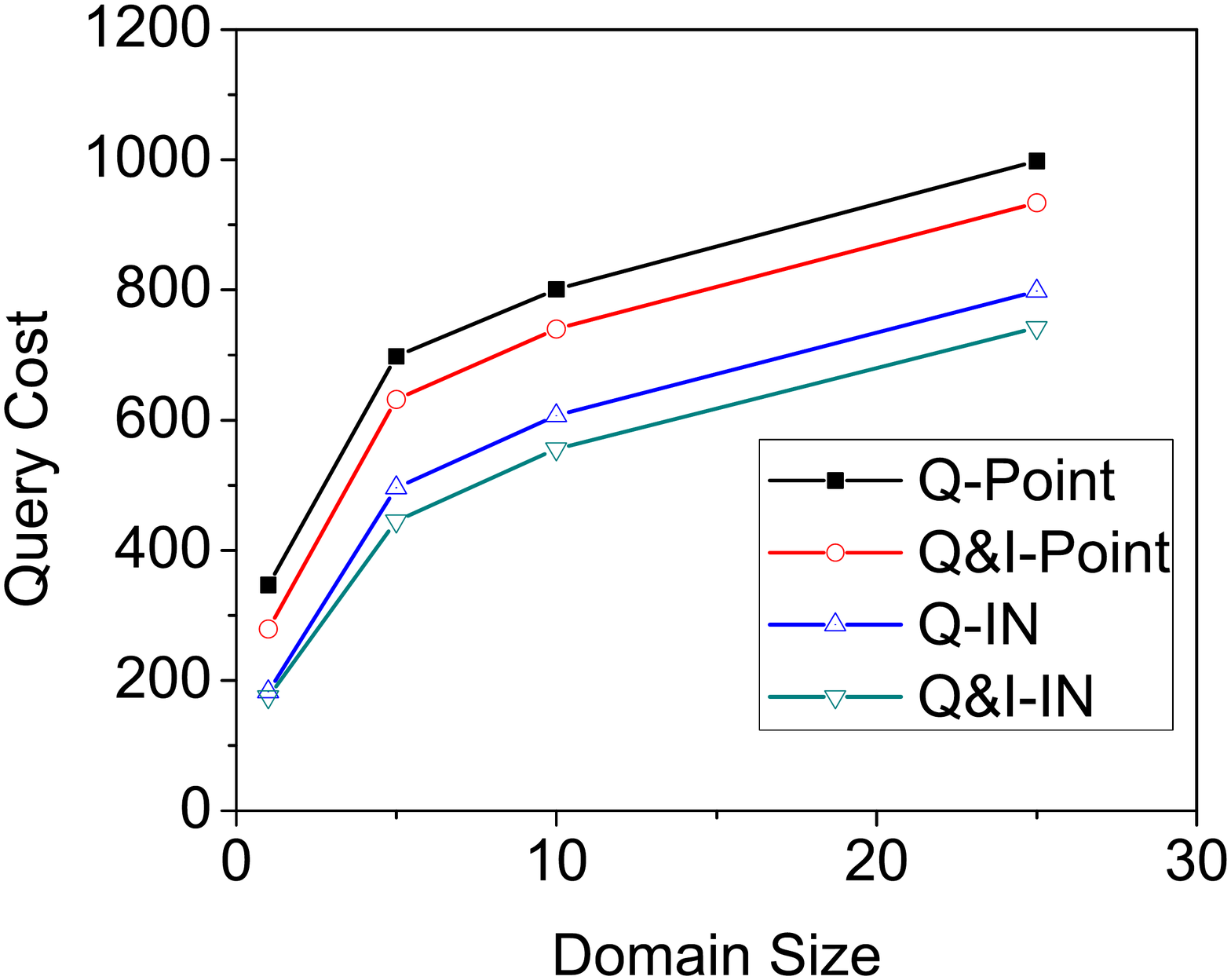}
\vspace{-7mm}\caption{Varying Domain Size of Inferred Attribute} 
\label{fig:eh_domainSizeVsQueryCost}
\end{minipage}
\hspace{-2mm}
\end{figure*}

\section{Experimental Results} \label{sec:exp}
\subsection{Experimental Setup}
\noindent{\bf Hardware and Platform:} 
All our experiments were performed on a quad-core 2 GHz AMD Phenom machine running Ubuntu 14.04 with 8 GB of RAM. 
The algorithms were implemented in Python.

\vspace{1mm}
\noindent {\bf Offline Datasets:} 
To verify the correctness of our results, we started by testing our algorithms locally over two real-world and two synthetic datasets. We have full access to these datasets, along with full control of the ranking function used. One dataset is from {\em eHarmony} \cite{eHarmony}, a prominent online dating service \cite{mcfee2010metric} and consists of anonymized profile information of 500K users. Each user has 56 attributes, of which more than 30 are boolean. The second dataset is {\em Yahoo!~Autos}, which contains 200K used cars for sale in the Dallas-Fort Worth area with 32 Boolean attributes and 6 categorical attributes, the domain cardinalities of which vary from 5 to 447. The third dataset is a synthetic Boolean i.i.d.~dataset with 200K tuples and 40 attributes, each following the uniform distribution.
The final synthetic dataset has 50 attributes with an average domain size of 150 with a Zipfian distribution ($z=2$).

The public and private attributes were randomly chosen from the set of available attributes. By default, we randomly picked 20 attributes for testing, designated $m=10$ of them as public and $m'=10$ as private, while varying $m$ and $m^\prime$ between 10 and 30 in various tests. Target attribute $B_1$ was chosen uniformly at random from all private attributes. Unless otherwise specified, we used the ranking function from (\ref{equ:scr}) with all weights set to 1.

\vspace{1mm}
\noindent {\bf  Online Demonstration:} In order to demonstrate the success of our attacks over real-world websites, we selected three high-profile real-world websites - Renren.com, Amazon Goodreads, Catch22Dating - and conducted live experiments using our algorithms. We would like to note that, without a partnership with these websites, we do not possess/assume any knowledge of their ranking function (other than the monotonicity and additivity properties defined in \S\ref{sec:pre}, which we verified through the correctness of our experiment outputs). The results of the online experiments can be found in \S\ref{subsec:expOnline}.

\vspace{1mm}
\noindent {\bf Performance Measures:} As discussed earlier in \S\ref{sec:spa}, we measure efficiency through query cost, i.e., the number of queries required for each attack - consistent with prior work\cite{dasgupta2007random,dasgupta2010unbiased}.

\subsection{Experiments over Real-World Datasets}
\label{subsec:expOffline}

\noindent {\bf Empirical Evaluation of Attack Success Rate:}
Figure~\ref{fig:attackSuccessRate} shows the attack success rate of all our algorithms over 3 offline datasets and the relevant algorithm (based on the problem subspace the website falls) over 3 online datasets. As expected, Q\&I-adversary has 100\% success rate for all datasets. For Q-only adversaries over real-world datasets, we were able to achieve a success rate of almost 100\%. The same holds for online tests except CD - the main reason here is that CD allows NULL value on the private attribute we are targeting, leading to failed attacks.

In the following discussion of offline experiments, we focus on results over eHarmony. 
The results on Yahoo!~Autos and the synthetic dataset were largely similar and detailed at the end of this section.

\vspace{1mm}
\noindent {\bf Query Cost versus $k$:} 
We first investigated the performance of our algorithms for different values of $k$.
Figure~\ref{fig:eh_kVsQueryCost} shows that query cost decreases with higher values of $k$ as expected.
Extending our algorithms for $k > 1$ is straightforward.
First, we seek to find a query that returns $v$ in top-$k$ (not just top-$1$).
Second, we extend the notion of differential queries (see \S\ref{sec:eqp}) such that the $v$ has a higher rank for query $q^{\prime}_{\theta}$ than for $q_{\theta}$. 
The query cost of our algorithms can be broadly categorized into two parts -
the query cost to identify a query $q$ that returns the victim tuple
and the query cost required to construct additional queries from $q$ through which the private attribute is inferred.
When the value of $k$ increases, the former query cost falls dramatically.
Further, the figure also shows that when IN-queries are available (Q\&I-IN and Q-IN), the query cost is lower than the cases where only point queries are allowed (Q\&I-Point and Q-Point), consistent with our discussions in \S\ref{sec:rq}.

\vspace{1mm}
\noindent {\bf Query Cost versus Database Size, $n$:}
Figure~\ref{fig:eh_NVsQueryCost} depicts the impact of database size on query cost when $k = 1$ (which is henceforth used as the default setting unless otherwise specified).
As expected, the increase in database size do not have any major impact and only results in a slight increase in overall query cost.
This is due to the fact that the number of queries needed to identify a randomly chosen tuple increases 
much more slowly than the database size.

\vspace{1mm}
\noindent {\bf Query Cost versus $m, m'$:}
In our next experiments, we investigate how varying the number of public and private attributes affect the query cost.
The results of these experiments are shown in Figures~\ref{fig:eh_mVsQueryCost} and \ref{fig:eh_mPrimeVsQueryCost}.
As expected, when the number of public attributes increase, the query cost drops significantly.
When the number of public attributes are limited, their values are not adequate to distinctly identify a random tuple.
Hence, we need to resort to using randomly chosen values for the private attributes which increases query cost.
However, when $m$ increases, most tuples become uniquely identified based on their public attributes only.
For a fixed $m$, the query cost increases with increasing $m'$ - 
 when the public attributes are inadequate for uniquely identifying the victim tuple,
our algorithms resort to issuing queries where the private attributes are chosen randomly from their respective value domains. But the number of such possible queries increases with higher $m'$ - hence the phenomenon.

\vspace{1mm}
\noindent {\bf Query Cost versus Ranking Weights:}
In this experiment, we fixed the weight of all public attributes to 1 and 
varied the weights of private attributes $w^\prime_i$ between $0.01$ and $100$.
The results shown in Figure~\ref{fig:eh_wPrimeVsQueryCost} are consistent with our theoretical results
from Sections~\ref{sec:pq} and \ref{sec:rq}.
When the weights over private attributes decrease, the query cost for Q\&I adversaries also decreases. 
This is due to the fact that identifying the query $q$ that returns the victim tuple $v$ becomes much easier for this case.
The opposite holds for Q-only adversaries where increasing the weights decreases the query cost. 


\vspace{1mm}
\noindent {\bf Other Experiments:}
In order to identify the fraction of tuples in a database that could be successfully compromised using our algorithms, 
we randomly chose 100K tuples and tried to compromise them.
Recall that the Q\&I adversary based algorithms are always guaranteed to succeed.
Figure~\ref{fig:eh_fracSuccQOnly} shows that the Q-only algorithms are able to compromise almost all the tuples.
Even with a highly restrictive interface of $k=1$, Q-Point compromises more than 99\% of the tuples. We then adapted our inference algorithms so that they seek to infer all $m^{\prime}=25$ private attributes.
Figure~\ref{fig:eh_kVsQueryCostInferAll} shows the result.
While the overall query cost seems high, the amortized query per private attribute varies between $35$ and $60$.
Figure~\ref{fig:eh_domainSizeVsQueryCost} shows how varying the domain size of the private attribute affects the query cost.
Consistent with our analysis query cost increases with larger domain size.

\begin{figure}[ht]
\centering
\begin{minipage}[h]{0.48\linewidth}
\centering
\includegraphics[width = 50mm, height = 32mm]{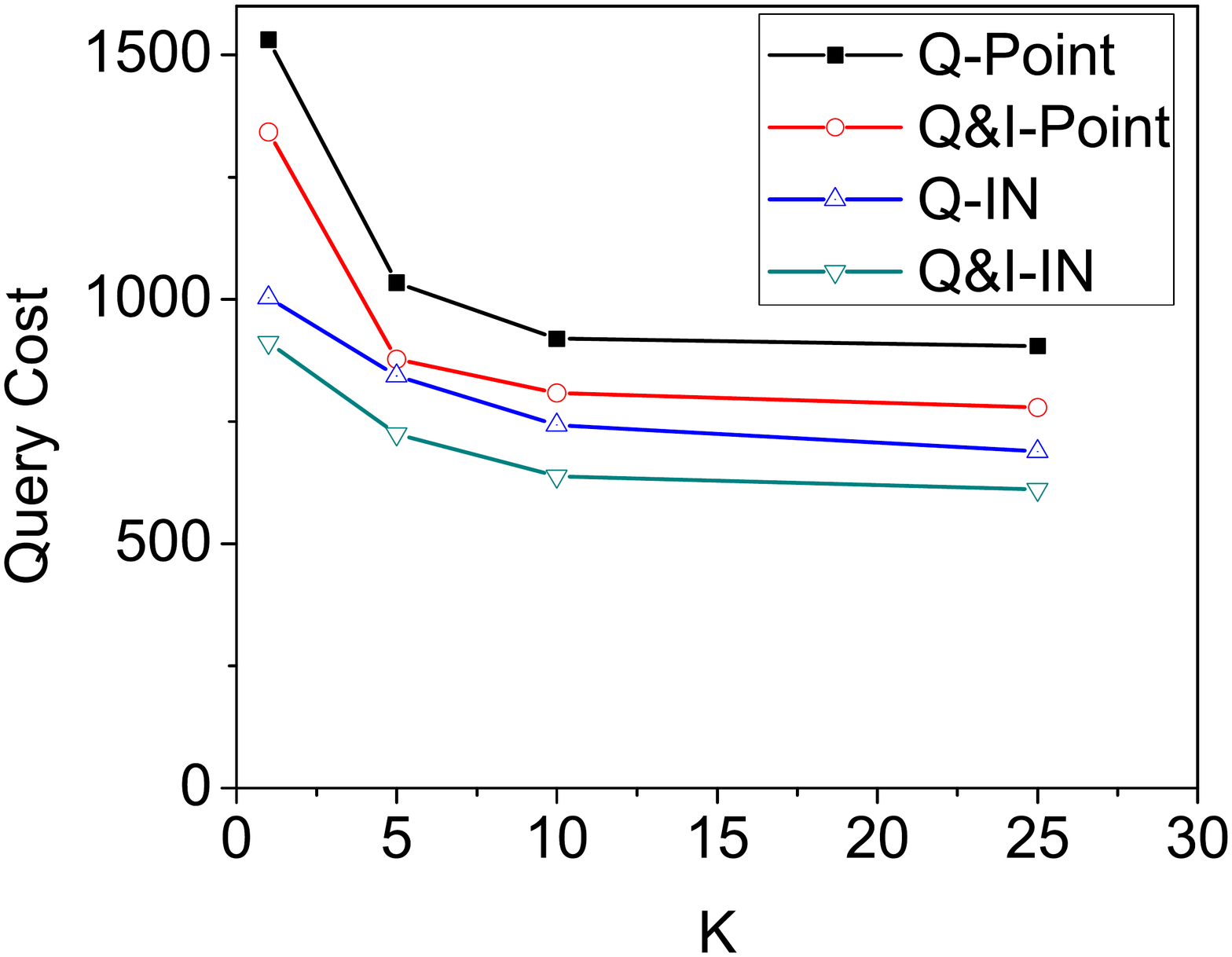}
\vspace{-7mm}\caption{Query Cost to Infer all Private Attributes}
\label{fig:eh_kVsQueryCostInferAll}
\end{minipage}
\hspace{1mm}
\begin{minipage}[h]{0.48\linewidth}
\includegraphics[width = 50mm, height = 32mm]{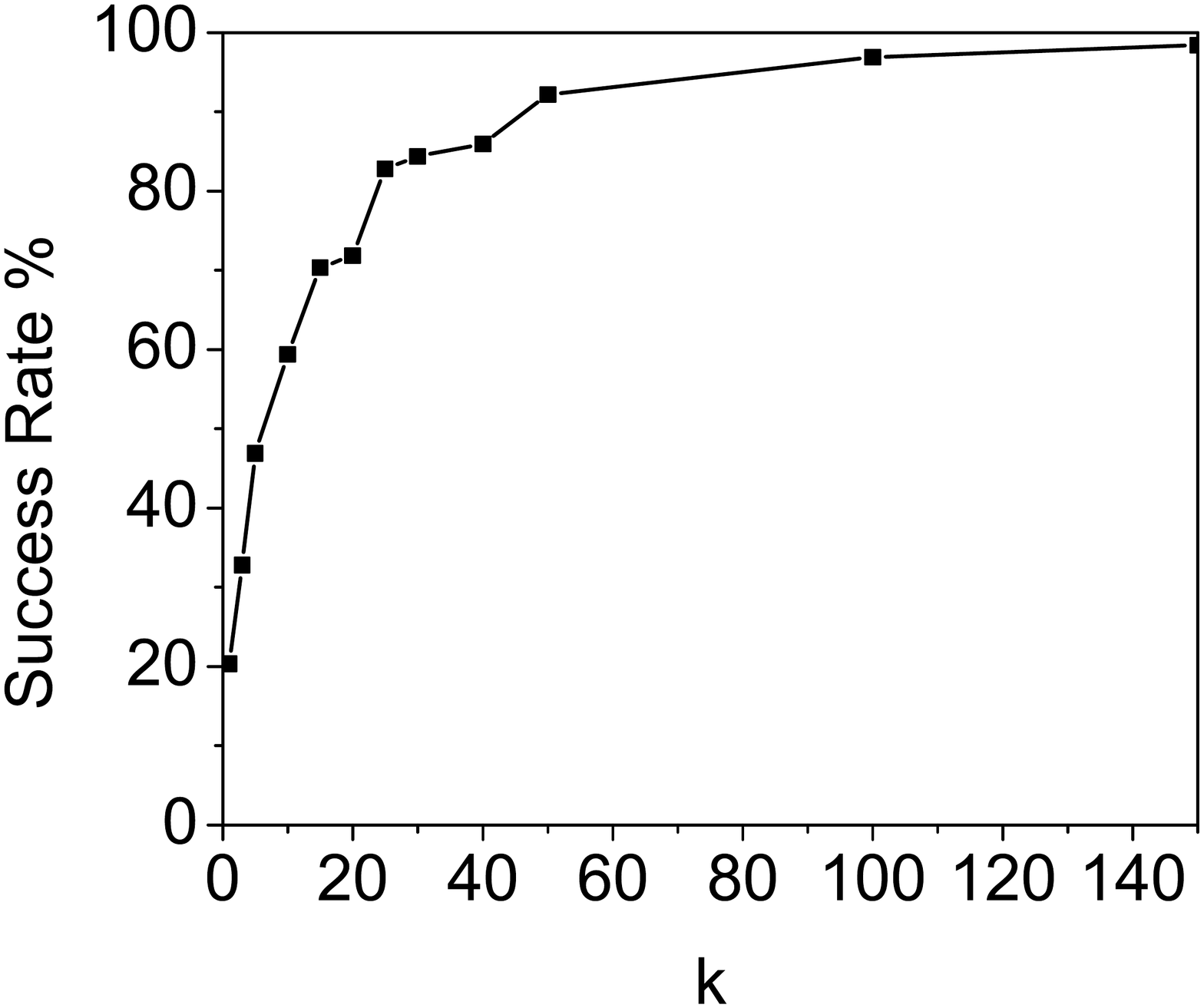}
\vspace{-7mm}\caption{Renren: Success Rate vs. $k$}
\label{fig:renrenKVsSuccRate}
\end{minipage}
\end{figure}

{\bf Experiments over Other Datasets:}
In addition to eHarmony dataset, we also conducted experiments over two other datasets, Yahoo Autos and BOOL-IID. 
The results of experiments over Yahoo Autos can be found in Figures~\ref{fig:ya_kVsQueryCost}-\ref{fig:ya_kVsQueryCostInferAll}.
The results for BOOL-IID can be found in Figures~\ref{fig:bool_kVsQueryCost}-\ref{fig:bool_kVsQueryCostInferAll}.
Finally, the results for Zipfian can be found in Figures~\ref{fig:zipf_mVsQueryCost}-\ref{fig:zipf_domainSizeAllVsQueryCost}.
We can see that the results are very similar to that of eHarmony demonstrating the wide spread applicability of our algorithms over diverse datasets.

\begin{figure*}[ht]
\begin{minipage}[t]{0.23\linewidth}
\centering
\includegraphics[width = 45mm, height = 32mm]{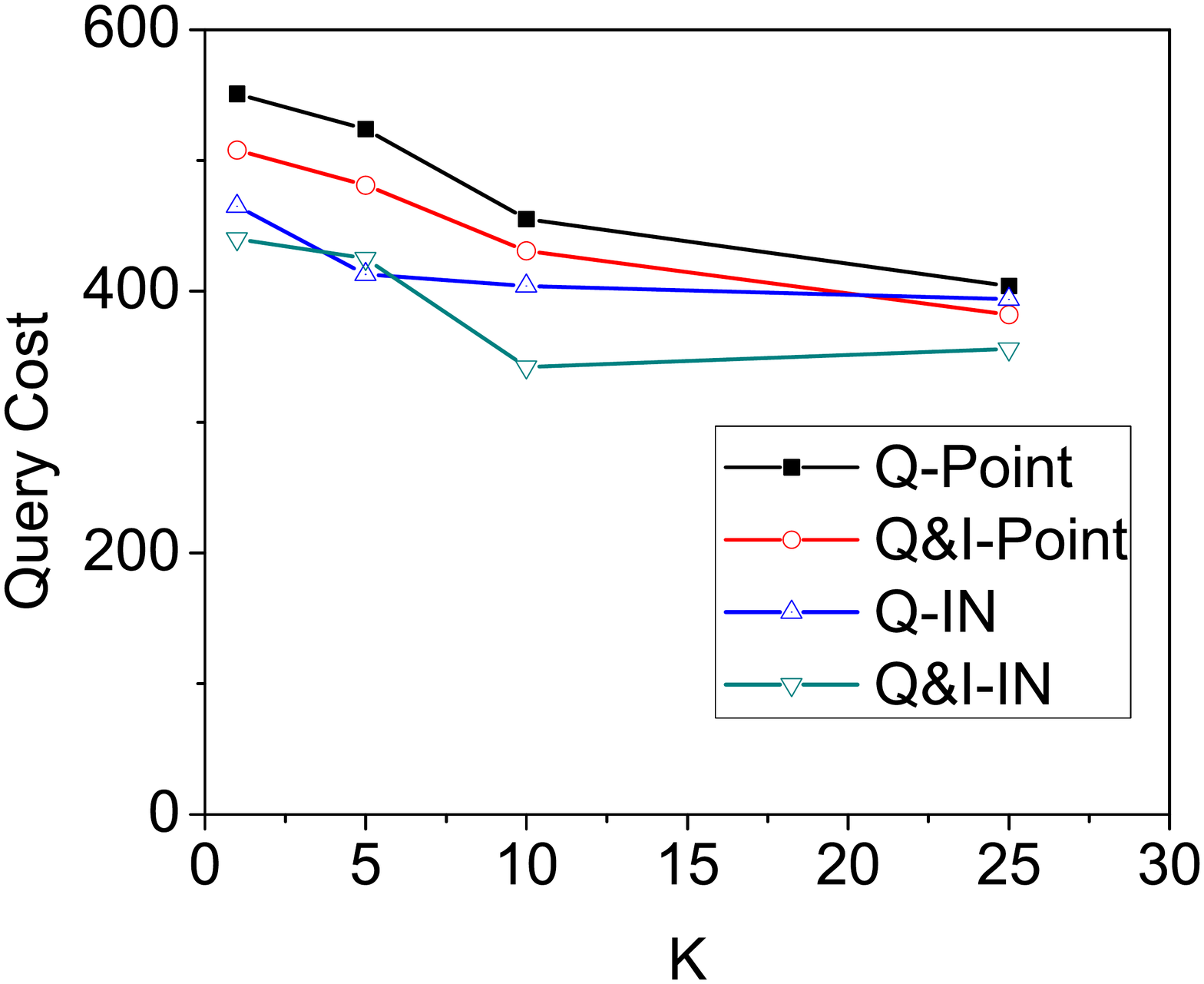}
\vspace{-7mm}\caption{Varying $k$ (YA)}
\label{fig:ya_kVsQueryCost}
\end{minipage}
\hspace{1mm}
\begin{minipage}[t]{0.23\linewidth}
\centering
\includegraphics[width = 45mm, height = 32mm]{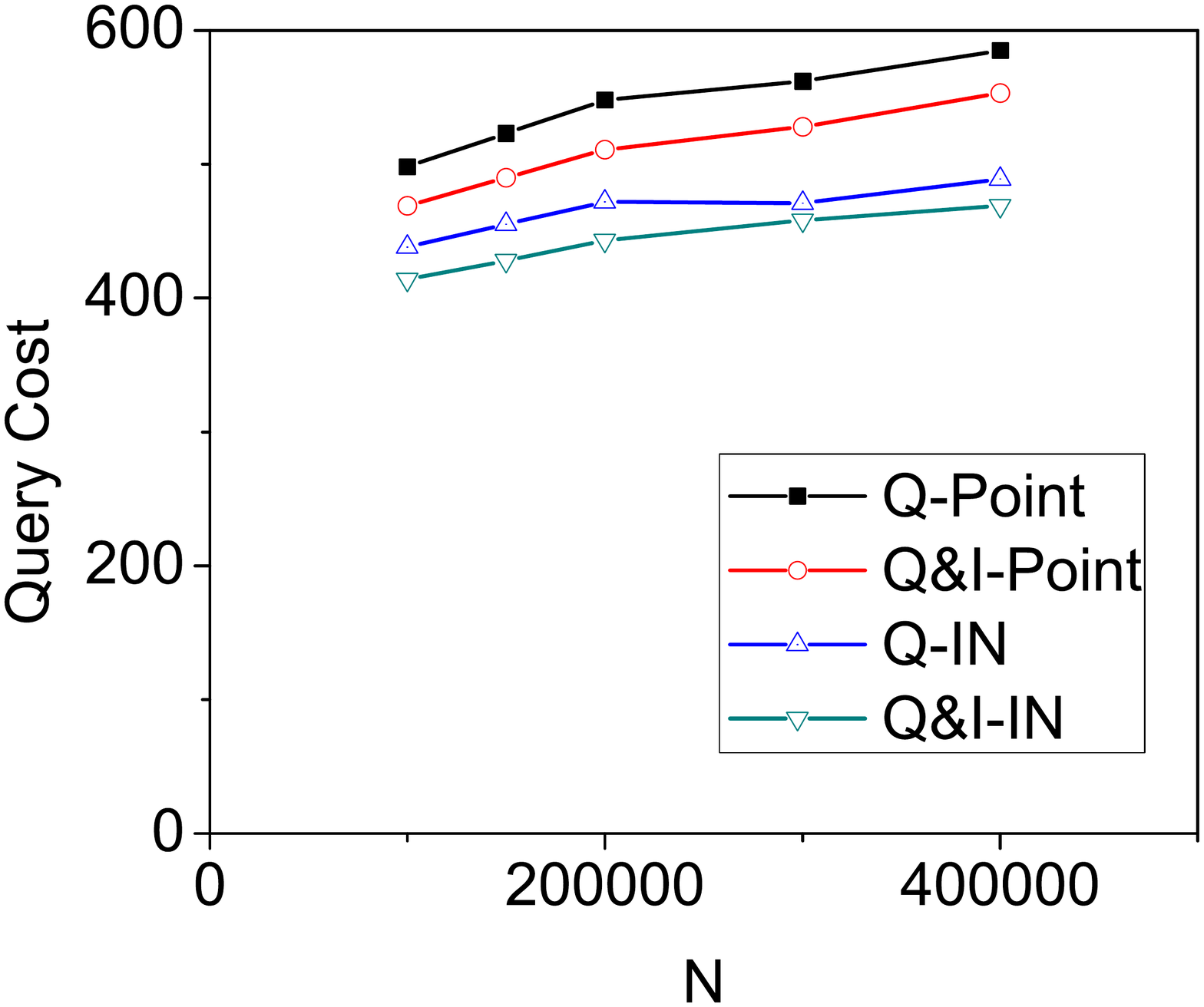}
\vspace{-7mm}\caption{Varying $n$ (YA)}
\label{fig:ya_NVsQueryCost}
\end{minipage}
\hspace{1mm}
\begin{minipage}[t]{0.23\linewidth}
\centering
\includegraphics[width = 45mm, height = 32mm]{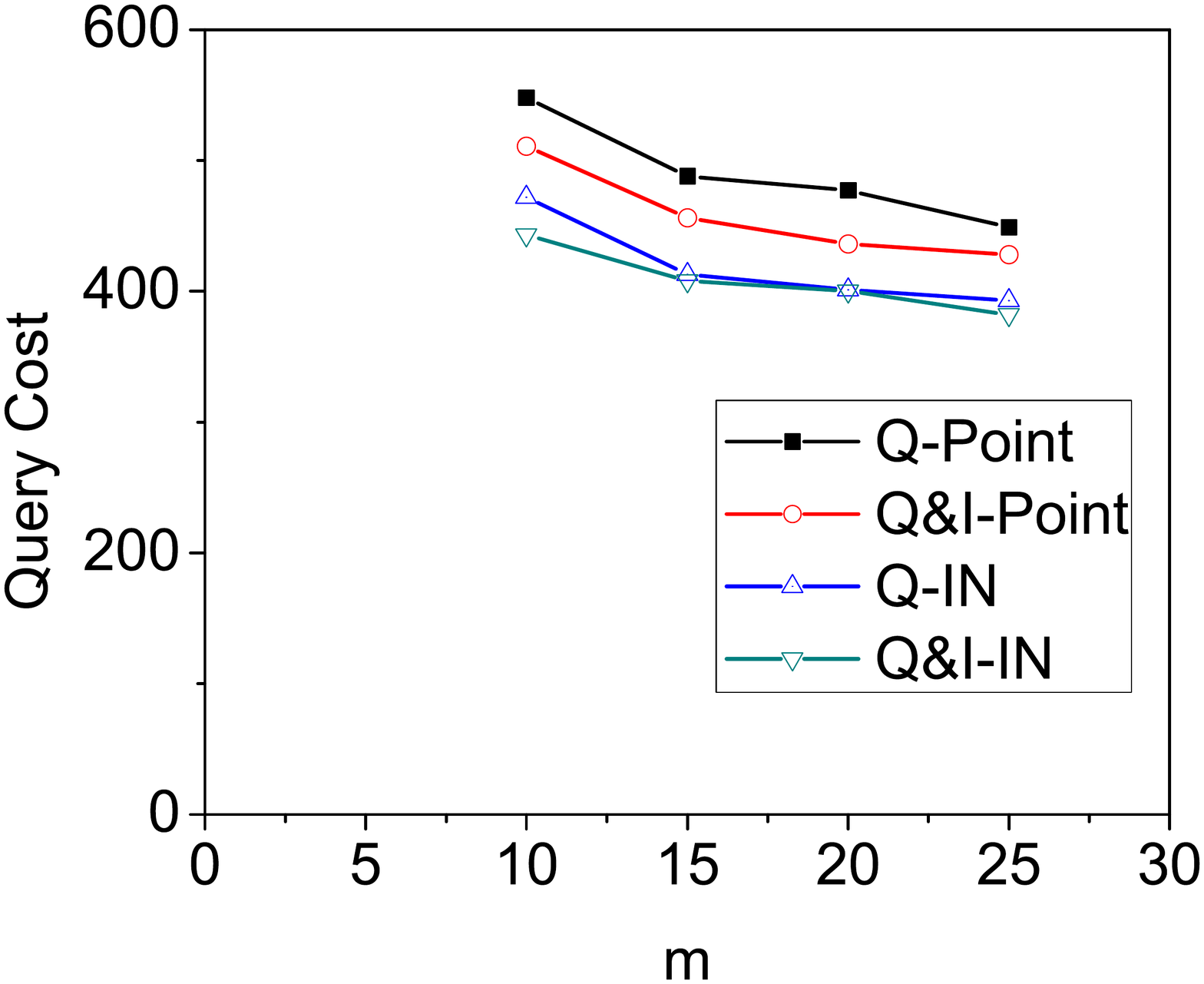}
\vspace{-7mm}\caption{Varying $m$ (YA)}
\label{fig:ya_mVsQueryCost}
\end{minipage}
\begin{minipage}[t]{0.23\linewidth}
\centering
\includegraphics[width = 50mm, height = 32mm]{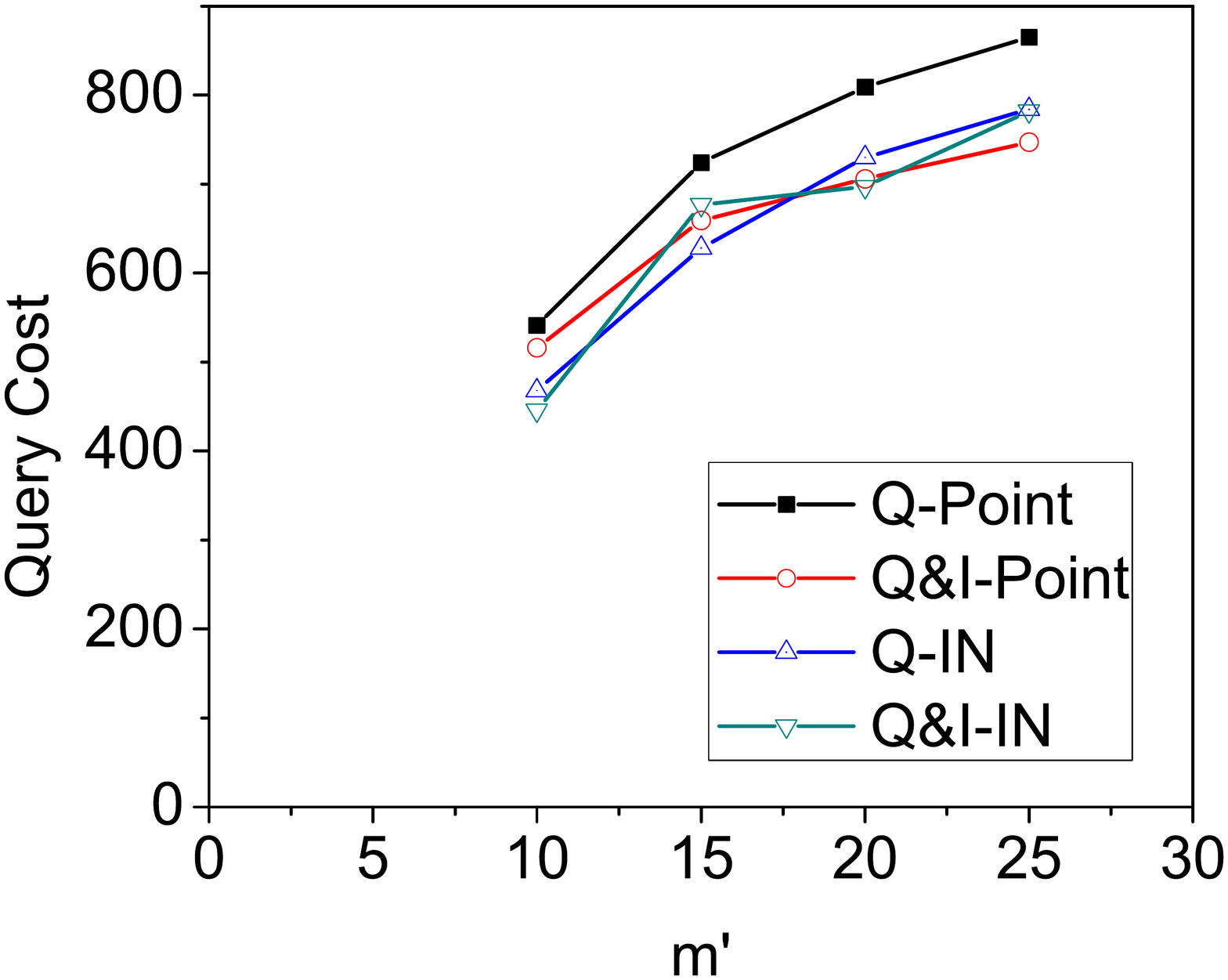}
\vspace{-7mm}\caption{Varying $m^\prime$ (YA)}
\label{fig:ya_mPrimeVsQueryCost}
\end{minipage}
\hspace{-2mm}
\end{figure*}

\begin{figure*}[t]
\begin{minipage}[t]{0.23\linewidth}
\centering
\includegraphics[width = 50mm, height = 32mm]{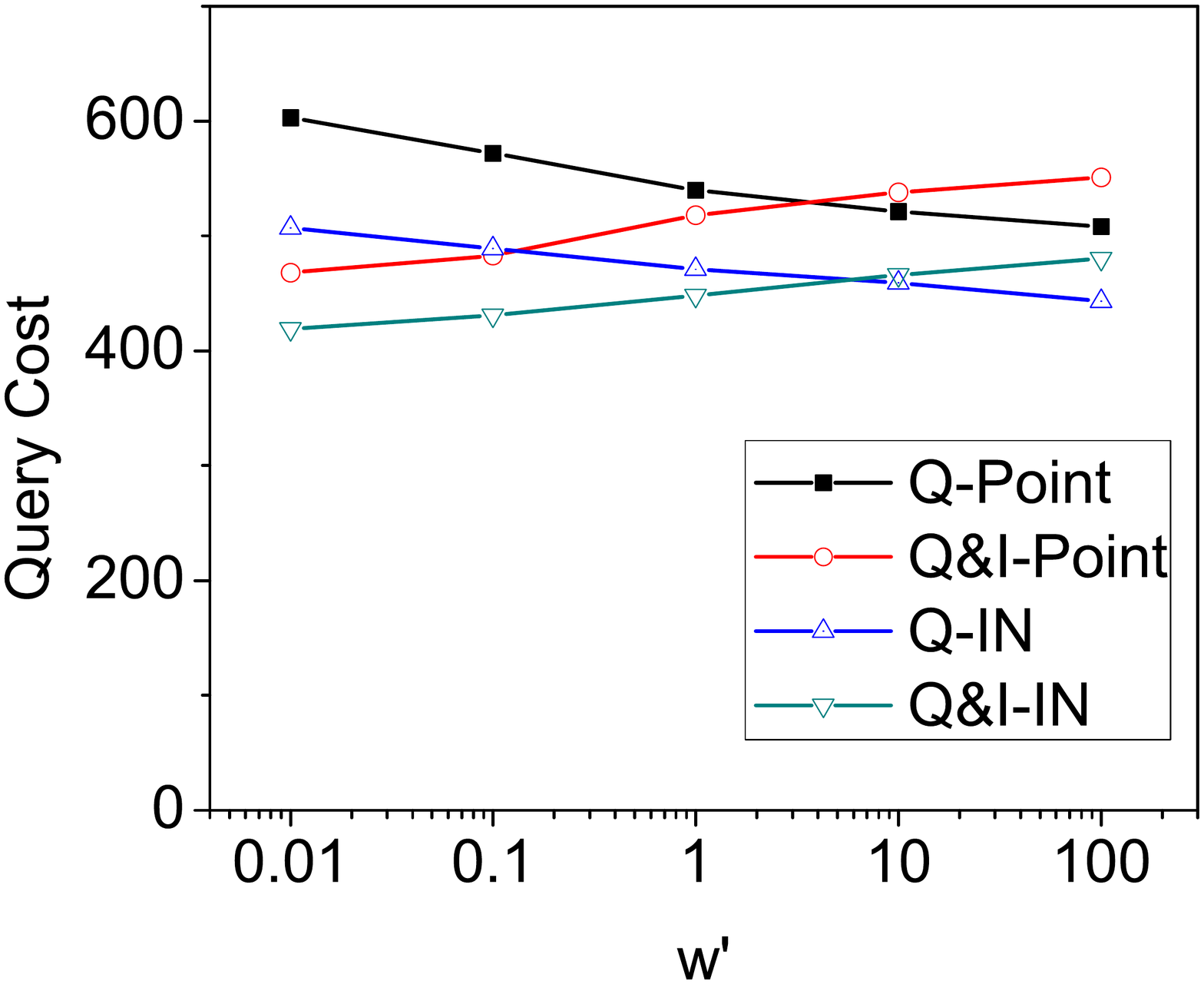}
\vspace{-7mm}\caption{Varying $w^{\prime}_1$ (YA)}
\label{fig:ya_wPrimeVsQueryCost}
\end{minipage}
\hspace{1mm}
\begin{minipage}[t]{0.23\linewidth}
\centering
\includegraphics[width = 50mm, height = 32mm]{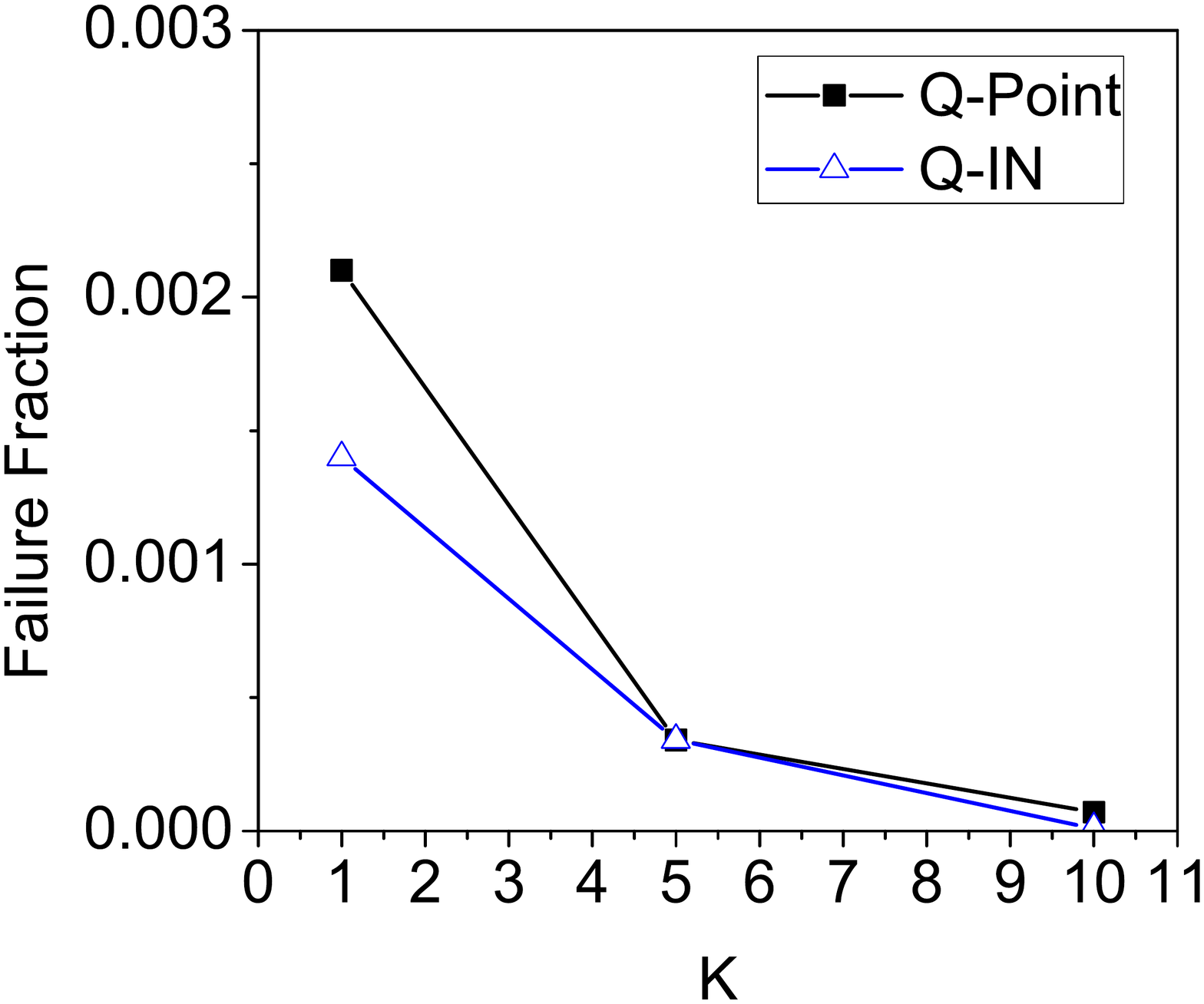}
\vspace{-7mm}\caption{Fraction of Uncompromised Accounts  (YA)}
\label{fig:ya_fracSuccQOnly}
\end{minipage}
\hspace{1mm}
\begin{minipage}[t]{0.23\linewidth}
\centering
\includegraphics[width = 50mm, height = 32mm]{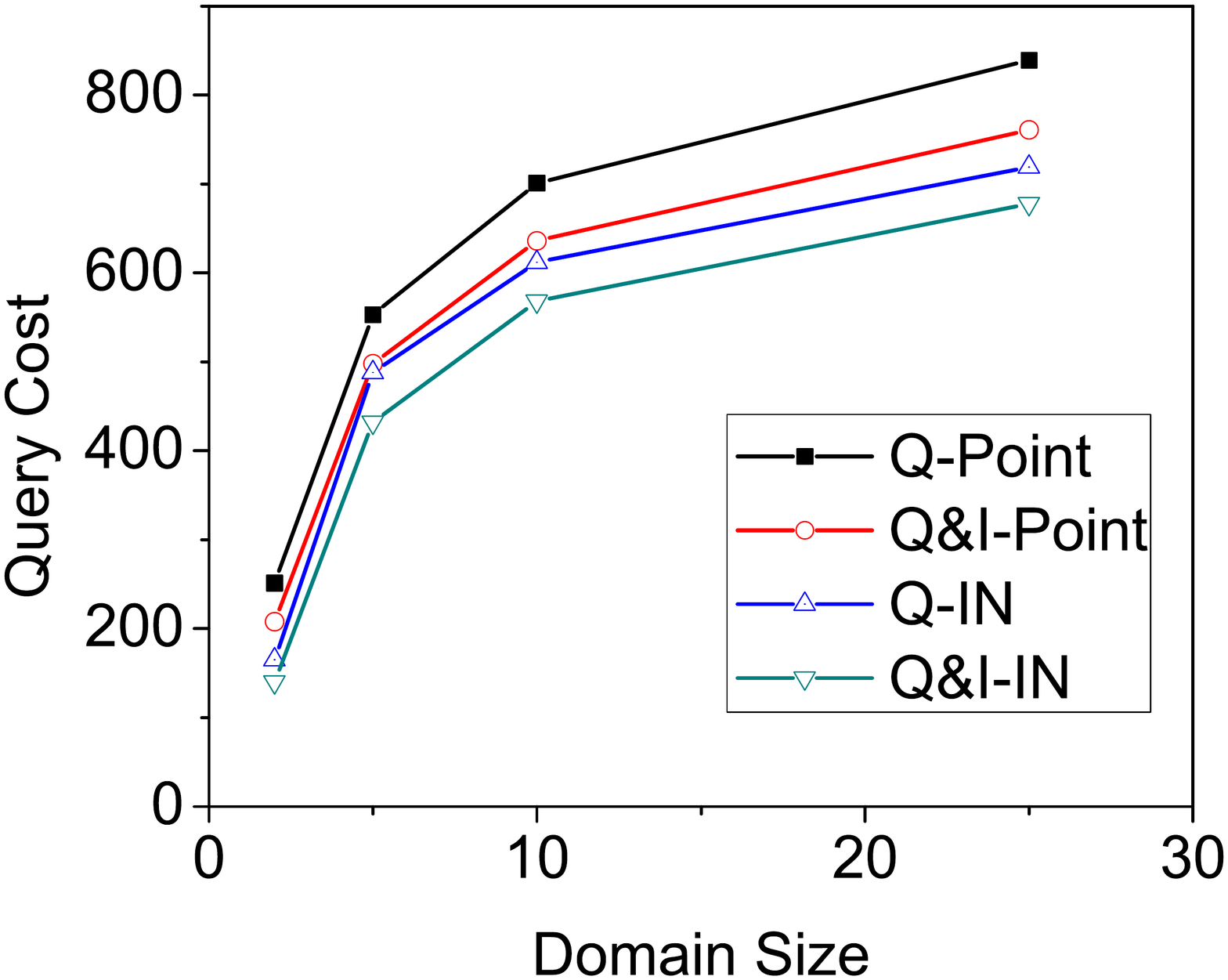}
\vspace{-7mm}\caption{Varying Domain Size of Inferred Attribute (YA)} 
\label{fig:ya_domainSizeVsQueryCost}
\end{minipage}
\hspace{1mm}
\begin{minipage}[t]{0.23\linewidth}
\centering
\includegraphics[width = 50mm, height = 32mm]{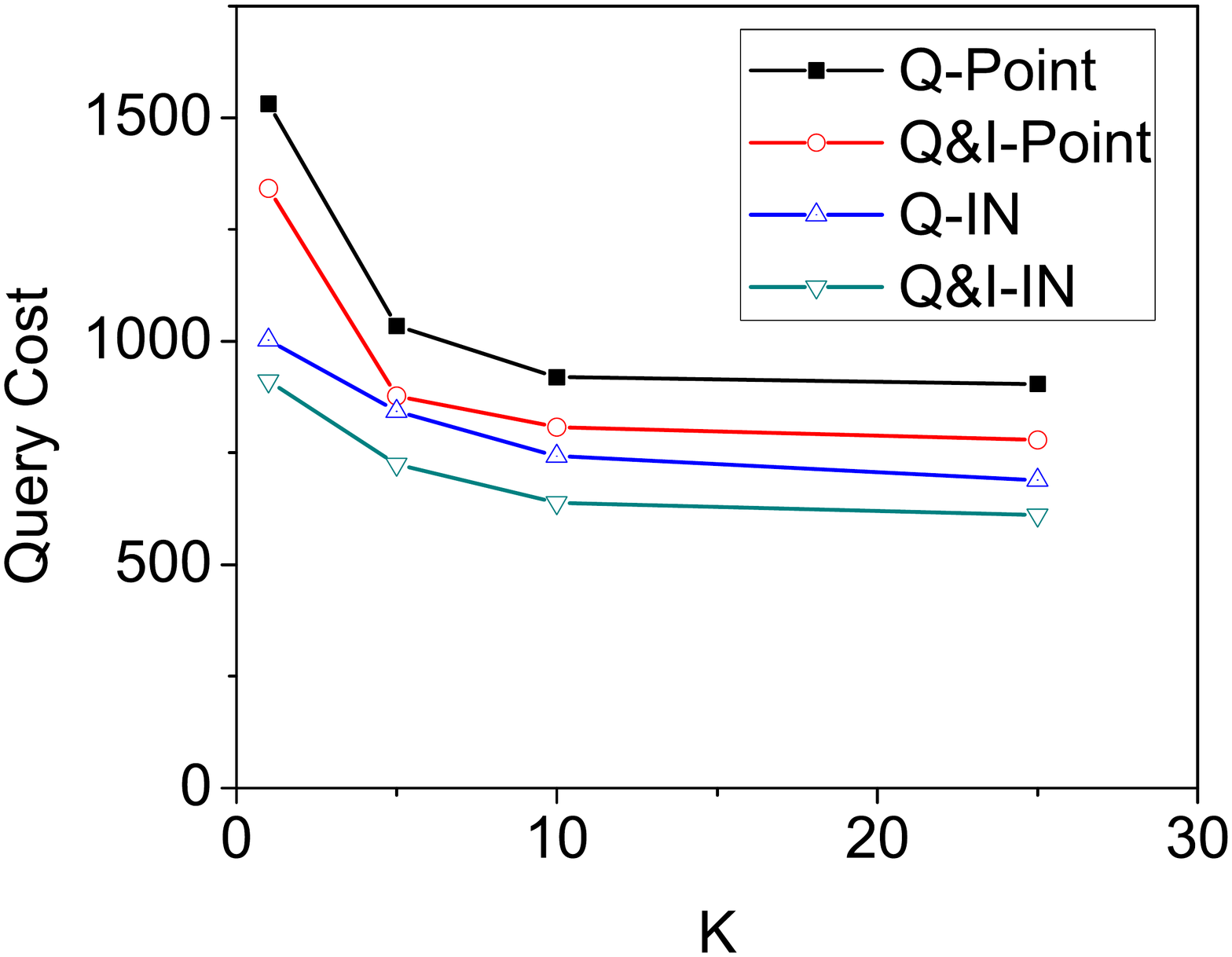}
\vspace{-7mm}\caption{Query Cost to Infer all Private Attributes (YA)}
\label{fig:ya_kVsQueryCostInferAll}
\end{minipage}
\hspace{-2mm}
\end{figure*}

\begin{figure*}[ht]
\begin{minipage}[t]{0.23\linewidth}
\centering
\includegraphics[width = 45mm, height = 32mm]{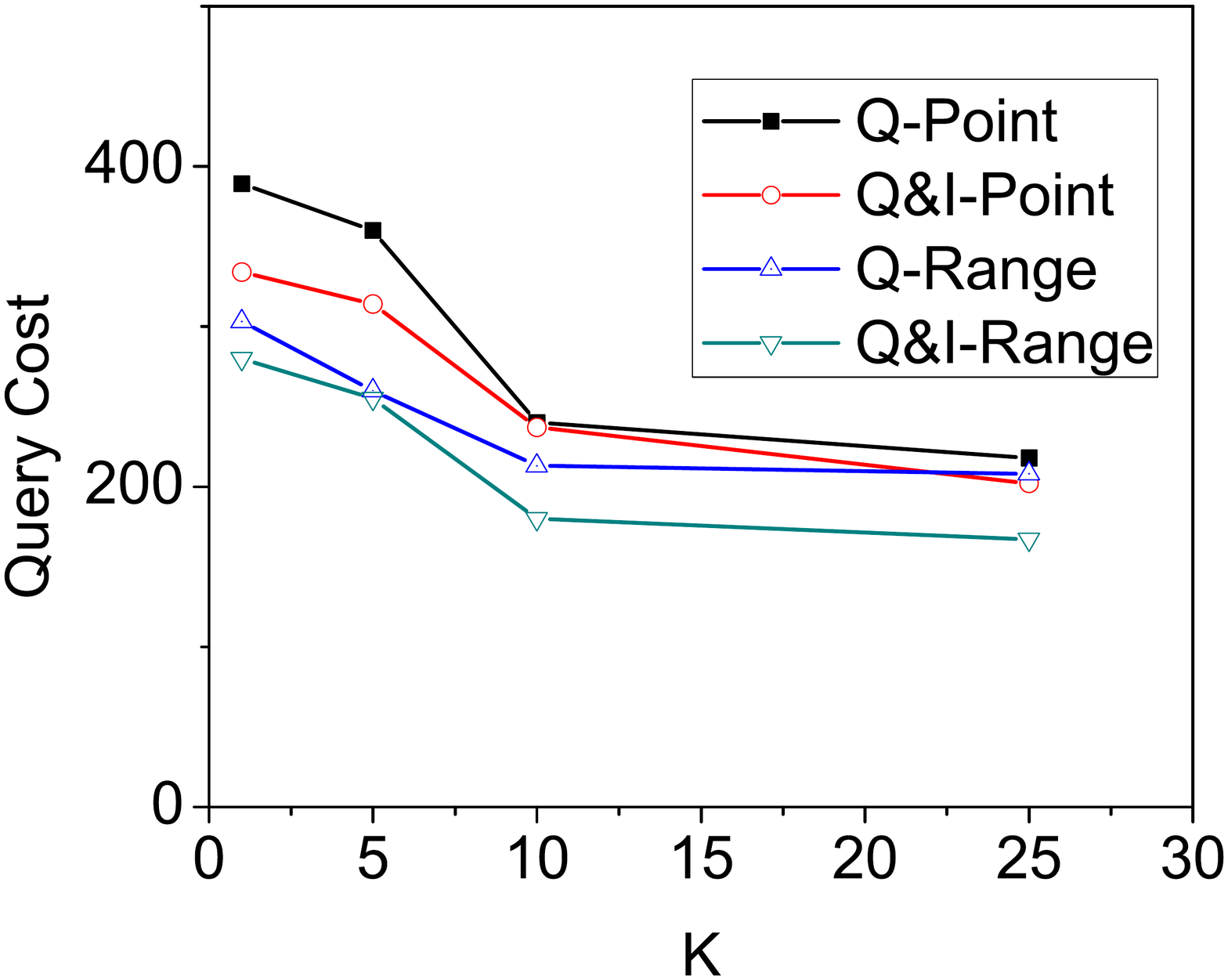}
\vspace{-7mm}\caption{Varying $k$ (BOOL)}
\label{fig:bool_kVsQueryCost}
\end{minipage}
\hspace{1mm}
\begin{minipage}[t]{0.23\linewidth}
\centering
\includegraphics[width = 45mm, height = 32mm]{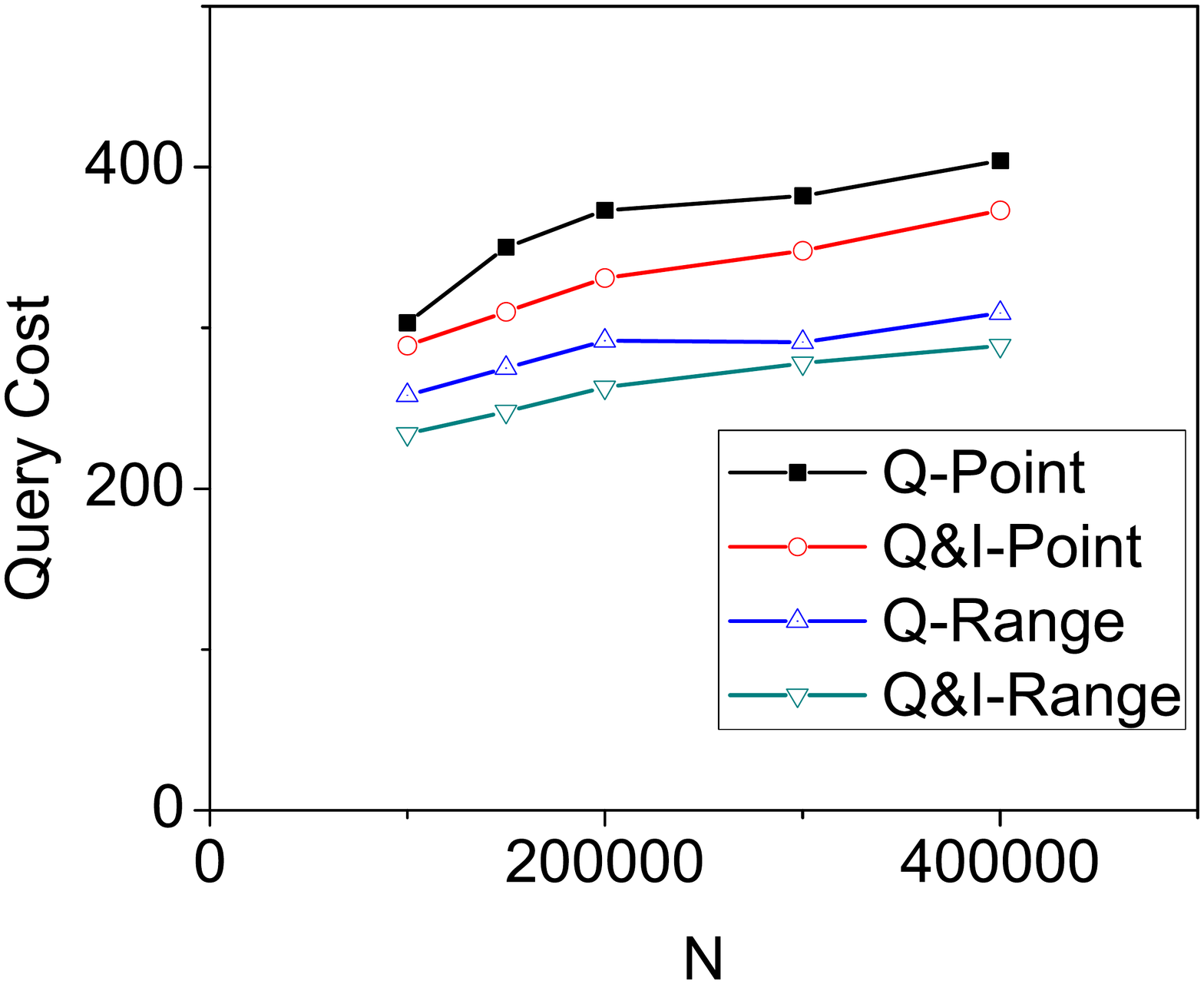}
\vspace{-7mm}\caption{Varying $n$ (BOOL)}
\label{fig:bool_NVsQueryCost}
\end{minipage}
\hspace{1mm}
\begin{minipage}[t]{0.23\linewidth}
\centering
\includegraphics[width = 45mm, height = 32mm]{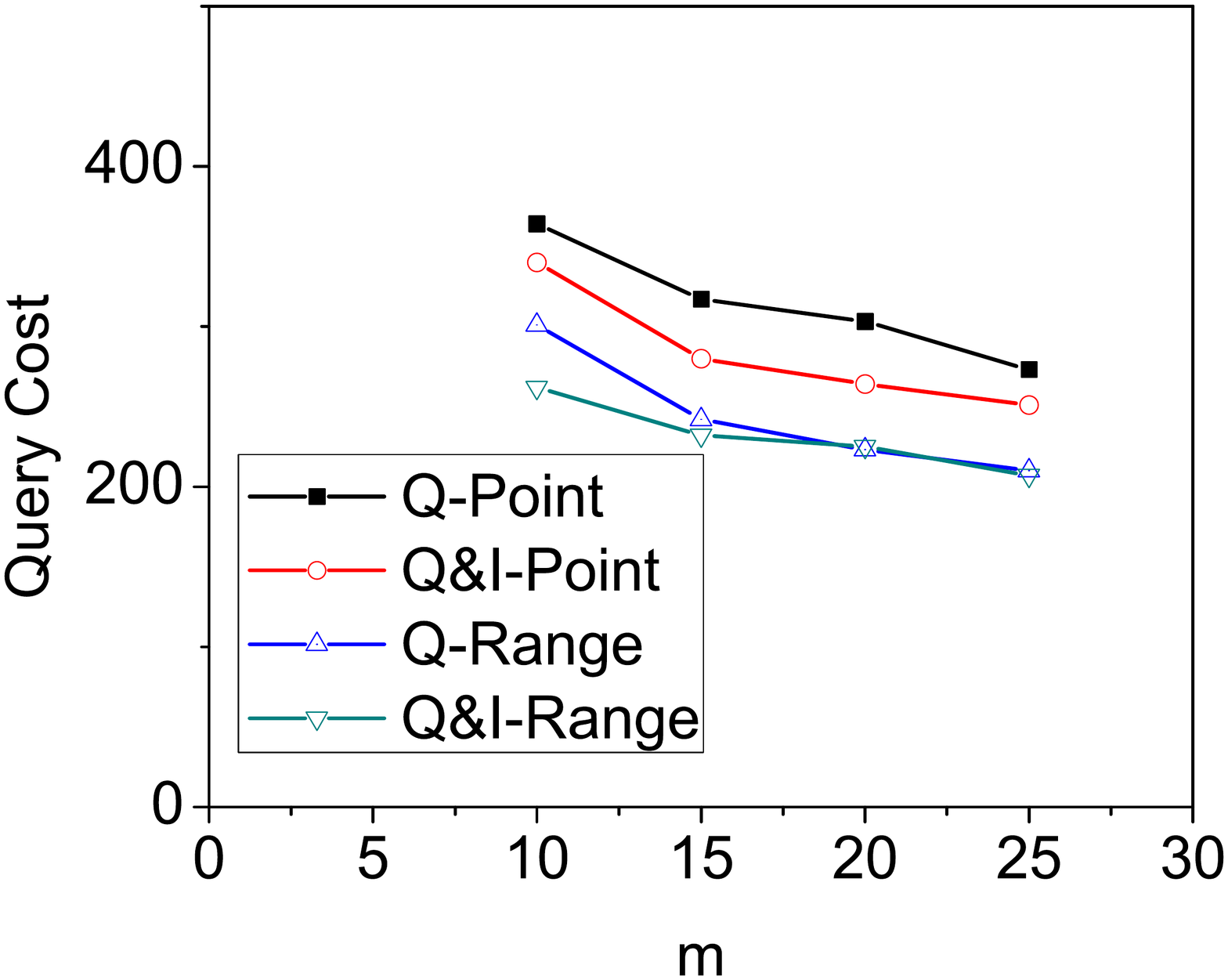}
\vspace{-7mm}\caption{Varying $m$ (BOOL)}
\label{fig:bool_mVsQueryCost}
\end{minipage}
\begin{minipage}[t]{0.23\linewidth}
\centering
\includegraphics[width = 50mm, height = 32mm]{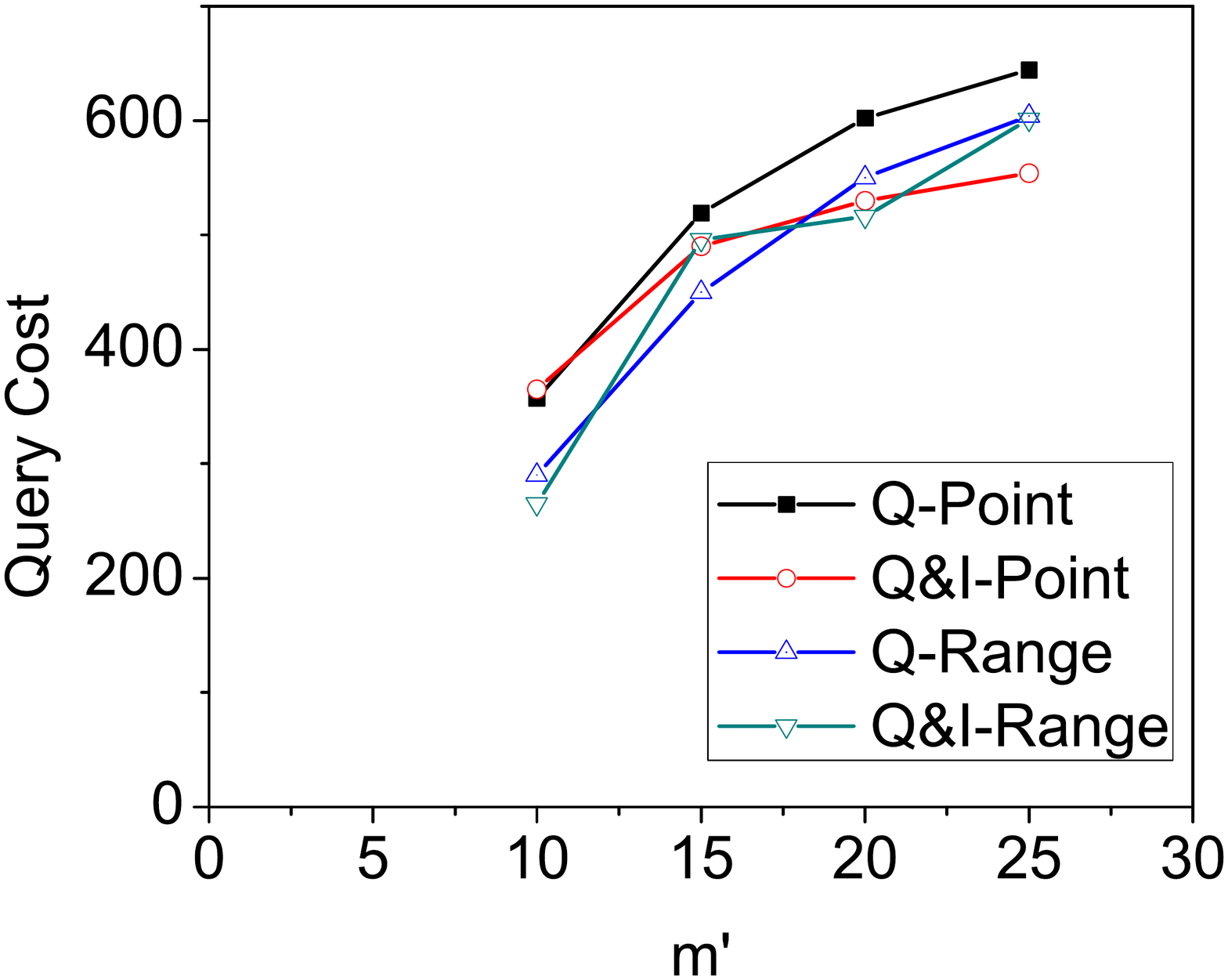}
\vspace{-7mm}\caption{Varying $m^\prime$ (BOOL)}
\label{fig:bool_mPrimeVsQueryCost}
\end{minipage}
\hspace{-2mm}
\end{figure*}

\begin{figure*}[t]
\begin{minipage}[t]{0.23\linewidth}
\centering
\includegraphics[width = 50mm, height = 32mm]{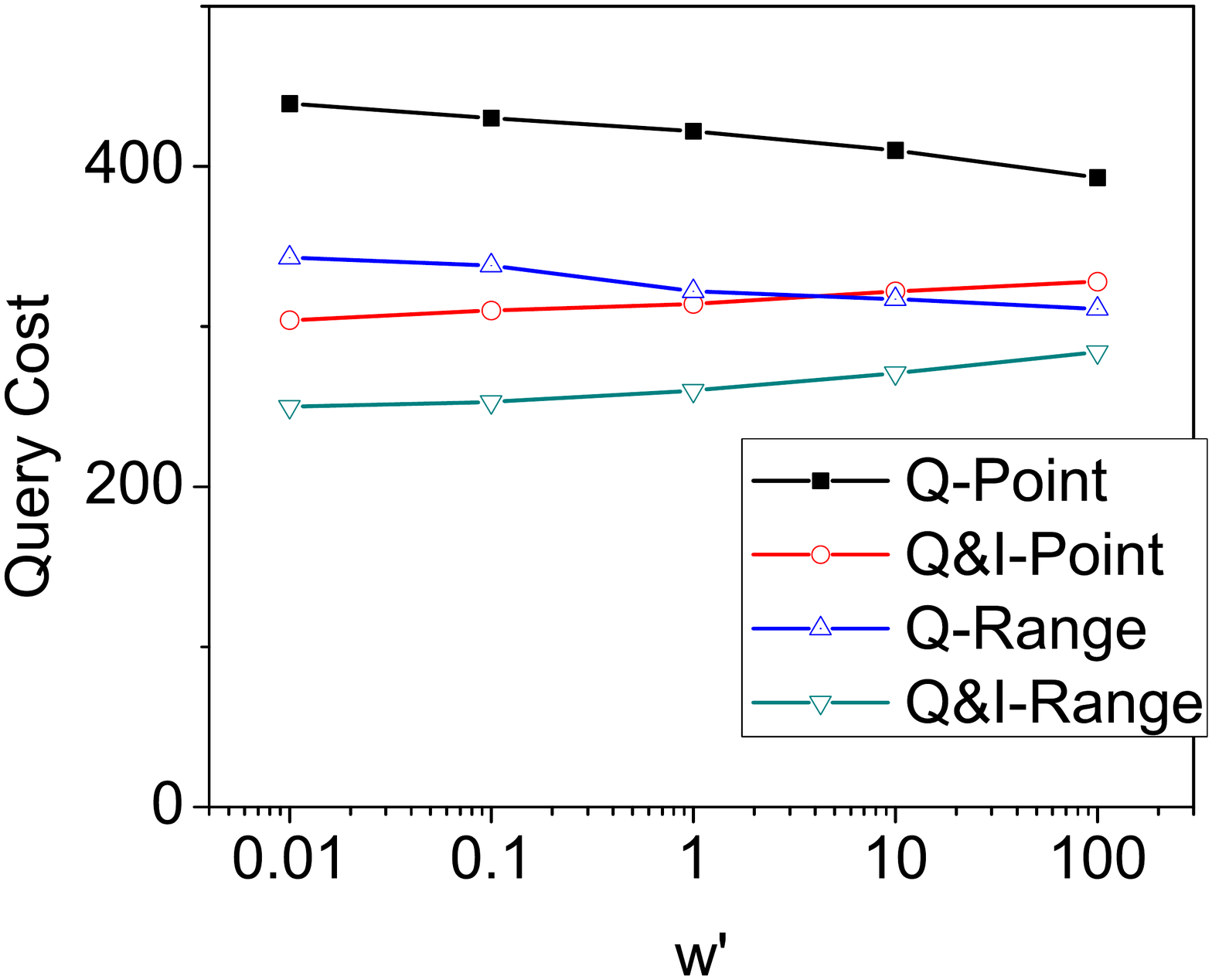}
\vspace{-7mm}\caption{Varying $w^{\prime}_1$ (BOOL)}
\label{fig:bool_wPrimeVsQueryCost}
\end{minipage}
\hspace{1mm}
\begin{minipage}[t]{0.23\linewidth}
\centering
\includegraphics[width = 50mm, height = 32mm]{figures/ya_fracSuccQOnly.pdf}
\vspace{-7mm}\caption{Fraction of Uncompromised Accounts  (BOOL)}
\label{fig:bool_fracSuccQOnly}
\end{minipage}
\hspace{1mm}
\begin{minipage}[t]{0.23\linewidth}
\centering
\includegraphics[width = 50mm, height = 32mm]{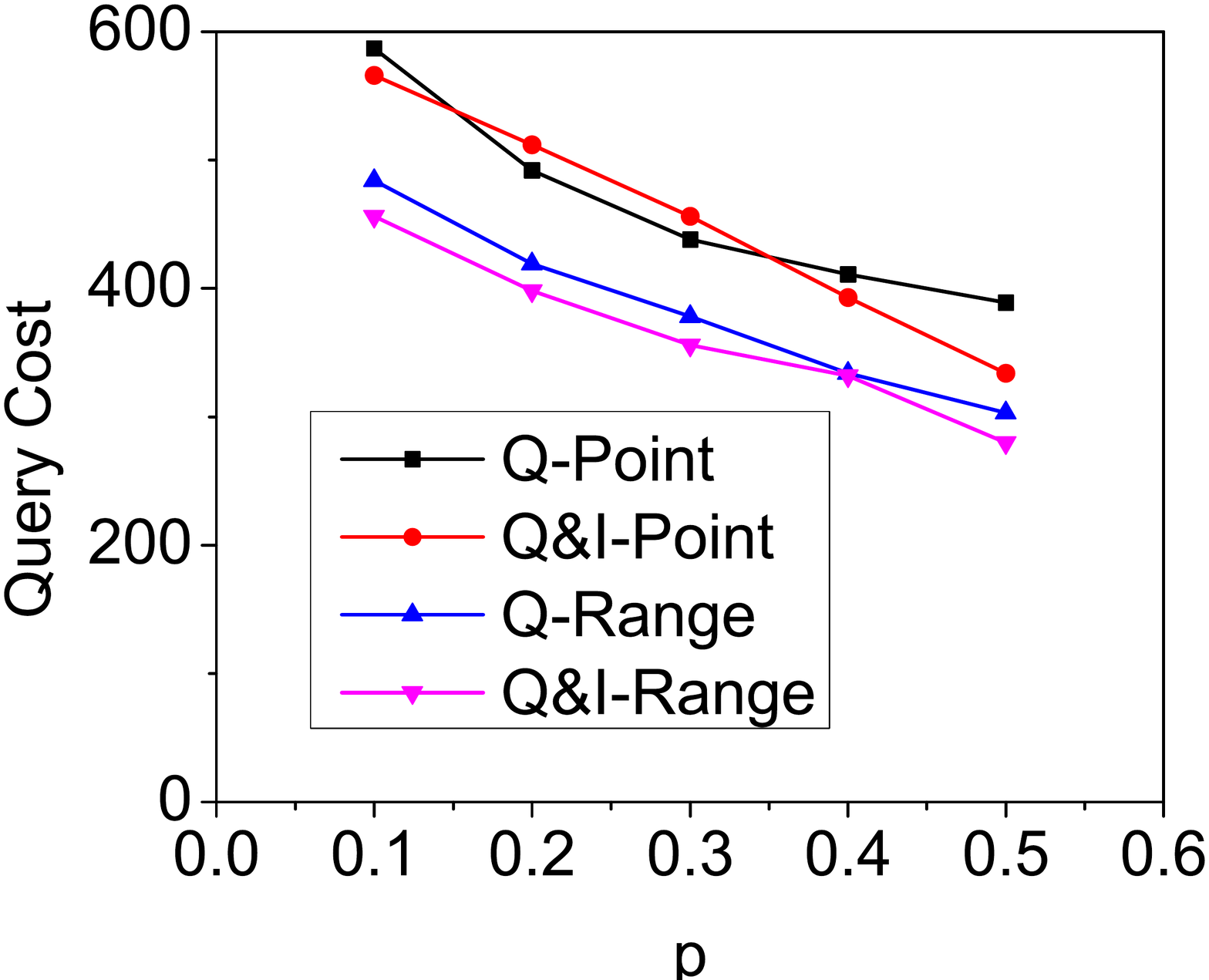}
\vspace{-7mm}\caption{Varying $p$ (BOOL)} 
\label{fig:bool_domainSizeVsQueryCost}
\end{minipage}
\hspace{1mm}
\begin{minipage}[t]{0.23\linewidth}
\centering
\includegraphics[width = 50mm, height = 32mm]{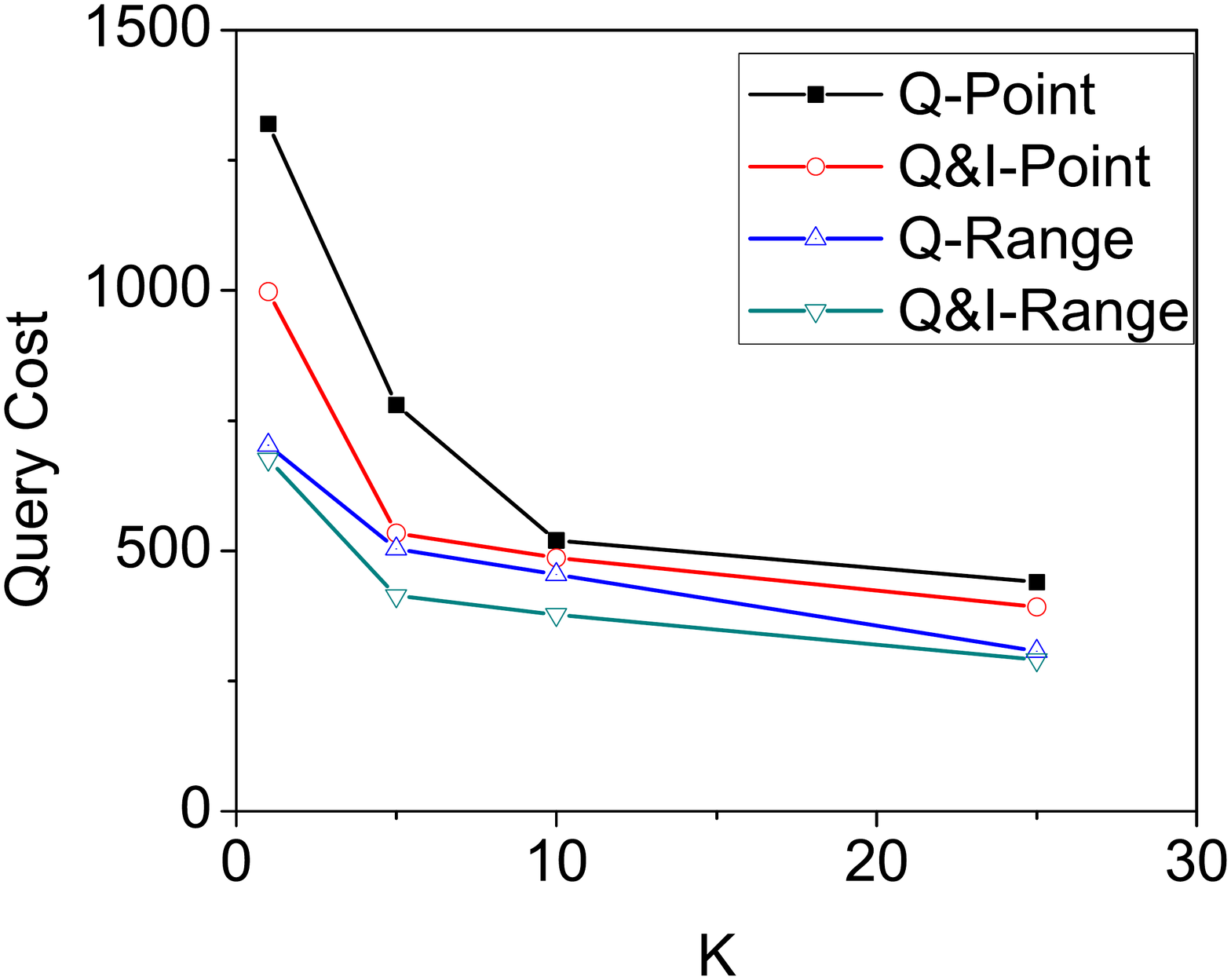}
\vspace{-7mm}\caption{Query Cost to Infer all Private Attributes (BOOL)}
\label{fig:bool_kVsQueryCostInferAll}
\end{minipage}
\hspace{-2mm}
\end{figure*}

\begin{figure*}[ht]
\begin{minipage}[t]{0.32\linewidth}
\centering
\includegraphics[width = 45mm, height = 32mm]{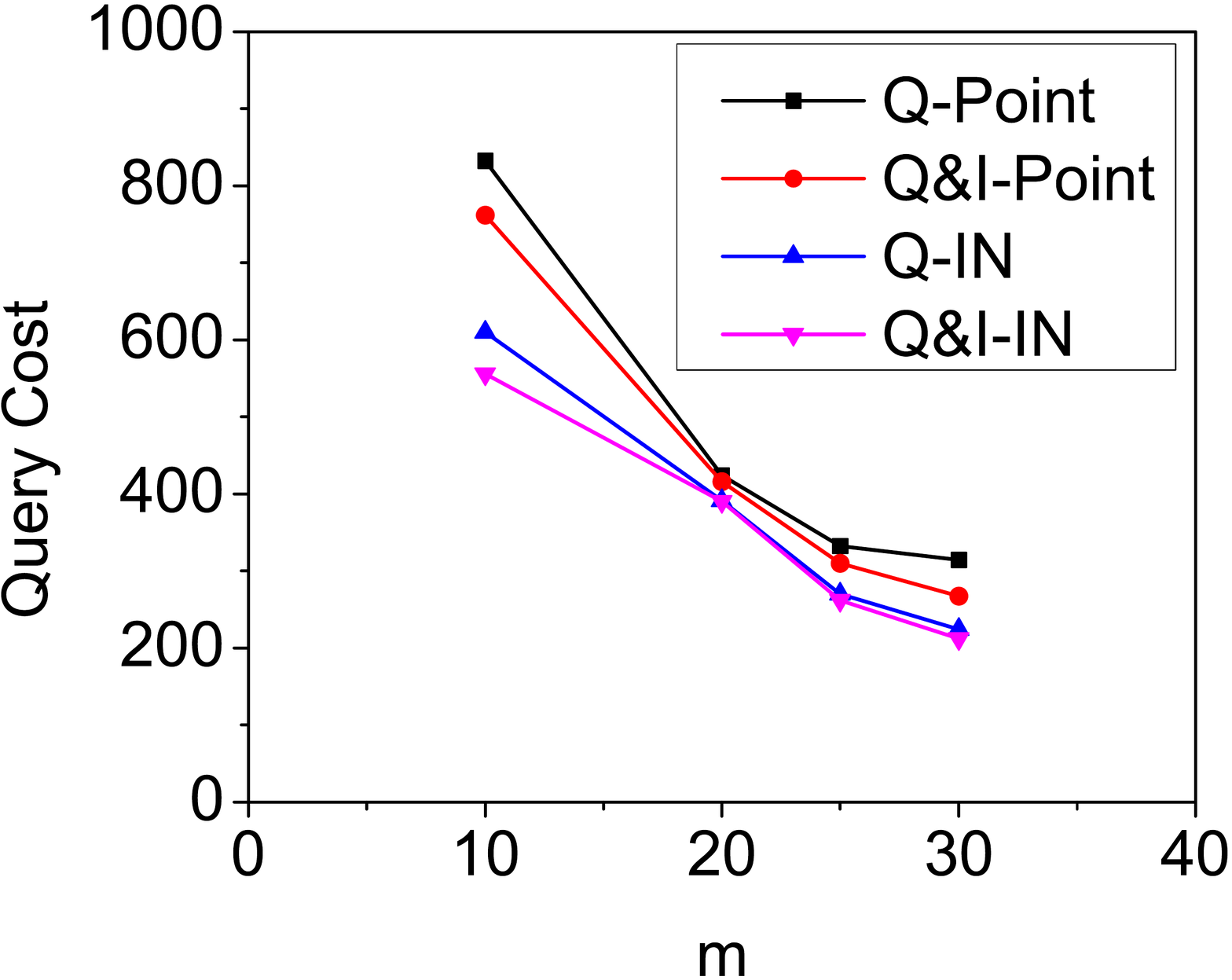}
\vspace{-3mm}\caption{Varying $m$ (Zipfian)}
\label{fig:zipf_mVsQueryCost}
\end{minipage}
\hspace{1mm}
\begin{minipage}[t]{0.32\linewidth}
\centering
\includegraphics[width = 45mm, height = 32mm]{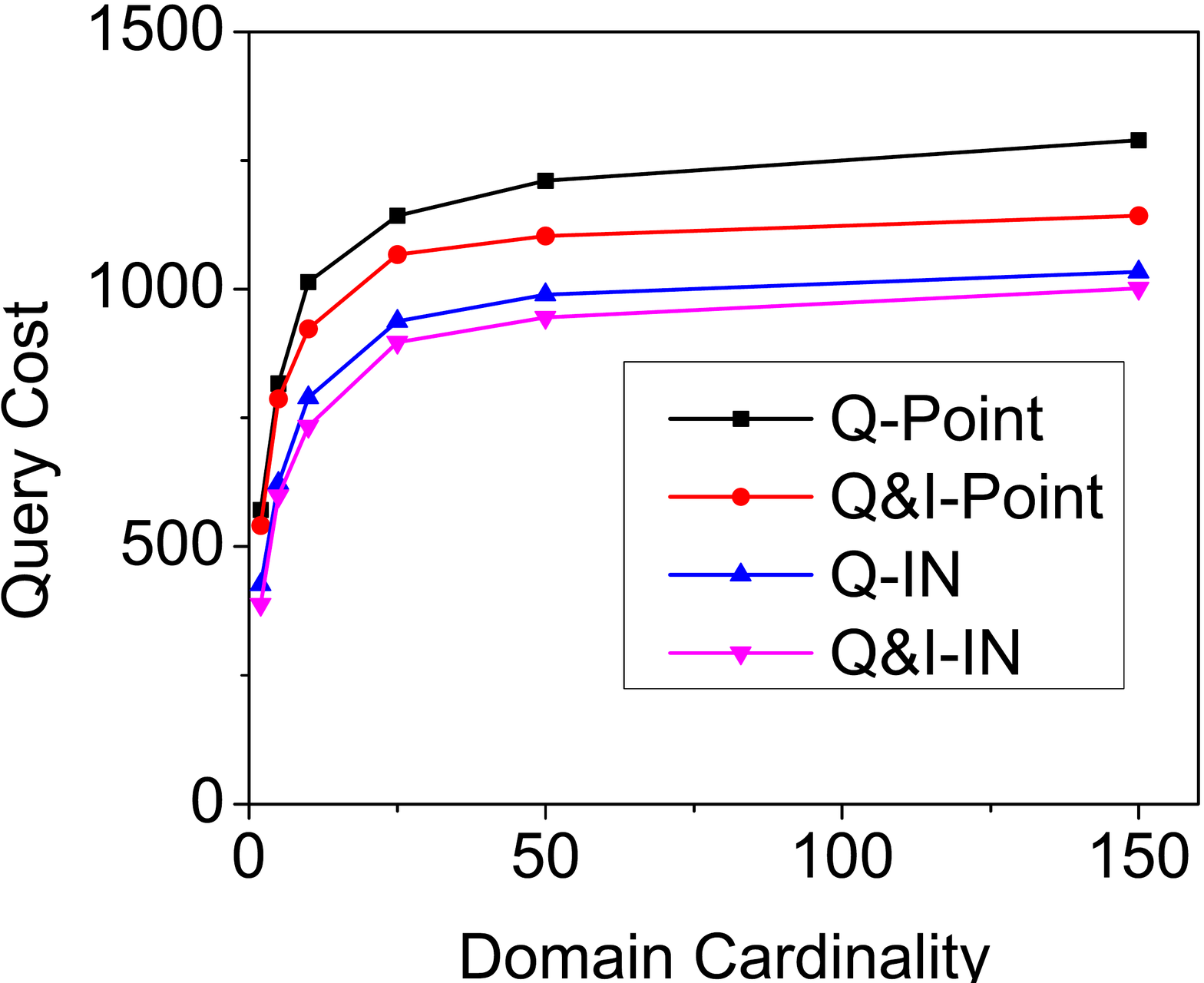}
\vspace{-3mm}\caption{Varying Domain Size of Inferred Attribute(Zipfian)}
\label{fig:zipf_domainSizeVsQueryCost}
\end{minipage}
\hspace{1mm}
\begin{minipage}[t]{0.32\linewidth}
\centering
\includegraphics[width = 45mm, height = 32mm]{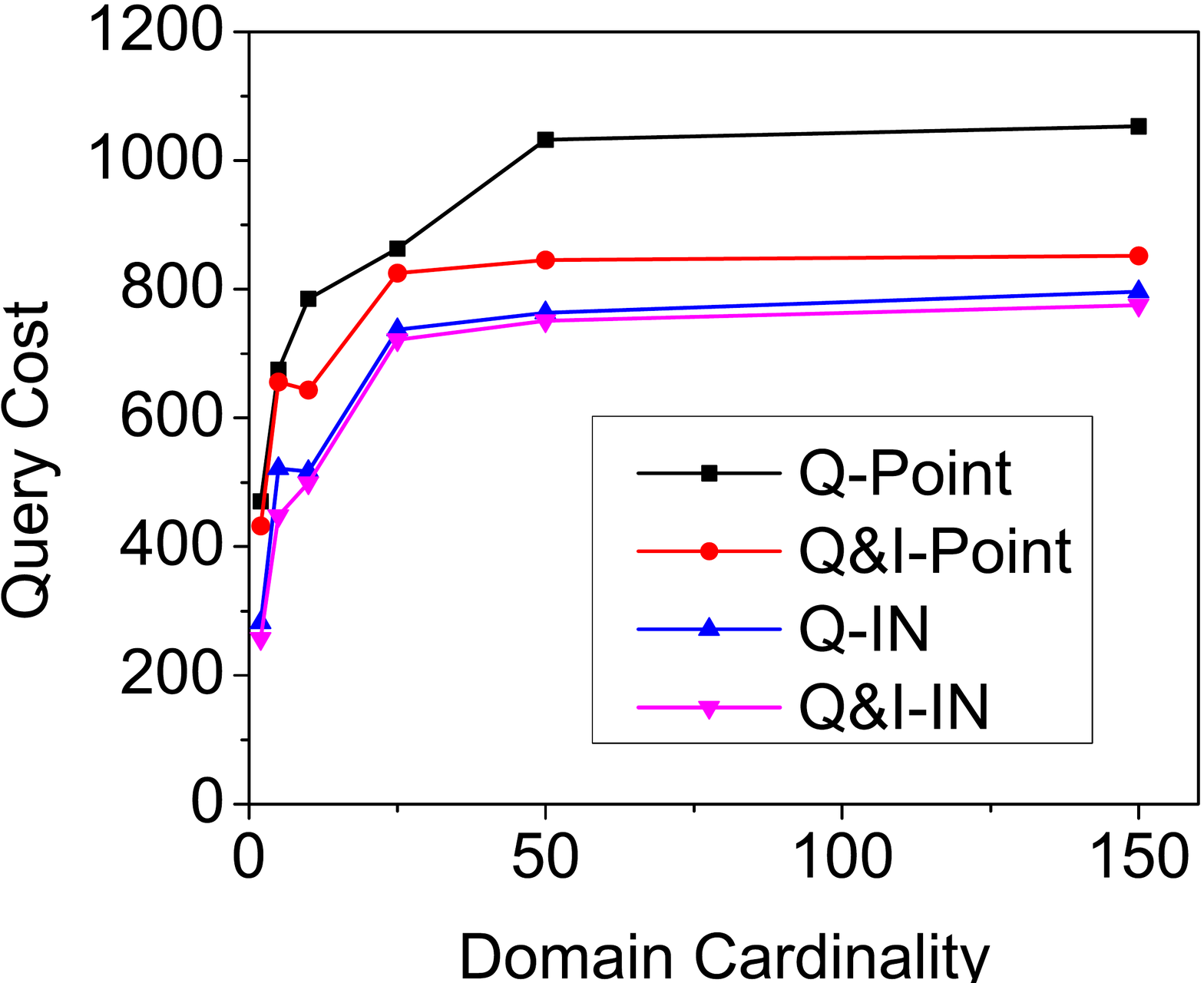}
\vspace{-3mm}\caption{Varying Domain Size of All Attributes (Zipfian)}
\label{fig:zipf_domainSizeAllVsQueryCost}
\end{minipage}
\end{figure*}

\subsection{Online Demonstration}
\label{subsec:expOnline}

In the online experiments, we sought to compromise private attributes of user profiles from Amazon Goodreads (GR), Catch22-Dating (CD) and Renren.com (RR) respectively. A detailed description of the procedure we used and its correctness can be found in \S\ref{sec:onlineAttacks}. Note that since we have no connection with these websites and thus do not have access to the ground truth, we limit the scope to a small-scale proof-of-concept.

\noindent {\em Renren:}
Renren (RR)\cite{renren} is a major Chinese social networking website (similar to Facebook) with more than 160 million users. The user profile consists of details like demographics, education and work affiliation. The website supports extensive privacy settings - allowing a user to specify any subset of profile attributes as public, private or only visible to friends.  In our experiments, we focused on one attribute \texttt{hometown province} with a domain size of 34 - which is set to private by all users we target. Renren has a search interface that accepts keyword queries on a user's public attributes such as profile name. It displays appropriate information depending on who issued the query - i.e. everyone can see public attributes, only friends can see attributes marked as visible to friends and none can see the private attributes. The results are ordered based on a ranking function that takes into account the entire profile regardless of privacy settings (as shown in \S\ref{sec:intro}). 

Renren enforces tuple insertion constraint and also allows NULL values. 
We conducted our attack using Q-Point algorithm.
One can see from Table~\ref{tbl:onlineExpSummary} that we were able to successfully infer the private attribute for 75 out of 76 profiles 
with an average query cost of 20 per profile.
We also conducted an experiment to measure the success rate of our attacks by varying $k$.
The search interface of Renren, has a large value of $k$ (ranging in hundreds).
In our experiment, we artificially truncated the results for different values of $k$ and verified if we can infer the private attribute.
Figure~\ref{fig:renrenKVsSuccRate} shows that for $k$ as little as $50$, we achieve a success rate of 92\%.

\vspace{1mm}
\noindent {\em Catch22Dating:} 
Catch22Dating (CD)\cite{catch22Dating} is an online dating website where users create profiles that are then matched to other users. The public attributes here capture the demographic information of a user, whereas the private attributes specify a user's matching preferences - e.g., the one private attribute we focus on is Boolean ``{\it Is it OK if your matches have been married before}'' (henceforth referred to as {\tt Married}). The search interface of Catch22Dating has an option called ``Both Perspectives'', which enables the ranking function to take into account both public and private attributes of all profiles on the website. It does enforce the tuple insertion constraint by requiring Student ID from selected universities during user registration. It also allows IN queries to be specified (e.g., one can set an attribute to be ``do not care'' in the query). Hence, we model the adversary as Q-only operating over an IN-interface.

The website allows NULL values on almost all attributes. As a result, our Q-IN attack might fail simply because the user specified NULL as the attribute value. One can see from Table~\ref{tbl:onlineExpSummary} that out of the 120 users we attacked, we compromised the private attribute {\tt Married} for 61 of them. For the other 60, either the user did not specify whether he/she would like to accept matches who have been married, or Q-IN attack fails on these users. The average query cost for the success and NULL/failure cases are 60 and 660, respectively, consistent with our prior discussions that failures generally consume many more queries than the successful cases.

\vspace{1mm}
\noindent{\em Amazon Goodreads} (GR)\cite{goodReads} a social cataloging site where the users can connect to each other and share their experience/opinions about books. 
The user profile consists of demographic information such as {\tt user name} which is always public, and attributes such as {\tt zipcode} which can be set as private. Regardless of a user's choice on location privacy, the ranking function used in the website's ``user search'' interface 
ranks each user according to its (geographic) distance from the location of the user performing the search. Goodreads allows free and instant account registration - i.e., there is no tuple insertion constraint - but no range query. Hence we use the Q\&I-Point algorithm.

We started with registering 10 fake accounts with randomly generated ZIP codes, and launched Q\&I-Point over it to verify the correctness of our algorithm. Then, to enable verification on real accounts, we identified 53 ``special'' users at Goodreads who have their ZIP code hidden but chose to reveal their city/state (in US). We launched Q\&I-Point successfully on all these users, and then verified that every ZIP code we compromised indeed belongs to its corresponding city/state revealed by the user. The average query cost, as shown in Table~\ref{tbl:onlineExpSummary}, is 455 per victim.

\begin{table}[h]
\centering
\caption{Summary of Online Experiments}
\label{tbl:onlineExpSummary}
\begin{tabular}{cp{0.5in}p{0.5in}p{0.5in}p{0.5in}p{0.5in}}
  \hline
        & \#Accounts Attacked & \#Success & Avg Cost (Success) & Avg Cost (Failure) \\ \hline
    CD  & 120                     & 61        & 60                      & 660               \\
    GR  & 53                      & 53        & 455                     & N/A               \\
    RR  & 76                      & 75        & 20                      & 34                \\ 
  \hline
\end{tabular}
\end{table}


\section{Additional Details for Online Experiments}
\label{sec:onlineAttacks} 
In this section, we provide some additional details for online experiments.
We first describe a practical attack where we infer the private attribute of a user in Catch22Dating website 
and provide a general approach followed by a formal argument as to its correctness.
We then provide the equivalent algorithm for Goodreads. 
The logic and correctness argument for Goodreads is similar.

\begin{figure*}[ht]
\begin{minipage}[h]{0.4\linewidth}
\centering
\includegraphics[scale=0.4]{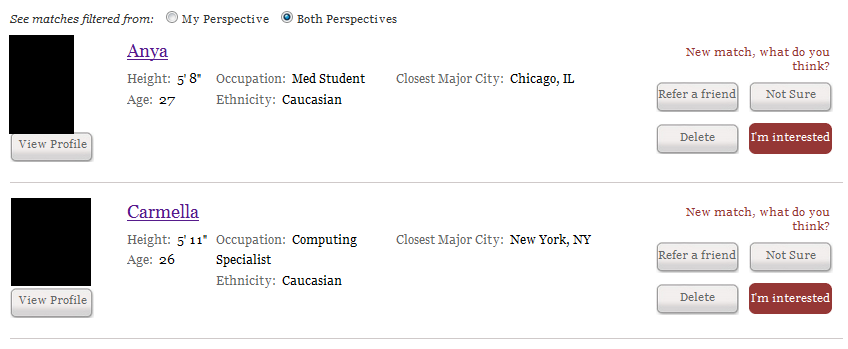}
\vspace{-7mm}\caption{Query $q_1$ where Anya is top ranked}
\label{fig:catch22Before}
\end{minipage}
\hspace{20mm}
\begin{minipage}[h]{0.4\linewidth}
\centering
\includegraphics[scale=0.4]{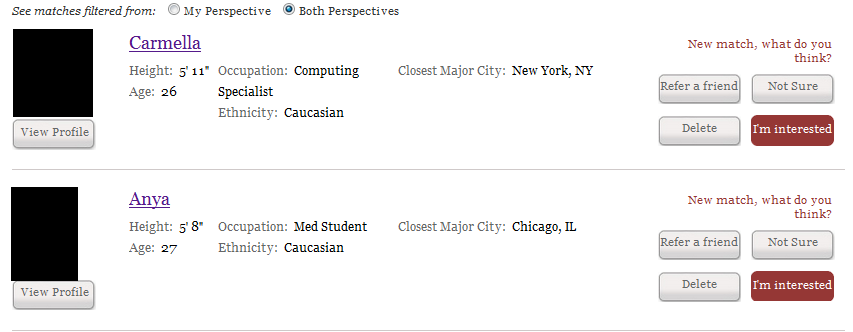}
\vspace{-7mm}\caption{Query $q_2$ where Anya is not top ranked}
\label{fig:catch22After}
\end{minipage}
\end{figure*}

\vspace{1mm}
\noindent {\bf Example Attack over Catch22Dating:}
{\em Catch22dating} (CD) is an online dating website with millions of users.
CD allows users to create profiles containing public (such as demographics) and 
private (such as matching preferences).
CD also has a search interface where users could specify a query (based on public information only) to search for other users.
CD uses a ranking function that matches the profile using {\em both} public and private information.
Suppose, we wish to infer a private information ({\em Is it ok if your matches have been married before}) of a user $v$ 
(with screen name {\tt Anya}).
We first created a fake user profile $u$ where we specified the marital status as `Never married'.
Under these circumstances, our results in {\tt Anya} as the best matching user. 
Figures~\ref{fig:catch22Before} shows the result.
We then change $u$'s profile to specify the marital status as `Previously Married'.
When we issue the {\em same} query (but for the modified profile), we can see that the rank of {\tt Anya} has dropped. 
We can now plausibly infer that {\tt Anya} has specified that she prefers her matches not to be married before.

\vspace{1mm}
\noindent {\bf Catch22Dating Inference:}
Using the notations from the technical sections, let $v$ be the victim tuple whose private attribute value $v[B_1]$ we seek to infer.
In the context of Catch22Dating, the private boolean attribute $B_1$ 
stores the user's response to question: {\em Is it ok if your matches have been married before}.
It can take two values - No or No Preference.
The public attribute most relevant to $B_1$ is $A_1$ which stores the user's response to the question:
{\em Have you married before}. It takes two values - Yes and No.
We construct a random point query by using the public attributes from $v$'s profile and 
chose the values for private attributes randomly.
However, we set the value for the attribute {\em Have you married before} to No.
If this randomly constructed query (say $q_1$) returned $v$, then we create an alternate query $q_2$. 
$q_2$ is identical to $q_1$ on all attributes except for the value of attribute $A_1$ - 
$q_1[A_1]$ = No (not married before) and $q_2[A_1]$ = Yes (had married before). 
Now if the rank of $v$ is lower in $q_2$ than in $q_1$ (i.e., $d_l(q_2, v) > d_l(q_1,v))$, 
the attacker can infer that the target profile $v$ has private attribute $B_1$ value set to No.

\vspace{1mm}
\noindent {\bf Correctness Argument:}
If the target profile $v$ has $B_1$ value set to No Preference, then $d_l(q_1,v) = d_l(q_2,v)$. 
This is because by setting $v[B_1]$ to No Preference, the target profile is accepting any value of $A_1$ in the search query. 
On the other hand if $v[B_1]$ = No then $d_l(q_1, v) < d_l(q_2, v)$. 
When the attacker issues a query $q_2$ followed by $q_1$, 
one of the three scenarios can arise: 
\begin{enumerate}
\itemsep0em 
\item rank of $v$ remains same as it was in $q_1$
\item rank of $v$ increases 
\item rank of $v$ decreases
\end{enumerate}

If $v[B_1]$ = No Preference, only (1) or (2) is possible. 
While scenario (1) is easy to understand, 
scenario (2) may appear if there exists a tuple $t$, 
such that $t[B_1]$ = No, $d_l(q_1, t) < d_l(q_1, v)$ and $d_l(q_2, t) > d_l(q_2, v)$. 
Scenario (3) is impossible when $v[B_1]$ = No Preference as it is not possible to find a 
tuple $t$ that has $d_l(q_1, v) < d_l(q_1, t)$ and $d(q_2,v) > d_l(q_2, t)$. 
So when the attacker finds that the rank of the target profile $v$ decreases after switching 
from $q_1$ to $q_2$, he/she can correctly infer that $v[B_1]$ = No,
because the only assignment $v[B_1]$ can have other than No Preference is No.

\vspace{1mm}
\begin{figure}[ht]
\centering
\includegraphics[scale=0.5]{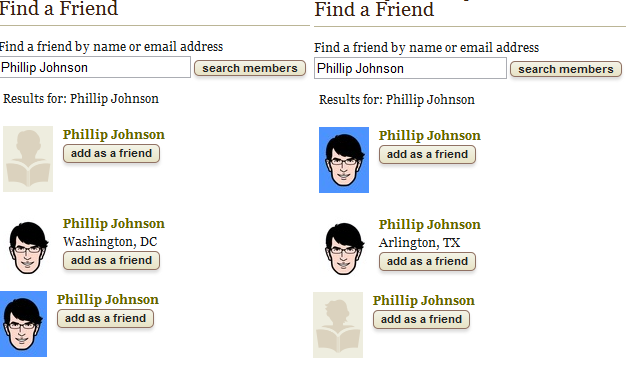}
\caption{Demonstration of an attack over Goodreads}
\label{fig:goodReadsBeforeAndAfter}
\end{figure}

\noindent {\bf Goodreads Inference:}
Goodreads has a single private attribute {\tt zipcode}.
The search interface to find other similar users allows only a single attribute - user name. 
When displaying the results of a search query 
it ranks the user profiles (who have the user name from the query) 
according to a proprietary distance function from the location of the user performing the search.
Based on our observations, Goodreads seems to use some proprietary variant of zipcode-zipcode distance function.
We used a publicly available distance function - but to address the uncertainty of Goodreads' ranking function
we added an error margin.

Our attack proceeds in two stages. 
We start with the set of all zipcodes in USA.
Since Goodreads allows an adversary to create multiple accounts, we create two accounts, say $a_1,a_2$.
We set the zipcode of $a_1, a_2$ to two different randomly chosen zipcodes.
We issue a search query based on victim $v$'s user name.
Suppose for $a_1$, $v$ has a higher rank than $a_2$ (which has the same name as $a_1$),
then remove all zipcodes that has distance higher than the 
distance between zipcodes of $a_1$ and $a_2$ (with an additional error margin) and vice versa.
This process is repeated till the zipcode list cannot be pruned anymore.
Let the set of all non-pruned zipcodes be $Z$.
In the second stage, we set the zipcode of $a_1$ to be a random zipcode from $Z$. 
We set the zipcode of $a_2$ to each value in $Z$ and search for $v$ till
$v$ has a higher rank than $a_1$.
We then use this information to narrow the zipcodes till we identify the user's zipcode.
The results of our experiments can be found in Figure~\ref{fig:goodReadsBeforeAndAfter}.

\section{Related Work}
\label{sec:relWork}

\noindent{\bf Database Ranking:} The area of ranking has been extensively studied 
in the context of deterministic\cite{ilyas2008survey,ho1997range}, probabilistic\cite{li2009unified} and incomplete\cite{green2006models} data. Processing top-$k$ query when the ranking score is a combination of scores of individual attributes was studied in \cite{fagin2003optimal, hristidis2001prefer}. A popular ranking function is nearest neighbor \cite{geng2008query} where the tuples are ordered based on the distance between tuple $t$ and the given query $q$. Other categorizations such as monotone, generic or no ranking (such as Skyline queries) has also been studied \cite{ilyas2008survey}. Recently, there have been studies on learning the rank of a tuple \cite{thirumuruganathan2013rank} or the ranking function design \cite{zheng2008general,podelski2004complete} through a top-$k$ static ranking interface.

\vspace{2mm}
\noindent{\bf Inference Control:} Prior work on privacy inference \cite{farkas2002inference} studied the problem of inferring individual tuple values \cite{chin1984efficient, chin1986security} and the existence of a tuple in a database \cite{nergiz2007hiding} from aggregates such as SUM, MIN, MAX, etc. The field of inference control\cite{farkas2002inference,adam1989security,domingo2008survey} seeks to prevent such attacks by through query auditing, controlling the number of tuples that match a query or modify query responses using perturbation, distortion etc\cite{chin1982auditing}. Researchers have also proposed multiple privacy preserving aggregate query processing techniques \cite{agrawal2000privacy,dwork2006calibrating}. Recently, \cite{Li:2014:YLB:2632951.2632953} has showed that it is possible to infer the location of a user in a Location based Social Network (LBSN) (which could be considered as a private attribute) if the ranking function returns the distance between the query and the victim tuple. However, we do not assume the availability of such information as most websites do not display the score of a tuple for a query. 


\section{Final Remarks}

In this paper, we identified a novel problem of rank-based inferencing over databases that use ranked retrieval model. We introduced a taxonomy of the problem space into four important subspaces based on varying interface designs and adversarial capabilities. For each problem subspace, we developed nontrivial attacking algorithms and conducted theoretical analysis of their feasibility and performance. We verified the effectiveness of the attacks using a comprehensive set of experiments on real-world datasets and online demonstrations on high-profile real-world websites.

It is our hope that the paper initiates a new topic of research on the privacy implications of database ranking; and future research will address the many open problems, e.g., how to design effective defensive strategies that thwart the rank-based inference of private attributes yet maintain the utility of ranking functions.

\bibliographystyle{abbrv}
\bibliography{rankInference}

\end{document}